\documentclass[fullpage,letterpage,11pt]{article}

\usepackage{amsmath,amsthm,amssymb}
\usepackage{mathtools}
\usepackage{algorithm}
\usepackage{algpseudocode}
\usepackage[T1]{fontenc}
\usepackage[capitalize]{cleveref}
\usepackage{thm-restate}
\usepackage{multirow}
\usepackage[margin=1in]{geometry}
\usepackage{todonotes}

\theoremstyle{definition}
\newtheorem{definition}{Definition}[section]
\theoremstyle{plain}
\newtheorem{theorem}[definition]{Theorem}
\newtheorem{lemma}[definition]{Lemma}
\newtheorem{proposition}[definition]{Proposition}

\newtheorem{observation}[definition]{Observation}

\mathchardef\standardl=\mathcode`l
\mathcode`l=\ell
\newcommand{\deactivatel}{\mathcode`l=\standardl}
\makeatletter
\edef\operator@font{\operator@font\noexpand\deactivatel}
\makeatother

\newcommand{\N}{\mathbb{N}}
\newcommand{\eps}{\varepsilon}

\begin{document}

\title{Counting and Sampling Labeled Chordal Graphs in Polynomial Time
\thanks{\'Ur.~H. and Da.~L. were supported by NSF grant CCF-2008838. Er.~V. was supported by NSF grant CCF-2147094.}
}

\author{
\'Ursula H\'ebert-Johnson\thanks{University of California, Santa Barbara, USA, \texttt{ursula@ucsb.edu}.}
\and
Daniel Lokshtanov\thanks{University of California, Santa Barbara, USA, \texttt{daniello@ucsb.edu}.}
\and
Eric Vigoda\thanks{University of California, Santa Barbara, USA, \texttt{vigoda@ucsb.edu}.}
}

\date{}

\begin{titlepage}
\def\thepage{}
\thispagestyle{empty}
\maketitle

\begin{abstract}
We present the first polynomial-time algorithm to exactly compute the number of labeled chordal graphs on $n$ vertices. Our algorithm solves a more general problem: given $n$ and $\omega$ as input, it computes the number of $\omega$-colorable labeled chordal graphs on $n$ vertices, using $O(n^7)$ arithmetic operations. A standard sampling-to-counting reduction then yields a polynomial-time exact sampler that generates an $\omega$-colorable labeled chordal graph on $n$ vertices uniformly at random. Our counting algorithm improves upon the previous best result by Wormald (1985), which computes the number of labeled chordal graphs on $n$ vertices in time exponential in $n$. An implementation of the polynomial-time counting algorithm gives the number of labeled chordal graphs on up to $30$ vertices in less than three minutes on a standard desktop computer. Previously, the number of labeled chordal graphs was only known for graphs on up to $15$ vertices.

In addition, we design two approximation algorithms: (1) an approximate counting algorithm that computes a $(1\pm\eps)$-approximation of the number of $n$-vertex labeled chordal graphs, and (2) an approximate sampling algorithm that generates a random labeled chordal graph according to a distribution whose total variation distance from the uniform distribution is at most $\eps$. The approximate counting algorithm runs in $O(n^3\log{n}\log^7(1/\eps))$ time, and the approximate sampling algorithm runs in $O(n^3\log{n}\log^7(1/\eps))$ expected time.
\end{abstract}
\end{titlepage}

\section{Introduction}
\label{sec:intro}

Generating random graphs from a prescribed graph family is a fundamental task for running simulations and testing conjectures. Although generating a random labeled graph on $n$ vertices is easy (just flip an unbiased coin for each potential edge), the first polynomial-time algorithm for generating an \emph{unlabeled} graph uniformly at random was only given in 1987, by Wormald~\cite{wormald1987generating}. The algorithm of Wormald runs in polynomial time in expectation, and to the best of our knowledge, the existence of a worst-case polynomial-time sampler of random unlabeled graphs remains open.

Naturally, when we wish to sample from a specified graph family, there are many interesting families of graphs for which this problem is nontrivial, even when the graphs are labeled. For the class of labeled trees, a sampling algorithm using Pr{\"u}fer sequences~\cite{prufer1918neuer} was discovered in 1918. More recently, a fast (exact) uniform sampler was presented by Gao and Wormald for $d$-regular graphs with $d=o(\sqrt{n})$ in 2017~\cite{GW-regular}, and then for power-law graphs in 2018~\cite{GW-power}. A more general problem is the following: given an arbitrary degree sequence, generate a random graph with those specified degrees --- this has been resolved for bipartite graphs~\cite{JSV,BBV} as well as for general graphs when the maximum degree is not too large~\cite{GS,AGW}. See Greenhill~\cite{Greenhill} for a survey of random generation of graphs with degree constraints. For planar graphs, Bodirsky, Gr\"{o}pl, and Kang presented a polynomial-time algorithm~\cite{BGK}, which uses dynamic programming to exactly compute the number of labeled planar graphs on $n$ vertices and generate a planar graph uniformly at random in time $\widetilde{O}(n^7)$. This was improved to $O(n^2)$ expected time by Fusy~\cite{Fusy}, using a Boltzmann sampler.

Our results fall naturally within this line of work. We consider the problem of generating a labeled {\em chordal} graph on $n$ vertices uniformly at random. A graph is chordal if it has no induced cycles of length at least 4. Despite being one of the most fundamental and well-studied graph classes, prior to our work, the fastest uniform sampling algorithm for labeled chordal graphs was the exponential-time algorithm of Wormald from 1985~\cite{Wormald}. (To be precise, optimizing the running time of an algorithm for counting chordal graphs was not the main focus of Wormald; rather, the main goal of the paper was to determine the asymptotic number of chordal graphs with given connectivity, and the exponential-time algorithm is a corollary of these results.) Since then, various algorithmic approaches have been proposed for generating chordal graphs (e.g., \cite{MVA,SHET,sun2020sampling,ESS,MR}), but these algorithms do not come with any formal guarantees about their output distribution. In particular,~\cite{SHET} specifically asks for the existence of a polynomial-time algorithm to sample chordal graphs uniformly at random as an open problem. In a recent abstract, Sun and Bez\'akov\'a~\cite{sun2020sampling} proposed a Markov chain for sampling chordal graphs, but this Markov chain comes with few mixing time guarantees.

We obtain the first polynomial (in $n$) time algorithm to exactly count the number of labeled chordal graphs on $n$ vertices, as well as the first polynomial-time uniform sampling algorithm for the class of labeled chordal graphs. Our algorithm also easily extends to counting and sampling \mbox{$\omega$-{\em colorable}} labeled chordal graphs. A graph $G$ is $\omega$-colorable if there exists a function $c\colon V(G)\to\{1,\ldots,\omega\}$ such that every edge $uv\in E(G)$ satisfies $c(u)\neq c(v)$.

\begin{theorem}
\label{thm:main}
There is a deterministic algorithm that given positive integers $n$ and $\omega \leq n$, computes the number of \mbox{$\omega$-colorable} labeled chordal graphs on $n$ vertices using $O(n^7)$ arithmetic operations. Moreover, there is a randomized algorithm that given the same input, generates a graph uniformly at random from the set of all $\omega$-colorable labeled chordal graphs on $n$ vertices using $O(n^7)$ arithmetic operations.
\end{theorem}

By the known equivalence between chromatic number, maximum clique size, and treewidth of chordal graphs~\cite{blair1993introduction}, \cref{thm:main} can be reinterpreted as counting and sampling labeled chordal graphs of clique size at most $\omega$, or treewidth at most $\omega-1$. The running time bound of \cref{thm:main} is stated in terms of the number of arithmetic operations. Since there are at most $2^{n^2}$ labeled graphs on $n$ vertices, the arithmetic operations need to deal with $n^2$-bit integers. Therefore, using the $O(n\log{n})$-time algorithm for multiplying two $n$-bit integers~\cite{harvey2021integer} yields an $O(n^9 \log n)$-time upper bound for our algorithm in the RAM model.

A straightforward implementation of our counting algorithm gives the number of labeled chordal graphs on up to $n = 30$ vertices in less than three minutes on a standard desktop computer. Previously, the number of labeled chordal graphs was only known for graphs on up to $15$ vertices \cite{oeis}. In addition, we use our implementation to compute the number of $\omega$-colorable labeled chordal graphs for $n \leq 12$ and $\omega \leq 12$. We chose to stop at $n = 12$ to keep the table at a reasonable size, not because of the computation time. We present these computational results in \cref{sec:implementation}.

In addition, we design two approximation algorithms: (1) an approximate counting algorithm that computes a $(1\pm\eps)$-approximation of the number of $n$-vertex labeled chordal graphs, and (2) an approximate sampling algorithm that generates a random labeled chordal graph according to a distribution whose total variation distance from the uniform distribution is at most $\eps$. The approximate counting algorithm runs in $O(n^3\log{n}\log^7(1/\eps))$ time, and the approximate sampling algorithm runs in $O(n^3\log{n}\log^7(1/\eps))$ expected time.

\subsection{A Brief Survey on Chordal Graphs}

The literature on chordal graphs is so vast that it would be impossible to fully do it justice. Discussions of chordal graphs in the literature go as far back as 1958~\cite{hajnal1958auflosung}. What follows is a summary of some of the most notable problems and milestones.

Many NP-hard optimization problems (such as {\em coloring}~\cite{Gavril} and {\em maximum independent set}~\cite{Farber}),
as well as \#P-hard counting problems (such as {\em independent sets}~\cite{OUU,BS}), and many others~\cite{GCWebpage}, are polynomial-time solvable on chordal graphs. Chordal graphs have a wide variety of applications, including phylogeny in evolutionary biology~\cite{Gusfield,PDSD} and Bayesian networks in machine learning~\cite{WBL}. When doing Gaussian elimination on a symmetric matrix, the set of matrix entries that are nonzero for at least one time-point of the elimination process corresponds to the edge set of a chordal graph. Thus the problem of finding an ordering in which to do Gaussian elimination that minimizes the number of nonzero matrix entries can be reduced to finding a chordal supergraph of a given graph with the minimum number of edges~\cite{rose1972graph}. This problem, known as {\em minimum fill-in}, was shown to be NP-complete by Yannakakis~\cite{yannakakis1981computing}. Chordal graphs play a central role in graph theory~\cite{brandstadt1999graph}, both through their connection to treewidth~\cite{heggernes2006minimal} and through their connection to perfect graphs~\cite{golumbic2004algorithmic}. From an algorithms perspective, chordal graphs can be recognized in linear time~\cite{rose1976algorithmic}.

An interesting and relevant result by Bender et al.~\cite{bender1985almost} is that a random $n$-vertex labeled chordal graph is a split graph with probability $1-o(1)$, i.e., the fraction of labeled chordal graphs that are not split is $o(1)$. This suggests the possibility of a simple {\em approximately uniform} sampler for labeled chordal graphs: simply sampling a random labeled split graph leads to an output distribution with total variation distance $o(1)$ from the uniform distribution on labeled chordal graphs. This is the main inspiration behind our approximate sampling algorithm. In addition, this idea suggests two things: The first is that it might be possible to find a simple and efficient \emph{uniform} random sampler for labeled chordal graphs. The second is that the type of chordal graphs that one usually envisions when thinking of a chordal graph (namely those with relatively small treewidth) are different from those most likely to be generated by a uniform random sampler (namely split graphs). Therefore, to generate the type of chordal graphs that one usually envisions, one should be sampling not from the set of all chordal graphs but rather from a subset, e.g., the set of all $\omega$-colorable chordal graphs. Fortunately, \cref{thm:main} provides this functionality.

\subsection{Methods}

Our exact counting algorithm is based on dynamic programming. While clique trees and tree decompositions are never explicitly mentioned in the description of the algorithm, the intuition behind the algorithm is based on these notions. Essentially, we generate a rooted clique tree where the dynamic-programming table encodes certain properties of the graph, including how it relates to the root node of the clique tree. A {\em clique tree} of a graph $G$ is a tree $T$ together with a function $f$ that maps each vertex of $G$ to a connected vertex subset of $T$, such that for every pair $u$,$v$ of vertices in $G$, $uv$ is an edge in $G$ if and only if $f(u) \cap f(v) \neq \emptyset$. It is well known that a graph $G$ has a clique tree if and only if it is chordal~\cite{blair1993introduction}.

The main difficulty with this approach is that different chordal graphs have different numbers of clique trees, so if we count the total number of clique trees, this will not give us an accurate count of the number of $n$-vertex chordal graphs. Therefore, the key idea behind our algorithm is to assign to each labeled chordal graph a unique ``canonical'' clique tree and to only count these canonical clique trees. The information stored in the dynamic-programming table is sufficient to ensure that every clique tree that we do count is the canonical tree of some chordal graph, and that the canonical clique tree of every chordal graph is counted. As it turns out, the best way to phrase our algorithm is not in terms of clique trees at all, but rather in terms of an (essentially) equivalent notion that we call an ``evaporation sequence.'' An evaporation sequence is closely related to the standard notion of a perfect elimination ordering (PEO) of a chordal graph. A {\em simplicial} vertex is a vertex whose neighborhood is a clique, and a PEO is an ordering of the vertices such that each is simplicial in the current induced subgraph, if we delete the vertices in that order (see \cref{sec:evaporation} for more details). An evaporation sequence is a type of ``canonical'' PEO: at each step, we remove {\em all} simplicial vertices from the current subgraph, rather than making an arbitrary choice of a single simplicial vertex. We say that all of the simplicial vertices evaporate at time $1$. Next, all of the vertices that become simplicial once the first set of simplicial vertices has been removed are said to evaporate at time $2$, and so on. It is easy to see that every labeled chordal graph has a unique evaporation sequence, and that this sequence does not depend on the labeling of the vertices.

Therefore, the number of chordal graphs on $n$ vertices is the sum over all evaporation sequences of the number of labeled chordal graphs with that evaporation sequence. While different evaporation sequences correspond to different numbers of chordal graphs, because this number is independent of the labels, we can simply guess the {\em number} $x$ of vertices that evaporate at any given time, and then without loss of generality assign the labels $1, \ldots, x$ to those vertices.

In our dynamic-programming algorithm, the recursive subproblems deal with counting {\em rooted} clique trees. In this context, the root of a clique tree is the set of vertices that evaporate last. Just as we would like to ``force'' a set of nodes to be in the root of the tree, we will sometimes want to force a set of nodes to evaporate last. This is done using what we call an {\em exception set}, i.e., a set of vertices that do not evaporate even if they are simplicial.

The random sampling algorithm in \cref{thm:main} follows from our counting algorithm using the standard sampling-to-counting reduction of~\cite{JVV}.

Our approximate counting algorithm is based on the result of Bender et al.~\cite{bender1985almost} mentioned above, which states that almost every labeled chordal graph is a split graph. When $n$ is large enough (compared to a function of $\frac{1}{\eps}$), simply counting (even approximately) the number of labeled split graphs leads to a $(1\pm\eps)$-approximation of the number of labeled chordal graphs. Similarly, sampling a random labeled split graph gives an approximately uniform sampler of labeled chordal graphs.

\subsection{Overview of the Paper}

Our dynamic-programming algorithm, including the associated recurrences, is presented in \cref{sec:counting}. The proof of correctness of the counting algorithm can be found in \cref{sec:main_proof} (here we reduce to counting connected chordal graphs) and \cref{sec:conn_counting_proof} (here we establish the recurrences). In \cref{sec:implementation}, one can find the tables of numbers produced by our implementation of the counting algorithm. In \cref{sec:sampling}, we describe the details of how the sampling algorithm follows from the counting algorithm. Lastly, in \cref{sec:approx}, we describe our approximate counting and approximate sampling algorithms.


\section{Preliminaries}
\label{sec:preliminaries}

All of our algorithms deal with vertex-labeled chordal graphs. For simplicity of notation, we assume the vertex set of each graph is a subset of $\N = \{1,2,3,\ldots\}$, which allows the labels to also serve as the names of the vertices. For example, we will speak of the vertex $5\in V(G)$ rather than a vertex $v$ with label~5.

\begin{definition}
A \emph{labeled graph} is a pair $G = (V,E)$, where the vertex set $V$ is a finite subset of $\N$ and the edge set $E$ is a set of two-element subsets of $V$.
\end{definition}

Henceforth, we implicitly assume \emph{all graphs that we consider are labeled graphs}. For nonnegative integers $n$, we use the notation $[n]\coloneqq\{1,2,\ldots,n\}$. Intervals of integers will often appear in our algorithm as the vertex set of a graph or subgraph, so we also define $$[a,b]\coloneqq\{a,a+1,\ldots,b\}$$ for nonnegative integers $a,b$. If $b<a$, then $[a,b] = \emptyset$.

\begin{definition}
Let $A = \{a_1,\ldots,a_r\}$ and $B = \{b_1,\ldots,b_r\}$ be finite subsets of $\N$ such that $|A| = |B|$, where the elements $a_i$ and $b_i$ are listed in increasing order. We define $\phi(A,B)\colon A\to B$ as the bijection that maps $a_i$ to $b_i$ for all $i\in[r]$.
\end{definition}

\begin{definition}
Let $G_1$, $G_2$ be two graphs, and suppose $C\coloneqq V(G_1)\cap V(G_2)$ is a clique in both $G_1$ and $G_2$. When we say we \emph{glue $G_1$ and $G_2$ together at $C$} to obtain $G$, this indicates that $G$ is the union of $G_1$ and $G_2$: the vertex set is $V(G) = V(G_1)\cup V(G_2)$, and the edge set is $E(G) = E(G_1)\cup E(G_2)$.
\end{definition}

For a graph $G$ and a vertex subset $S\subseteq V(G)$, $G[S]$ denotes the induced subgraph on the vertices of $S$. For a vertex $v\in V(G)$, we denote the neighborhood of $v$ in $G$ by $N_G(v)$, or $N(v)$ if the graph is clear from the context. For $S\subseteq V(G)$, the open neighborhood of $S$ is denoted by $$N_G(S)
\coloneqq\{v\in V(G)\setminus S:uv\in E(G)\text{ for some $u\in S$}\}$$ and the closed neighborhood of $S$ is denoted by $N_G[S]
\coloneqq S\cup N_G(S)$, or simply $N(S)$ and $N[S]$, respectively. For $S,T\subseteq V(G)$, we say $S$ \emph{sees all of} $T$ if $T\subseteq N(S)$.

\section{Counting labeled chordal graphs}
\label{sec:counting}

\subsection{The evaporation sequence}
\label{sec:evaporation}

\begin{definition}
A vertex $v$ in a graph $G$ is \emph{simplicial} if $N(v)$ is a clique.
\end{definition}

A \emph{perfect elimination ordering} of a graph $G$ is an ordering $v_1,\ldots,v_n$ of the vertices of $G$ such that for all $i\in[n]$, $v_i$ is simplicial in the subgraph induced by the vertices $v_i,\ldots,v_n$. It is well known that a graph is chordal if and only if it has a perfect elimination ordering~\cite{blair1993introduction}. For our counting algorithm, we define the notion of the evaporation sequence of a chordal graph, which can be viewed as a canonical version of the perfect elimination ordering. In the evaporation sequence, rather than making an arbitrary choice for which of the simplicial vertices in $G$ will go first in the ordering, we place the set of all simplicial vertices as the first item in the sequence. As an example, if $G$ is a tree, then the set of all simplicial vertices would be exactly the leaves of $G$.

To build the evaporation sequence, we use the fact that every chordal graph contains a simplicial vertex \cite{blair1993introduction}. This means that given a chordal graph $G$, if we repeatedly remove all simplicial vertices, then eventually no vertices remain. By observing at which time step each vertex is deleted, we obtain a partition of $V(G)$, which allows us to classify and understand the structure of $G$. In our algorithm, it will also be useful to set aside a set of exceptional vertices which are never deleted, even if they are simplical.

To formalize this, suppose we are given a chordal graph $G$ and a clique $X\subseteq V(G)$. We define the \emph{evaporation sequence} of $G$ with \emph{exception set} $X$ as follows: If $X = V(G)$, then the evaporation sequence of $G$ is the empty sequence. If $X\subsetneq V(G)$, then let $\widetilde L_1$ be the set of all simplicial vertices in $G$, and let $L_1 = \widetilde L_1\setminus X$. Suppose $L_2,\ldots,L_t$ is the evaporation sequence of $G\setminus L_1$ (with exception set $X$). Then $L_1,L_2,\ldots,L_t$ is the evaporation sequence of $G$.

To see that this is well-defined, we need to check that $\widetilde L_1\setminus X$ is nonempty whenever $X\subsetneq V(G)$, to ensure that we eventually reach the base case $X = V(G)$. This follows from the fact that every chordal graph that is not a complete graph contains two non-adjacent simplicial vertices \cite{blair1993introduction}. Suppose $\widetilde L_1\setminus X = \emptyset$. This means $\widetilde L_1$ is a clique since $X$ is a clique, so $G$ does not contain any non-adjacent simplicial vertices. Thus $G$ must be a complete graph, in which case $\widetilde L_1\subseteq X$ implies $X = V(G)$. Therefore, we always reach the base case $X = V(G)$.

If the evaporation sequence $L_1,L_2,\ldots,L_t$ of $G$ has length $t$, then we say $G$ \emph{evaporates} at time $t$ with \emph{exception set} $X$, and $t$ is called the \emph{evaporation time}. We define $L_G(X)\coloneqq L_t$ to be the last set in the evaporation sequence of $G$, and we let $L_G(X) = \emptyset$ if the sequence is empty. Similarly, we define the evaporation time of a vertex subset. Suppose $G$ has evaporation sequence $L_1,L_2,\ldots,L_t$ with exception set $X$, and suppose $S\subseteq V(G)\setminus X$ is a nonempty vertex subset. Let $t_S$ be the largest index $i$ such that $L_i\cap S\ne\emptyset$. We say $S$ \emph{evaporates} at time $t_S$ in $G$ with exception set $X$.

\subsection{Setup for the counting algorithm}

Given positive integers $n$ and $\omega$, we wish to count the number of $\omega$-colorable chordal graphs on $n$ vertices. This can easily be reduced to the problem of counting $\emph{connected}$ $\omega$-colorable chordal graphs on at most $n$ vertices (see \cref{lemma:disconn} in \cref{sec:main_proof}). Therefore, our main focus is to describe the following algorithm:

\begin{theorem}
\label{thm:conn_counting}
There is an algorithm that given $n\in\N$, computes the number of $\omega$-colorable labeled connected chordal graphs with vertex set $[n]$ using $O(n^7)$ arithmetic operations.
\end{theorem}

We first give an overview of this algorithm and describe the various dynamic-programming tables (\cref{defn:counters}). Next, we describe the recurrences in detail in \cref{sec:recurrences}. In \cref{sec:main_proof}, we show that the counting portion of \cref{thm:main} (counting chordal graphs) follows from \cref{thm:conn_counting} (counting connected chordal graphs). In \cref{sec:conn_counting_proof}, we prove correctness of the recurrences and complete the proof of \cref{thm:conn_counting}.

\medskip
\textbf{Algorithm overview.} To count $\omega$-colorable connected chordal graphs $G$, we classify these graphs based on the behavior of their evaporation sequence. We make use of several \emph{counter functions} (these are our dynamic-programming tables), each of which keeps track of the number of chordal graphs in a particular subclass. The arguments of the counter functions tell us the number of vertices in the graph, the evaporation time, the size of the exception set $X$, the size of the last set of simplicial vertices $L_G(X)$, etc. Initially, we consider all possibilities for the evaporation time of $G$ with exception set $X = \emptyset$. Then, using several of our recursive formulas, we reduce the number of vertices by dividing up the graph into smaller subgraphs and counting the number of possibilities for each subgraph. As we do so, the exception set $X$ increases in size. When we consider these various subgraphs, we also make sure that in each subgraph, the maximum clique size is at most $\omega$. In the end, the algorithm understands the possibilities for the entire graph, including the cliques that make up the very first set in its evaporation sequence.

The purpose of the exception set is to allow us to restrict to smaller subgraphs without distorting the evaporation behavior of the graph. For example, suppose we wish to count the number of connected chordal graphs on $n$ vertices that evaporate at time $t$, such that the vertices $1,2,\ldots,l$ make up the last set to evaporate, i.e., $L_G(\emptyset) = [l]$. Let $L = [l]$. One subproblem of interest would be to count the number of possibilities for the first connected component of $G\setminus L$. Formally, we count the number of possibilities for $G'\coloneqq G[L\cup C]$, where $C$ is the connected component of $G\setminus L$ that contains the vertex $l+1$. For each possible number of vertices in $G'$, we make a recursive call to count the number of possible subgraphs $G'$ of that size. However, if we were to restrict to $G'$ with a still-empty exception set, then the evaporation time of $G'$ alone could be much less than the evaporation time of $V(G')$ in $G$. Indeed, there may be vertices in $G\setminus G'$ adjacent to $L$ that prevent $L$ from evaporating before time $t$, so when we restrict to the subgraph $G'$, $L$ would now evaporate too soon. This would cause a cascading effect, causing vertices near $L$ to evaporate as well, and changing the entire evaporation sequence of $G'$. To resolve this, we add the vertices of $L$ to the exception set to preserve the evaporation behavior of $G'$.

The list of counter functions is given in \cref{defn:counters}; see \cref{fig:counter_functions} for an illustration. As shown below, the number of $\omega$-colorable connected chordal graphs on $n$ vertices is the sum of various calls to the fourth function $\tilde g_1$, since $\tilde g_1(t,0,n)$ is the number of $\omega$-colorable connected chordal graphs on $n$ vertices that evaporate at time $t$ with empty exception set.

To remember the names of these counter functions, one can think of them as follows. The $g$-functions keep track of the size of the exception set $X$, but these do not have information about the size of $L_G(X)$. The $f$-functions have an additional argument $l$, which is the size of $L_G(X)$. As a mnemonic, one can say that $g$ stands for ``glued'' and $f$ stands for ``free.'' In the $g$-functions, $X$ is the ``root,'' and all of the vertices of $X$ are glued, in the sense that they cannot evaporate. In the $f$-functions, $X\cup L_G(X)$ is the ``root,'' and some of the vertices in the root are free, since the vertices in $L_G(X)$ are allowed to evaporate. For functions with a tilde, the connected components of $G\setminus X$ (for $g$-functions) or $G\setminus(X\cup L(G,X))$ (for $f$-functions) all evaporate at the same time. Lastly, a subscript $p$ means that the neighborhoods of the connected components of $G\setminus X$ (or $G\setminus(X\cup L(G,X))$), when intersected with $X$ (or $X\cup L(G,X)$), are \emph{proper} subsets of $X$ (or $X\cup L(G,X)$). For all of these functions, we only consider graphs that are $\omega$-colorable.

\begin{definition}
\label{defn:counters}
The following functions count various subclasses of chordal graphs. The arguments $t,x,l,k,z$ are nonnegative integers. They take on values up to at most $n$ and satisfy the domain requirements listed below.
\begin{enumerate}
    \item $g(t,x,k,z)$ is the number of $\omega$-colorable connected chordal graphs $G$ with vertex set $[x+k]$ that evaporate in time at most $t$ with exception set $X\coloneqq[x]$, where $X$ is a clique, with the following property: every connected component of $G\setminus X$ (if any) has at least one neighbor in $X\setminus[z]$. \textbf{Domain:} $t\ge 0$, $x\ge 1$, $z<x$.
    \item $\tilde g(t,x,k,z)$ is the same as $g(t,x,k,z)$, except every connected component of $G\setminus X$ (if any) evaporates at time \emph{exactly} $t$ in $G$. \emph{Note: A graph with $V(G) = X$ would be counted because in that case, $\tilde g$ is the same as $g$.}
    \textbf{Domain:} $t\ge 1$, $x\ge 1$, $z<x$.
    \item $\tilde g_p(t,x,k,z)$ is the same as $\tilde g(t,x,k,z)$, except no connected component of $G\setminus X$ sees all of~$X$. \textbf{Domain:} $t\ge 1$, $x\ge 1$, $z<x$.
    \item $\tilde g_1(t,x,k)$ and $\tilde g_{\ge 2}(t,x,k)$ are the same as $\tilde g(t,x,k,z)$, except every connected component of $G\setminus X$ sees all of $X$ (hence we no longer require every component of $G\setminus X$ to have a neighbor in $X\setminus[z]$), and furthermore, for $\tilde g_1$ we require that $G\setminus X$ has exactly one connected component, and for $\tilde g_{\ge 2}$ we require that $G\setminus X$ has at least two components. \textbf{Domain for $\tilde g_1$:} $t\ge 1$, $x\ge 0$. \textbf{Domain for $\tilde g_{\ge 2}$:} $t\ge 1$, $x\ge 1$.
    \item $f(t,x,l,k)$ is the number of $\omega$-colorable connected chordal graphs $G$ with vertex set $[x+l+k]$ that evaporate at time \emph{exactly} $t$ with exception set $X\coloneqq[x]$, such that $G\setminus X$ is connected, $L_G(X) = [x+1,x+l]$, and $X\cup L_G(X)$ is a clique. \textbf{Domain:} $t\ge 1$, $x\ge 0$, $l\ge 1$.
    \item $\tilde f(t,x,l,k)$ is the same as $f(t,x,l,k)$, except every connected component of $G\setminus(X\cup L_G(X))$ evaporates at time \emph{exactly} $t-1$ in $G$, and there exists at least one such component, i.e., $X\cup L_G(X)\subsetneq V(G)$.
    \textbf{Domain:} $t\ge 2$, $x\ge 0$, $l\ge 1$.
    \item $\tilde f_p(t,x,l,k)$ is the same as $\tilde f(t,x,l,k)$, except no connected component of $G\setminus(X\cup L_G(X))$ sees all of $X\cup L_G(X)$. \textbf{Domain:} $t\ge 2$, $x\ge 0$, $l\ge 1$.
    \item $\tilde f_p(t,x,l,k,z)$ is the same as $\tilde f_p(t,x,l,k)$, except rather than requiring that $G\setminus X$ is connected, we require that $G\setminus[z]$ is connected. \textbf{Domain:} $t\ge 2$, $x\ge 0$, $l\ge 1$, $z\le x$.
\end{enumerate}
\end{definition}

\begin{figure}
    \centering
    \includegraphics[width=0.7\textwidth]{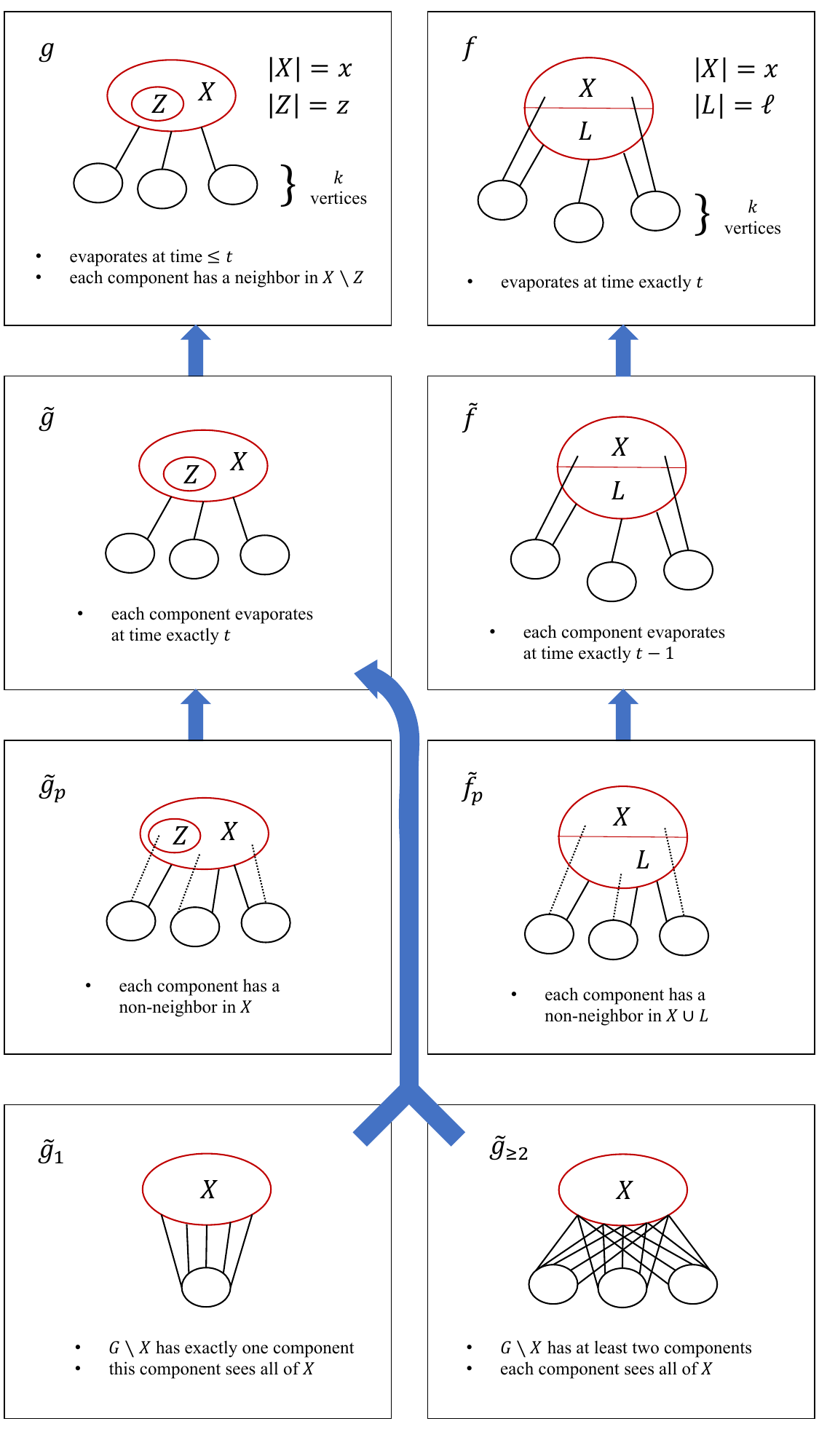}
    \caption{The counter functions. Here $L = L_G(X)$ and $Z=[z]$. An arrow from one function to another, say from $\tilde g$ to $g$, indicates that the definition of $\tilde g$ is the same as that of $g$, except where indicated otherwise. The drawing of $\tilde f_p$ represents $\tilde f_p(t,x,l,k)$. The function $\tilde f_p(t,x,l,k,z)$ is similar but also keeps track of the argument $z$.}
    \label{fig:counter_functions}
\end{figure}

\subsection{Recurrences for the counting algorithm}
\label{sec:recurrences}


We implicitly assume \emph{all graphs in this section are connected and $\omega$-colorable}. For $k\in\N$, let $c(k)$ denote the number of $\omega$-colorable connected chordal graphs with vertex set $[k]$. To compute $c(n)$, we first consider all possibilities for the evaporation time. Initially, the exception set is empty. We observe that $\tilde g_1(t,0,n)$ is the number of (connected, \mbox{$\omega$-colorable}) chordal graphs with vertex set $[n]$ that evaporate at time exactly $t$ with empty exception set. Therefore, $$c(n) = \sum_{t=1}^n\tilde g_1(t,0,n).$$ We compute the necessary values of $\tilde g_1$ by evaluating the following recurrences top-down using memoization. We take this approach rather than bottom-up dynamic programming to simplify the description slightly, since by memoizing we do not need to specify in what order the entries of the various dynamic-programming tables are computed. We simply compute each value of the counter functions as needed. For the following recurrences, let $X = [x]$ according to the current value of the argument $x$, and let $L = L_G(X)$.

To compute $\tilde g_1(t,0,n)$, the number of chordal graphs that evaporate at time exactly $t$, we consider all possibilities for the size $l$ of $L$. Once $l$ is given, there are $\binom{n}{l}$ possibilities for the label set of $L$. Recall that $f$ counts the number of chordal graphs where $L$ is fixed and evaporates at time $t$, and $G\setminus X$ is connected. Formally, we have the following recurrence --- the first time this is used, the arguments are $t$, $x = 0$, and $k = n$.

\begin{restatable}{lemma}{gTildeOne}
\label{lemma:g1_tilde}
For $\tilde g_1$, we have $$\tilde g_1(t,x,k) = \sum_{l=1}^k\binom{k}{l}f(t,x,l,k-l).$$
\end{restatable}

The proof of \cref{lemma:g1_tilde}, along with the proofs of all of the following recurrences, can be found in \cref{sec:recurrence_proofs}. We proceed in this order, first presenting the recurrences and later presenting their proofs, to first give the reader the context of the entire algorithm. For now, to see the intuition behind \cref{lemma:g1_tilde}, recall that in the definition of $f(t,x,l,k)$ we require a specific label set for $L_G(X)$, namely $[x+1,x+l]$. If we were to replace that requirement with $L_G(X) = L'$ for any other subset $L'$ of $[x+1,x+l+k]$ of size $l$, this would not change the value of $f(t,x,l,k)$. Therefore, in the recurrence for $\tilde g_1$, it is sufficient to compute $f(t,x,l,k-l)$ and multiply by $\binom{k}{l}$, rather than computing $\binom{k}{l}$ distinct counter functions.

For $f$, to count chordal graphs where $L$ is fixed and evaporates at time $t$, and $G\setminus X$ is connected, we consider all possibilities for the set of labels that appear in connected components of $G\setminus(X\cup L)$ that evaporate at time exactly $t-1$. Recall that $\tilde f$ is the same as $f$, except all components of $G\setminus(X\cup L)$ evaporate at time exactly $t-1$. For each possible $k'$ (the size of the chosen label set), $\tilde f$ allows us to count the number of possibilities for the subgraph $G_1$ consisting of $X\cup L$ and all components of $G\setminus(X\cup L)$ that evaporate at time exactly $t-1$. The function $g$ allows us to count the number of possibilities for the subgraph $G_2$ consisting of $X\cup L$ and all other components.

\begin{restatable}{lemma}{f}
For $f$, we have $$f(t,x,l,k) = \sum_{k'=1}^k\binom{k}{k'}\tilde f(t,x,l,k')g(t-2,x+l,k-k',x).$$
\end{restatable}

When $f$ is called for the first time in the initial steps of the algorithm, this is the first moment when $X$ becomes nonempty (when we call the function $g$), since at this point we are restricting to a subgraph with fewer than $n$ vertices. When we restrict to the subgraph $G_2$, we want to ensure that its vertices have the same evaporation behavior as they did in $G$. In particular, we need to ensure that the vertices of $L$ do not evaporate too soon, since their presence may be essential for preventing other vertices from evaporating. For this reason, we let $X\cup L$ be the exception set for $G_2$. For $G_1$, the exception set is simply $X$ because the components that evaporate at time exactly $t-1$ are still present in $G_1$, preventing the vertices of $L$ from evaporating before time $t$.

For the subgraph $G_2$, now that all of $L$ has been pushed into the exception set, we no longer have information about the argument $l$ since the last set of simplicial vertices in $G_2$ evaporates further back in time. This is why we call $g$ rather than $f$ to count the possibilities for $G_2$. In fact, $G_2\setminus[x+l]$ might not even be connected, which is required by $f$. Finally, the fourth argument of $g$ indicates that every connected component of $G_2\setminus(X\cup L)$ has at least one neighbor in $L$. This ensures that $G\setminus X$ is connected.

For $g$, to count chordal graphs that evaporate in time at most $t$, we consider all possibilities for the set of labels that appear in connected components of $G\setminus X$ that evaporate at time exactly $t$. Recall that $\tilde g$ counts the number of chordal graphs where all connected components of $G\setminus X$ evaporate at time exactly $t$.

\begin{restatable}{lemma}{g}
\label{lemma:g}
For $g$, we have $$g(t,x,k,z) = \sum_{k'=0}^k\binom{k}{k'}\tilde g(t,x,k',z)g(t-1,x,k-k',z).$$
\end{restatable}

For $\tilde g$, to count chordal graphs where all connected components of $G\setminus X$ evaporate at time exactly $t$, we consider all possibilities for the label set of the component $C$ of $G\setminus X$ that contains the lowest label not in $X$ (namely $x+1$). We also consider all possibilities for $N(C)$. In this recurrence, $k'$ stands for $|C|$, and $x'$ stands for $|N(C)|$.

\begin{restatable}{lemma}{gTilde}
\label{lemma:g_tilde}
For $\tilde g$, we have $$\tilde g(t,x,k,z) = \sum_{k'=1}^k\sum_{x'=1}^x\left(\binom{x}{x'}-\binom{z}{x'}\right)\binom{k-1}{k'-1}\tilde g_1(t,x',k')\tilde g(t,x,k-k',z).$$
\end{restatable}

The constraint $x'\ge 1$ ensures that $G$ is connected. We subtract all ways of selecting $x'$ elements from $[z]$ to ensure that $N(C)$ is not contained in $[z]$. We also subtract 1 in the binomial coefficient $\binom{k-1}{k'-1}$ because the label set for $C$ always contains $x+1$, along with $k'-1$ other labels.

For $\tilde f$, we need to count chordal graphs where $L$ is fixed and evaporates at time $t$, $G\setminus X$ is connected, and all connected components of $G\setminus(X\cup L)$ evaporate at time exactly $t-1$. The number of such graphs in which no component of $G\setminus(X\cup L)$ sees all of $X\cup L$ is $\tilde f_p(t,x,l,k)$, by the definition of $\tilde f_p(t,x,l,k)$. Now if there is at least one all-seeing component, then we break this down into two further cases: either exactly one component sees all of $X\cup L$, or at least two components see all of $X\cup L$. Recall that $\tilde g_1$ (resp.\ $\tilde g_{\ge 2}$) counts the number of chordal graphs where all components of $G\setminus X$ evaporate at time exactly $t$, every component sees all of $X$, and there is exactly one such component (resp.\ at least two such components). In the first (resp.\ second) case, $\tilde g_1$ (resp.\ $\tilde g_{\ge 2}$) corresponds to the all-seeing component(s), and $\tilde f_p$ (resp.\ $\tilde g_p$) corresponds to the remaining components.

\begin{restatable}{lemma}{fTilde}
For $\tilde f$, we have
\begin{align*}
\tilde f(t,x,l,k) &= \tilde f_p(t,x,l,k)+\sum_{k'=1}^k\binom{k}{k'}\tilde g_1(t-1,x+l,k')\tilde f_p(t,x,l,k-k') \\
&\quad+\sum_{k'=1}^k\binom{k}{k'}\tilde g_{\geq 2}(t-1,x+l,k')\tilde g_p(t-1,x+l,k-k',x).
\end{align*}
\end{restatable}

The above cases are relevant because if at least two components of $G\setminus(X\cup L)$ see all of $X\cup L$, then this prevents the vertices of $L$ from evaporating before time $t$. Indeed, each vertex $u\in L$ has a neighbor in each of those two components, meaning $u$ has two non-adjacent neighbors. Otherwise, if there is at most one such component, then the neighborhoods of the remaining components of $G\setminus(X\cup L)$ must together cover $L$ to ensure that $L$ does not evaporate until time $t$. In that case, for each vertex $u\in L$ that is covered by a proper-subset neighborhood $N(C)$ of a component $C$, $u$~has a neighbor $v\in C$ as well as a neighbor $w\in(X\cup L)\setminus N(C)$, and $v$ and $w$ are non-adjacent.

The reason we require $G\setminus X$ to be connected in the definition of $f$ (rather than just requiring $G$ to be connected) can be seen from the recurrence for $\tilde f$. Since in the first sum over $k'$ we only wish to consider graphs with exactly one all-seeing component, in the definition of $\tilde g_1$ we require $G\setminus X$ to be connected. The recurrence for $\tilde g_1$ depends on $f$, so this carries over into requiring $G\setminus X$ to be connected in the definition of $f$. This explains the need for the argument $z$ (for example, in $g$): as was mentioned above when we discussed the recurrence for $f$, keeping track of $z$ lets us ensure that $G\setminus X$ is connected in all graphs counted by $f$.

For $\tilde g_{\ge 2}$, to count chordal graphs where all connected components of $G\setminus X$ evaporate at time exactly $t$, every component sees all of $X$, and there are at least two such components, we consider all possibilities for the label set of the component that contains the lowest label not in $X$. For the remaining components, there is either exactly one of them or at least two.

\begin{restatable}{lemma}{gTildeTwo}
For $\tilde g_{\ge 2}$, we have $$\tilde g_{\ge 2}(t,x,k) = \sum_{k'=1}^{k-1}\binom{k-1}{k'-1}\tilde g_1(t,x,k')\Big(\tilde g_1(t,x,k-k')+\tilde g_{\geq 2}(t,x,k-k')\Big).$$
\end{restatable}

For $\tilde g_p$, to count chordal graphs where all connected components of $G\setminus X$ evaporate at time exactly $t$ and no component sees all of $X$, we proceed as we did for $\tilde g$, except we require $x'<x$ rather than $x'\le x$.

\begin{restatable}{lemma}{gTildeP}
For $\tilde g_p$, we have $$\tilde g_p(t,x,k,z) = \sum_{k'=1}^k\sum_{x'=1}^{x-1}\left(\binom{x}{x'}-\binom{z}{x'}\right)\binom{k-1}{k'-1}\tilde g_1(t,x',k')\tilde g_p(t,x,k-k',z).$$
\end{restatable}

For $\tilde f_p$, we need to count chordal graphs where $L$ is fixed and evaporates at time $t$, all connected components of $G\setminus(X\cup L)$ evaporate at time exactly $t-1$, and no component sees all of $X\cup L$. We first observe that when $z = x$, requiring $G\setminus[z]$ to be connected is the same as requiring $G\setminus X$ to be connected.

\begin{restatable}{lemma}{fTildeP}
We have $\tilde f_p(t,x,l,k) = \tilde f_p(t,x,l,k,x)$.
\end{restatable}

The next recurrence for $\tilde f_p$ counts the number of such graphs in which $G\setminus[z]$ is connected. On the first reading, one can skip the two ``otherwise'' cases in \cref{lemma:f_tilde_p}. In this lemma, we consider all possibilities for the label set of the component $C$ of $G\setminus(X\cup L)$ that contains the lowest label not in $X\cup L$. In the outer sum, $k'$ stands for $|C|$. Additionally, we consider all possibilities for the size $x'$ of $N(C)\cap X$ and the size $l'$ of $N(C)\cap L$, and we consider all possibilities for their respective label sets. If $l'>0$, then $N(C)$ is automatically not contained in $[z]$ since $z\le x$, so there are $\binom{x}{x'}$ possible label sets for $N(C)\cap X$.

The intuition behind the two ``otherwise'' cases is as follows. If $l' = 0$, then we must subtract $\binom{z}{x'}$ from the number of possible label sets for $N(C)\cap X$ to ensure that $N(C)\not\subseteq[z]$. If $l' = l$, then all of the vertices of $L$ have now been pushed into the exception set, so the evaporation time of the subgraph formed from the remaining components is $t-1$. In this case, we call $\tilde g_p$ since we no longer know the size of the last set of simplicial vertices.

The dot in front of each of the curly braces denotes multiplication. For example, if $l'>0$, then we multiply by $\binom{x}{x'}$.



\begin{restatable}{lemma}{fTildePWithZ}
\label{lemma:f_tilde_p}
For $\tilde f_p(t,x,l,k,z)$, we have
\begin{alignat*}{3}
& && \tilde f_p(t,x,l,k,z) = && \\
& && \sum_{k'=1}^k\sum_{\substack{0\le x'\le x \\ 0\le l'\le l \\ 0<x'+l'<x+l}}\binom{k-1}{k'-1}\binom{l}{l'}\tilde g_1(t-1,x'+l',k') && \cdot
\begin{cases}
\binom{x}{x'} & \text{ if $l'>0$}  \\
\binom{x}{x'}-\binom{z}{x'} & \text{ otherwise}
\end{cases} \\
& && && \cdot
\begin{cases}
\tilde f_p(t,x+l',l-l',k-k',z) & \text{if $l'<l$} \\
\tilde g_p(t-1,x+l,k-k',z) & \text{otherwise}.
\end{cases}
\end{alignat*}
\end{restatable}

The base cases are as follows. We reach the base case for $g$ when $t = 0$:
\[ g(0,x,k,z) = \begin{cases}
      1 & \text{ if $k = 0$}  \\
      0 & \text{ if $k>0$}.
   \end{cases}
\]
For $\tilde g$ and $\tilde g_p$, we have $\tilde g(t,x,0,z) = 1$ and $\tilde g_p(t,x,0,z) = 1$ when $k = 0$. For, $\tilde g_1$ we observe that $\tilde g_1(t,x,k) = 0$ if $t = 0$ or $k = 0$. Similarly, for $\tilde g_{\ge 2}$ we have $\tilde g_{\ge 2}(t,x,k) = 0$ if $t = 0$ or $k = 0$. We reach the base case for $f$ when $x+l>\omega$, $t = 1$, or $k = 0$. If $x+l>\omega$, then $f(t,x,l,k) = 0$. Remarkably, this is the only place where $\omega$ appears in the algorithm. If $x+l\le\omega$, then we have
\[ f(1,x,l,k) = \begin{cases}
      1 & \text{ if $k = 0$}  \\
      0 & \text{ otherwise}.
   \end{cases}
\]
If $x+l\le\omega$ and $t\ge 2$, then $f(t,x,l,0) = 0$. For $\tilde f$, we have $\tilde f(t,x,l,k) = 0$ if $t = 1$ or $k = 0$. Similarly, for $\tilde f_p$ we have $\tilde f_p(t,x,l,k,z) = 0$ if $t = 1$ or $k = 0$. For the version of $\tilde f_p$ without the fifth argument $z$, we do not need a base case since we always immediately call $\tilde f_p$ with $z$.

The control flow formed by these recurrences is shown in \cref{fig:control_flow}. The algorithm terminates because either the value of $t$ or the number of vertices in the graph (i.e., $x+k$ or $x+l+k$) decreases each time we return to the same function. For the running time, note that we first compute all of the necessary binomial coefficients. In particular, we compute all values of $\binom{a}{b}$ such that $0\le a\le n$ and $0\le b\le a$, which can be done using $O(n^2)$ arithmetic operations. Next, when computing the counter functions, the running time is dominated by the arithmetic operations needed to compute $\tilde f_p$. The recurrence for $\tilde f_p$ involves a triple summation, and there are five arguments, so a naive implementation uses $O(n^8)$ arithmetic operations. However, in \cref{sec:wrap_up_proof}, we show that the running time can in fact be improved to $O(n^7)$ arithmetic operations.

\begin{figure}
    \centering
    \includegraphics[width=0.6\textwidth]{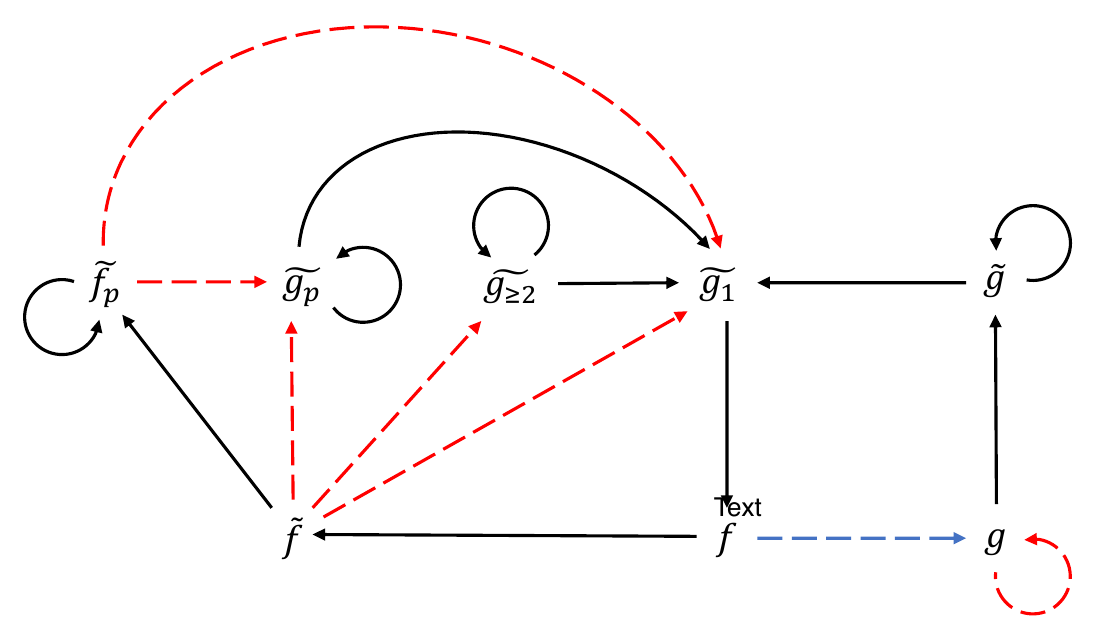}
    \caption{The control flow of the counting algorithm. An arrow from $f$ to $g$ indicates that the recursive formula for $f$ depends on $g$. The arrow is red (and dashed) if $t$ decreases by 1, blue (and dashed) if $t$ decreases by 2, and black if $t$ does not change. The black self-loops do not cause an infinite loop because in those recursive calls, the number of vertices decreases.}
    \label{fig:control_flow}
\end{figure}

\subsection{Proof of \cref{thm:main} (counting)}
\label{sec:main_proof}

In this section, we prove the counting portion of \cref{thm:main} using \cref{thm:conn_counting}. (See \cref{sec:sampling} for the proof of the sampling portion of \cref{thm:main}.) In other words, we describe an algorithm to count chordal graphs, assuming we have an algorithm to count connected chordal graphs. \cref{thm:conn_counting} --- counting connected chordal graphs --- is proved in \cref{sec:conn_counting_proof}.

For $k\in\N$, let $a(k)$ denote the number of $\omega$-colorable chordal graphs with vertex set $[k]$. Recall that $c(k)$ is the number of $\omega$-colorable connected chordal graphs with vertex set $[k]$.

\begin{lemma}
\label{lemma:disconn}
The number of $\omega$-colorable chordal graphs with vertex set $[n]$ is given by $$a(n) = \sum_{k=1}^n\binom{n-1}{k-1}c(k)a(n-k)$$ for all $n\in\N$.
\end{lemma}

\begin{proof}
Suppose $G$ is an $\omega$-colorable chordal graph with vertex set $[n]$. Let $G_1$ be the graph formed by the connected component of $G$ that contains the label 1, and let $G_2$ be the graph formed by all other connected components of $G$ (which can potentially be empty). Let $C$ be the set of labels that appear in $G_1$, let $k = |C|$, and let $D = [n]\setminus C$, i.e., $D$ is the set of labels that appear in $G_2$. Now relabel $G_1$ by applying $\phi(C,[k])$ to the labels in $C$, and relabel $G_2$ by applying $\phi(D,[n-k])$ to the labels in $D$. (Recall that $\phi(A,B)$ is defined in \cref{sec:preliminaries}.) We can see that $G_1$ is now a connected $\omega$-colorable chordal graph with vertex set $[k]$, and $G_2$ is an $\omega$-colorable chordal graph with vertex set $[n-k]$. The map that takes any chordal graph $G$ to the resulting pair $(G_1,G_2)$ is injective since $\phi(C,[k])$ and $\phi(D,[n-k])$ are both bijections. Therefore, $a(n)$ is at most the number of possible triples $(G_1,G_2,C)$, which is given by the summation above.

To see that $a(n)$ is bounded below by the same summation, suppose we are given $1\le k\le n$, a connected $\omega$-colorable chordal graph $G_1$ with vertex set $[k]$, an $\omega$-colorable chordal graph $G_2$ with vertex set $[n-k]$, and a subset $C\subseteq[n]$ of size $k$ that contains 1. Let $D = [n]\setminus C$. We construct an $\omega$-colorable chordal graph $G$ with vertex set $[n]$ as follows: Relabel $G_1$ by applying $\phi([k],C)$ to its label set, and relabel $G_2$ by applying $\phi([n-k],D)$ to its label set. Now let $G$ be the union of $G_1$ and $G_2$ (by taking the union of the vertex sets and the edge sets). Clearly $G$ is an $\omega$-colorable chordal graph with vertex set $[n]$. The map that takes the triple $(G_1,G_2,C)$ to the resulting graph $G$ is injective, so $a(n)$ is at least the summation above, as desired.
\end{proof}

\cref{lemma:disconn} directly gives a dynamic-programming algorithm to compute the number of $\omega$-colorable chordal graphs with vertex set $[n]$, given $n$ as input. First, we compute all of the necessary binomial coefficients, i.e., all values of $\binom{a}{b}$ such that $0\le a\le n$ and $0\le b\le a$, which can be done using $O(n^2)$ arithmetic operations. Next, by \cref{thm:conn_counting}, we can compute $c(k)$ for all $k\in[n]$ at a cost of $O(n^7)$ arithmetic operations. (Technically, the algorithm of \cref{thm:conn_counting} just computes $c(n)$, but by reusing the same dynamic-programming tables from that algorithm, we can also compute $c(k)$ for all $k\in[n]$.) To compute $a(n)$, we use the recurrence in \cref{lemma:disconn} at a cost of $O(n^2)$ arithmetic operations. For the base case, we observe that $a(0) = 1$. Therefore, we have an algorithm to count $\omega$-colorable chordal graphs on $n$ vertices using $O(n^7)$ arithmetic operations.

\section{Implementation of the counting algorithm}
\label{sec:implementation}

An implementation of the counting algorithm in C{}\verb!++! can be successfully run for inputs as large as $n = 30$ in about 2.5 minutes on a standard desktop computer.\footnote{Our implementation is available on GitHub at https://github.com/uhebertj/chordal.} Previously, the number of labeled chordal graphs was only known up to $n = 15$. \cref{table:conn_chordal} shows the number of connected chordal graphs on $n$ vertices for $n\le 30$, with the chromatic number unrestricted. \cref{table:n_and_omega} shows the number of $\omega$-colorable connected chordal graphs on $n$ vertices for various values of $n$ and $\omega$.\footnote{These tables were made using the version of the algorithm that uses $O(n^8)$ arithmetic operations, without the optimization using the function $h$, since this version ran fastest in practice. This is why there is no mention of $h$ in the implementation on GitHub.}

\begin{table}[h!]
\caption{Numbers of labeled connected chordal graphs on $n$ vertices}
\label{table:conn_chordal}
\centering
\medskip
\begin{tabular}{r r}
$c(n)$ & $n$ \\ [0.5ex]
1 & 1 \\
1 & 2 \\
4 & 3 \\
35 & 4 \\
541 & 5 \\
13302 & 6 \\
489287 & 7 \\
25864897 & 8 \\
1910753782 & 9 \\
193328835393 & 10 \\
26404671468121 & 11 \\
4818917841228328 & 12 \\
1167442027829857677 & 13 \\
374059462390709800421 & 14 \\
158311620026439080777076 & 15 \\
88561607724193506845709239 & 16 \\
65629642803250494352023169033 & 17 \\
64646285130595946195244365518454 & 18 \\
84997214469704246545711429635276299 & 19 \\
149881423568752945444616261913109046421 & 20 \\
356260551239284266908724943672911100488558 & 21 \\
1147374494946449194450825817605340123679150461 & 22 \\
5032486852040265322461550844695939678052967384053 & 23 \\
30210545039307528599583618386687349227933725131035504 & 24 \\
249400383130659050580193267861459579254489822650065685961 & 25 \\
2844134548699568981561554629043146070324332400944867482340313 & 26 \\
44993294034522185332489548856700572371349354518671249097245374660 & 27 \\
991277251392360301443460288397009109066708275778086061470009877027739 & 28 \\
30526157144572224953157514915475479605501638476250575941226904780179348933 & 29 \\
1318363800739595427128835554231270770209426196402736248743162258824492158995254 & 30 \\ [1ex]
\end{tabular}
\end{table}

\begin{table}[h!]
\caption{Numbers of $\omega$-colorable labeled connected chordal graphs on $n$ vertices. \\
When $\omega = 2$, the algorithm counts labeled trees.}
\label{table:n_and_omega}
\centering
\medskip
\begin{tabular}{r r r r r r r r r r r}
\multicolumn{10}{c}{$n$} \\ [.2cm]
2 & 3 & 4 & 5 & 6 & 7 & 8 & 9 & & & \\ [.2cm]
1 & 3 & 16 & 125 & 1296 & 16807 & 262144 & 4782969 & & 2 & {\multirow{8}{*}{$\omega$}} \\
& 4 & 34 & 480 & 9831 & 268093 & 9185436 & 379623492 & & 3 \\
& & 35 & 540 & 13136 & 466683 & 22732032 & 1437072780 & & 4 \\
& & & 541 & 13301 & 488873 & 25736782 & 1873146621 & & 5 \\
& & & & 13302 & 489286 & 25863916 & 1910084529 & & 6 \\
& & & & & 489287 & 25864896 & 1910751531 & & 7 \\
& & & & & & 25864897 & 1910753781 & & 8 \\
& & & & & & & 1910753782 & & 9 \\
\end{tabular}
\begin{tabular}{r r r r r r}
\multicolumn{5}{c}{$n$} \\ [.2cm]
10 & 11 & 12 & & & \\ [.2cm]
100000000 & 2357947691 & 61917364224 & & 2 & {\multirow{11}{*}{$\omega$}} \\
18376225525 & 1019282908941 & 63707908718994 & & 3 \\
112588153700 & 10535042533301 & 1144261607209084 & & 4 \\
181962472490 & 22726623077466 & 3513611793935959 & & 5 \\
192919501307 & 26158547399061 & 4666697716137194 & & 6 \\
193325509217 & 26400465973728 & 4813890013657154 & & 7 \\
193328830337 & 26404655450778 & 4818876084111431 & & 8 \\
193328835392 & 26404671456933 & 4818917765689886 & & 9 \\
193328835393 & 26404671468120 & 4818917841203841 & & 10 \\
& 26404671468121 & 4818917841228327 & & 11 \\
& & 4818917841228328 & & 12 \\
\end{tabular}
\end{table}

\section{Proof of \cref{thm:conn_counting} (counting connected chordal graphs)}
\label{sec:conn_counting_proof}

In this section, we prove correctness of the recurrences presented in \cref{sec:recurrences}, in order to prove \cref{thm:conn_counting}. \cref{sec:chordal_properties,sec:evap_properties} contain a number of lemmas on chordal graphs and their evaporation sequences. In \cref{sec:recurrence_proofs}, we use these lemmas to prove the recurrences. We then discuss the running time in \cref{sec:wrap_up_proof}.

\subsection{Chordal graph properties}
\label{sec:chordal_properties}

We make use of the following standard facts about chordal graphs. The proof of \cref{lemma:simplicial} can be found in \cite{blair1993introduction}.

\begin{lemma}
\label{lemma:simplicial}
Every chordal graph contains a simplicial vertex.
\end{lemma}

\begin{lemma}
\label{lemma:gluing}
Let $G_1$, $G_2$ be two chordal graphs, and suppose $Y\coloneqq V(G_1)\cap V(G_2)$ is a clique in both $G_1$ and $G_2$. If we glue $G_1$ and $G_2$ together at $Y$, then the resulting graph $G$ is chordal.
\end{lemma}

\begin{proof}
Suppose $G$ contains an induced cycle $C$ of length at least 4. Since $G_1$ and $G_2$ are chordal, $C$ contains some vertex $u\in V(G_1)\setminus Y$ and some $v\in V(G_2)\setminus Y$. There are no edges connecting $G_1\setminus Y$ to $G_2\setminus Y$, so $C$ contains two non-adjacent vertices that belong to the clique $Y$, which is a contradiction.
\end{proof}

\subsection{Properties of the evaporation sequence}
\label{sec:evap_properties}

\subsubsection{Basic properties}
\label{sec:basic_prop}

In this section (\cref{sec:basic_prop}), suppose $G$ is a chordal graph that evaporates at time $t$ with exception set $X$, where $X\subseteq V(G)$ is a clique, and let $L_1,\ldots,L_t$ be the evaporation sequence. For $i\in[t]$, let $G_i$ be the state of the graph $G$ just before time step $i$. In other words, $G_1 = G$ and $G_i = G\setminus(L_1\cup\cdots\cup L_{i-1})$ for all $i\in[t]$.

\begin{observation}
\label{obs:cliques}
For all $i\in[t]$, every connected component of $G[L_i]$ is a clique.
\end{observation}

\begin{proof}
Let $i\in[t]$, and suppose $u$ and $v$ are two vertices that belong to the same connected component in $G[L_i]$. We have some path $u,w_1,w_2,\ldots,w_k,v$ from $u$ to $v$ in that connected component. Since $w_1$ is simplicial in $G_i$ and thus in $G[L_i]$, $u$ is adjacent to $w_2$. Similarly, by induction, $u$ is adjacent to all of the vertices $w_1,\ldots,w_k,v$. Hence $u$ and $v$ are adjacent, as desired.
\end{proof}

\begin{observation}
\label{obs:clique}
If $G\setminus X$ is connected, then $L_t$ is a clique.
\end{observation}

\begin{proof}
Let $u,v\in L_t$, and let $P$ be a shortest path from $u$ to $v$ in $G\setminus X$. None of the vertices of $P$ are simplicial in $G_i$ for $i\in[t-1]$, since otherwise there would be a shorter path than $P$ from $u$ to $v$. Hence $P\subseteq L_t$, so $u$ and $v$ belong to the same connected component of $G[L_t]$, which is a clique by \cref{obs:cliques}.
\end{proof}

\begin{observation}
\label{obs:neighborhood}
Let $i\in[t]$, let $C$ be a connected component of $G[L_i]$, and let $u,v\in C$. Then $N_{G_i}(u)\setminus\{v\} = N_{G_i}(v)\setminus\{u\}$.
\end{observation}

\begin{proof}
If $u = v$ then we are done, so suppose $u\ne v$. Let $w\in N_{G_i}(u)\setminus\{v\}$. We know $u$ and $v$ are adjacent by \cref{obs:cliques}, and $u$ is simplicial in $G_i$. Hence $w\in N_{G_i}(v)\setminus\{u\}$.
\end{proof}

\begin{observation}
\label{obs:L_union_X}
If $G\setminus X$ is connected, $X$ is a clique, and $X\subseteq N(L_t)$, then $X\cup L_t$ is a clique.
\end{observation}

\begin{proof}
By \cref{obs:neighborhood} with $i = t$, we know $N(u)\cap X = N(v)\cap X$ for all $u,v\in L_t$. This implies that $x$ and $v$ are adjacent for all $x\in X$, $v\in L_t$. Since $L_t$ is a clique by \cref{obs:clique}, we conclude that $X\cup L_t$ is a clique.
\end{proof}

\begin{observation}
\label{obs:conn}
If $G\setminus X$ is connected, then $G_i\setminus X$ is connected for all $i\in[t]$.
\end{observation}

\begin{proof}
We proceed by induction by showing that if $G_i\setminus X$ is connected, then $G_{i+1}\setminus X$ is connected, for all $i\in [t-1]$. Suppose $G_i\setminus X$ is connected for some $i$, and let $u,v$ be vertices in $G_{i+1}\setminus X$. Let $P$ be a shortest path from $u$ to $v$ in $G_i\setminus X$. No intermediate vertices on $P$ are simplicial in $G_i$, since otherwise there would be a shorter path from $u$ to $v$ in $G_i\setminus X$. Therefore, the same path $P$ still exists in $G_{i+1}\setminus X$, and hence $G_{i+1}\setminus X$ is connected.
\end{proof}

\begin{lemma}
\label{lemma:neighborhood}
Let $C = V(G)\setminus X$. If $G\setminus X$ is connected, then $N(C) = N(L_t)\cap X$.
\end{lemma}


\begin{proof}
It is clear that $N(L_t)\cap X\subseteq N(C)$. To see that $N(C)\subseteq N(L_t)\cap X$, it is sufficient to show $N(L_i)\cap X\subseteq N(L_{i+1}\cup\ldots\cup L_t)\cap X$ for all $i\in [t-1]$. Suppose $u\in N(L_i)\cap X$, so there is some vertex $v\in L_i$ that is adjacent to $u\in X$. We know $G_i\setminus X$ is connected by \cref{obs:conn}, so there is a path $v,w_1,w_2,\ldots,w_k,w$ in $G_i\setminus X$ from $v$ to some vertex $w\in L_{i+1}\cup\ldots\cup L_t$. We can assume $v,w_1,w_2,\ldots,w_k\in L_i$ by shortening the tail of the path if need be. Since $v$ is simplicial in $G_i$, $w_1$ is adjacent to $u$. Similarly, by induction, all of the vertices $w_1,\ldots,w_k,w$ are adjacent to $u$. Hence $u\in N(L_{i+1}\cup\ldots\cup L_t)\cap X$.
\end{proof}

\begin{lemma}
\label{lemma:neighborhood_2}
If $C$ is a connected component of $G\setminus X$, then $N(C) = N(L_t\cap C)\cap X$.
\end{lemma}

\begin{proof}
By an argument similar to \cref{obs:conn}, since $C$ is connected, what remains of $C$ in $G_i$ is connected for all $i\in[t]$. Therefore, we can see that $N(L_i\cap C)\cap X\subseteq N((L_{i+1}\cup\ldots\cup L_t)\cap C)\cap X$ for all $i\in [t-1]$ by an argument similar to the proof of \cref{lemma:neighborhood}.
\end{proof}

\subsubsection{Evaporation time is preserved after gluing or taking induced subgraphs}

The rest of the evaporation sequence lemmas that we need all have a similar flavor. For almost every recurrence, we need to show that the evaporation time of a chordal graph is preserved after gluing or taking induced subgraphs, as long as we choose the exception sets correctly and certain properties hold. For example, the simplest of these lemmas, \cref{lemma:commute_g}, will be used to prove the recurrence for $g(t,x,k,z)$ (\cref{lemma:g}).

\begin{lemma}
\label{lemma:commute_g_1}
Suppose $G$ is a connected chordal graph that contains a clique $X$, and suppose $C$ is a connected component of $G\setminus X$. Then the evaporation time of $C$ in $G$ is equal to the evaporation time of $G[X\cup C]$, assuming both have exception set $X$.
\end{lemma}

\begin{proof}
Let $G' = G[X\cup C]$. Let $L_1,\ldots,L_t$ be the evaporation sequence of $G$ with exception set $X$, and let $G_i = G\setminus(L_1\cup\cdots\cup L_{i-1})$ for $i\in[t]$. Similarly, let $L_1',\ldots,L_s'$ be the evaporation sequence of $G'$ with exception set $X$, and let $G_i' = G'\setminus(L_1'\cup\cdots\cup L_{i-1}')$ for $i\in[s]$. For every $v\in C$, we know $N_{G'}(v) = N_G(v)$ since $N_G(v)\subseteq X\cup C$, so $v$ is simplicial in $G$ if and only if $v$ is simplicial in $G'$. Hence $L_1' = L_1\cap C$. Continuing by induction, we can see that at each time step, the same set of vertices evaporates from $C$ in both $G_i$ and $G_i'$, and the vertices of $X$ never evaporate. More precisely, for all $1\le i\le s$, and for all $v\in C$ that have not yet evaporated, we have $N_{G_i'}(v) = N_{G_i}(v)$, and thus $L_i' = L_i\cap C$. Therefore, $C$ has the same evaporation time in both $G$ and $G'$.
\end{proof}

\begin{lemma}
\label{lemma:commute_g}
Suppose $G$ is a connected chordal graph that contains a clique $X$, and suppose $C_1,\ldots,C_r$ are connected components of $G\setminus X$. Then for every $j\in[r]$, the evaporation time of $C_j$ in $G$ is equal to the evaporation time of $C_j$ in $G[X\cup C_1\cup\cdots\cup C_r]$, assuming both have exception set $X$.
\end{lemma}

\begin{proof}
For each $j\in[r]$, the statement for $C_j$ follows immediately from two applications of \cref{lemma:commute_g_1}: first with the graph $G$ and the component $C_j$, and then with the graph $G[X\cup C_1\cup\cdots\cup C_r]$ and the component $C_j$.
\end{proof}

\begin{lemma}
\label{lemma:commute_f_1}
Suppose $G$ is a connected chordal graph that evaporates at time $t\ge 2$ with exception set $X$. Let $L = L_G(X)$, and suppose $X\cup L$ is a clique. Let $A$ be the set of vertices in connected components $C$ of $G\setminus(X\cup L)$ such that $C$ evaporates at time exactly $t-1$ in $G$ with exception set $X$, and let $B$ be the set of vertices of all other components of $G\setminus(X\cup L)$. Let $G_1 = G[X\cup L\cup A]$ and $G_2 = G[X\cup L\cup B]$. Then the following statements hold:
\begin{enumerate}
    \item $G_1$ evaporates at time $t$, every component of $G_1\setminus (X\cup L)$ evaporates at time $t-1$ in $G_1$, and we have $L_{G_1}(X) = L$, when $X$ is the exception set. Furthermore, $A$ is nonempty.
    \item $G_2$ evaporates in time at most $t-2$ with exception set $X\cup L$.
\end{enumerate}
\end{lemma}


\begin{proof}
To prove (1), we proceed using an argument similar to \cref{lemma:commute_g_1}. In that lemma, the induction step relied on the fact that the vertices of $X$ never evaporate. For a similar argument to hold here, we need to show that at time steps $1,\ldots,t-1$ in $G_1$, none of the vertices of $L$ evaporate. Let $L_1,\ldots,L_t$ be the evaporation sequence of $G$ with exception set $X$, and let $G_i = G\setminus(L_1\cup\cdots\cup L_{i-1})$ for $i\in[t]$. Similarly, let $L_1^{(1)},\ldots,L_s^{(1)}$ be the evaporation sequence of $G_1$ with exception set $X$, and let $G_i^{(1)} = G_1\setminus(L_1^{(1)}\cup\cdots\cup L_{i-1}^{(1)})$ for $i\in[s]$. None of the vertices of $L$ are simplicial in $G_{t-1}$. This means that every vertex $u\in L$ has two distinct non-adjacent neighbors $v_1,v_2\in A$ that belong to $G_{t-1}$, since the other components of $G\setminus (X\cup L)$ besides those in $A$ have already evaporated before we reach $G_{t-1}$. The existence of $v_1$ and $v_2$ (since $L$ is nonempty) implies that $A$ is nonempty.

By an argument similar to \cref{lemma:commute_g_1}, for every $u\in L$, its neighbors $v_1$ and $v_2$ do not evaporate in the first $t-2$ steps of the evaporation sequence of $G_1$, i.e., they still exist in $G_{t-1}^{(1)}$. Furthermore, the vertices that have evaporated from $A$ in each time step to create $G_t$ are the same as those for $G_t^{(1)}$. Therefore, every component of $G_1\setminus (X\cup L)$ evaporates at time $t-1$ in $G_1$. In the end, $G_t^{(1)}$ is the clique $X\cup L$, so all of $L$ is now simplicial, which yields $L_{G_1}(X) = L$.

To prove (2), we again use an argument similar to \cref{lemma:commute_g_1}. In both $G$ and $G_2$ (with the specified exception sets), none of the vertices of $X\cup L$ evaporate before time step $t$. Therefore, in the first $t-2$ time steps, the evaporation behavior of $G_2$ is exactly the same as it was for the corresponding vertices in $G$, which evaporate in time at most $t-2$.
\end{proof}

\begin{lemma}
\label{lemma:commute_f_2}
Let $G_1$ and $G_2$ be connected chordal graphs such that $V(G_1)\cap V(G_2) = X\cup L$, where $L = L_{G_1}(X)$ and $X\cup L$ is a clique. Suppose $G_1$ evaporates at time $t\ge 2$ with exception set $X$, and suppose every connected component of $G_1\setminus(X\cup L)$ evaporates at time exactly $t-1$ in $G_1$ with exception set $X$. Also, suppose every connected component of $G_2\setminus(X\cup L)$ evaporates in time at most $t-2$ in $G_2$ with exception set $X\cup L$. If we glue $G_1$ and $G_2$ together at $X\cup L$, then the resulting graph $G$ evaporates at time $t$ with exception set $X$, and $L_G(X) = L$.
\end{lemma}

\begin{proof}
Adding more vertices to a graph can only increase the evaporation time of a given vertex. Therefore, when we attach $G_2$ to $G_1$ to obtain $G$, the resulting evaporation time of each vertex in $L$ is at least $t$ (with exception set $X$). This implies that the evaporation behavior of the vertices of $G_2\setminus(X\cup L)$ is the same in $G$ as it was in $G_2$ by an argument similar to \cref{lemma:commute_g_1}, since all of $X\cup L$ certainly stays alive until time $t$ in both $G_2$ (since it is the exception set) and $G$ (by the argument above).

Now we just need to show that the evaporation behavior of the vertices of $G_1$ is the same in $G$ as it was in $G_1$. When we attach $G_2$ to $G_1$, the only vertices of $G_1$ whose neighborhoods can change are those in $X\cup L$. The only way this could affect the evaporation behavior of $G_1$ would be by increasing the evaporation time of some vertices in $L$. However, we just showed that all of the vertices of $G_2\setminus(X\cup L)$ evaporate in time at most $t-2$. Therefore, just before time step $t$ (with exception set $X$), $G_1$ and $G$ are identical since by then, all of $G_2\setminus G_1$ has evaporated. Therefore, all of the vertices of $L$ evaporate at time $t$ in $G$, just as they did in $G_1$.
\end{proof}

\begin{lemma}
\label{lemma:commute_f_tilde_1}
Suppose $G$ is a connected chordal graph such that $X\cup L$ is a clique, where $L = L_G(X)$, and suppose every connected component of $G\setminus(X\cup L)$ evaporates at time exactly $t-1$ in $G$ with exception set $X$, where $t\ge 2$. Suppose $C_1,\ldots,C_r$ are connected components of $G\setminus(X\cup L)$. Then for every $j\in[r]$, $C_j$ evaporates at time exactly $t-1$ in $G[X\cup L\cup C_1\cup\cdots\cup C_r]$ with exception set $X\cup L$.
\end{lemma}

\begin{proof}
Let $G' = G[X\cup L\cup C_1\cup\cdots\cup C_r]$. By \cref{lemma:commute_g}, the evaporation time of $C_j$ in $G$ with exception set $X\cup L$ is equal to the evaporation time of $C_j$ in $G'$ with the same exception set. Furthermore, the evaporation time of $C_j$ in $G$ is the same with either exception set ($X$ or $X\cup L$) since the evaporation time of $C_j$ in $G$ is less than $t$ when $X$ is the exception set.
\end{proof}

Suppose $G$ is a chordal graph with evaporation sequence $L_1,\ldots,L_t$ and exception set $X$, where $t\ge 2$. For a connected component $C$ of $G[L_{t-1}]$, since by \cref{obs:neighborhood} all vertices of $C$ have the same neighborhood outside of $C$ in $G_{t-1}$, we denote the neighborhood of $C$ in $X\cup L_t$ by $\hat N(C)\coloneqq N_{G_{t-1}}(v)\cap(X\cup L_t) = N_G(v)\cap(X\cup L_t)$ for all $v\in C$.

Given a universe $U$ and a subset $L\subseteq U$, we say a collection $\mathcal{S}$ of subsets of $U$ is a \emph{set cover} for $L$ if $\bigcup_{S\in\mathcal{S}} S\supseteq L$. For a chordal graph $G$ as described above, let $$\mathcal{S}_G = \{\hat N(C):\text{$C$ is a connected component of $G[L_{t-1}]$}\}.$$ Each set in $\mathcal{S}_G$ is a subset of the universe $X\cup L_t$.

\begin{lemma}
\label{lemma:set_cover}
Suppose $G$ is a chordal graph with evaporation sequence $L_1,\ldots,L_t$ and exception set $X$, where $t\ge 2$, and suppose $X\cup L_t$ is a clique. Then one of the following holds:
\begin{enumerate}
    \item There is at most one connected component $C$ of $G[L_{t-1}]$ such that $\hat N(C) = X\cup L_t$, and $\mathcal{S}_G\setminus\{X\cup L_t\}$ is a set cover for $L_t$.
    \item There are at least two distinct connected components $C_1,C_2$ of $G[L_{t-1}]$ such that $\hat N(C_i) = X\cup L_t$ for $i\in\{1,2\}$.
\end{enumerate}
\end{lemma}

\begin{proof}
We show that if there is at most one connected component of $G[L_{t-1}]$, say $C$, such that $\hat N(C) = X\cup L_t$, then $\mathcal{S}_G\setminus\{X\cup L_t\}$ is a set cover for $L_t$. If $C$ does not exist, then we let $C = \emptyset$ to simplify notation. Suppose to the contrary that there is some vertex $v\in L_t$ that is not covered by any neighborhood $S\subsetneq X\cup L_t$ in $\mathcal{S}_G$. If there is a component $C'$ of $G[L_{t-1}]$ such that $v\in\hat N(C')$, then we must have $\hat N(C') = X\cup L_t$. Therefore, there is at most one such component $C'$, and if one exists we have $C' = C$. Hence $N_{G_{t-1}}(v) = (X\cup L_t\cup C)\setminus\{v\}$. Since $C$ is a clique by \cref{obs:cliques}, $X\cup L_t\cup C$ is a clique. This means $v$ is simplicial in $G_{t-1}$, which contradicts $v\in L_t$.
\end{proof}

\begin{lemma}
\label{lemma:commute_f_tilde_2}
Suppose $G$ is a connected chordal graph that evaporates at time $t\ge 2$ with exception set $X$. Let $L = L_G(X)$, and suppose $X\cup L$ is a clique. Furthermore, suppose every connected component of $G\setminus(X\cup L)$ evaporates at time exactly $t-1$ in $G$ with exception set $X$. Suppose exactly one connected component $C$ of $G\setminus(X\cup L)$ sees all of $X\cup L$, and let $D$ be the set of vertices of all other components of $G\setminus(X\cup L)$. Let $G_2 = G[X\cup L\cup D]$. Then when $X$ is the exception set, we have the following: $G_2$ evaporates at time $t$, every component of $G_2\setminus(X\cup L)$ evaporates at time $t-1$ in $G_2$, and $L_{G_2}(X) = L$. Furthermore, $D$ is nonempty.
\end{lemma}

\begin{proof}
Let $L_1,\ldots,L_t$ be the evaporation sequence of $G$ with exception set $X$. By \cref{lemma:neighborhood_2} (with exception set $X\cup L$ rather than $X$), we have $N(C) = N(L_{t-1}\cap C)\cap(X\cup L)$ for every connected component $C'$ of $G\setminus(X\cup L)$. Furthermore, the connected components of $G\setminus(X\cup L)$ each intersected with $L_{t-1}$ are the same as the connected components of $G[L_{t-1}]$ by an argument similar to \cref{obs:conn}. Exactly one component of $G\setminus(X\cup L)$ sees all of $X\cup L$, so this means that exactly one component of $G[L_{t-1}]$ sees all of $X\cup L$. By \cref{lemma:set_cover}, this implies that $\mathcal{S}_G\setminus\{X\cup L\}$ is a set cover for $L$. We observe that $\mathcal{S}_G\setminus\{X\cup L\} = \mathcal{S}_{G_2}\setminus\{X\cup L\}$ since we only removed the component $C$ whose neighborhood is $\hat N(C) = X\cup L$ from $G$ when defining $G_2$. Therefore, $\mathcal{S}_{G_2}\setminus\{X\cup L\}$ is also a set cover for $L$. This means $D$ must be nonempty since $L$ is nonempty.

As long as none of the vertices of $L$ evaporate before time $t$ in $G_2$, then we can see by an argument similar to \cref{lemma:commute_g_1} that the evaporation behavior of $G_2$ is the same as it was for the corresponding vertices in $G$ (both with exception set $X$). Suppose for a contradiction that $u\in L$ is the first vertex in $L$ to evaporate before time $t$ in $G_2$. Suppose this happens at time $s$, where $s<t$. By an argument similar to \cref{lemma:commute_g_1}, since none of $L$ had evaporated before that point, all of $L_{t-1}$ still exists in $G_2$ just before time $s$. The vertex $u$ is contained in some neighborhood from $\mathcal{S}_{G_2}\setminus\{X\cup L\}$ since this is a set cover. Let $v$ be a vertex in $L_{t-1}$ from the component of $G_2\setminus(X\cup L)$ corresponding to that neighborhood. Since that neighborhood is a proper subset of $X\cup L$, there is also some vertex $w\in X\cup L$ adjacent to $u$ that does not belong to that neighborhood. Hence $v$ and $w$ are not adjacent, so $u$ is not simplicial at time $s$, which is a contradiction.

Therefore, the evaporation behavior of $G_2$ is the same as it was for the corresponding vertices in $G$, so $G_2$ evaporates at time $t$, every component of $G_2\setminus(X\cup L)$ evaporates at time $t-1$, and $L_{G_2}(X) = L$.
\end{proof}

\begin{lemma}
\label{lemma:commute_f_tilde_3}
Let $G_1$ and $G_2$ be connected chordal graphs such that $V(G_1)\cap V(G_2) = \hat X$, where $\hat X$ is a clique. Suppose $G_1\setminus\hat X$ has at least two connected components, both of which see all of $\hat X$. Suppose every connected component of $G_1\setminus\hat X$ (resp.\ $G_2\setminus\hat X$) evaporates at time exactly $t-1$ in $G_1$ (resp.\ $G_2$) with exception set $\hat X$, where $t\ge 2$. Let $X\subsetneq\hat X$. If we glue $G_1$ and $G_2$ together at $\hat X$, then the resulting graph $G$ evaporates at time $t$ with exception set $X$, every connected component of $G\setminus\hat X$ evaporates at time $t-1$ in $G$, and $L_G(X) = \hat X\setminus X$.
\end{lemma}

\begin{proof}
When we glue $G_1$ and $G_2$ together (and change the exception set from $\hat X$ to $X$), as long as it is still the case that none of the vertices of $\hat X$ evaporate before time $t$, then we can see by an argument similar to \cref{lemma:commute_g_1} that the evaporation behavior of these two graphs does not change. Suppose for a contradiction that $u\in\hat X\setminus X$ is the first vertex in $\hat X$ to evaporate before time $t$ in $G$. Suppose this happens at time $s$, where $s<t$. Let $L_1,\ldots,L_t$ be the evaporation sequence of $G_1$ with exception set $\hat X$. Since $u$ is the first such vertex to evaporate, all of $L_{t-1}$ still exists in $G$ just before time $s$. Let $v$ and $w$ be neighbors of $u$ from $L_{t-1}$ that belong to two distinct components of $G_1\setminus\hat X$. Since $v$ and $w$ are not adjacent, this means $u$ is not simplicial at time $s$, which is a contradiction. Therefore, the evaporation behavior of the vertices of $G_1\setminus\hat X$ and $G_2\setminus\hat X$ in $G$ is the same as in $G_1$ and $G_2$. At time $t$, all of $\hat X\setminus X$ evaporates since $\hat X$ is a clique.
\end{proof}

\begin{lemma}
\label{lemma:commute_f_p_tilde_1}
Suppose $G$ is a connected chordal graph that evaporates at time $t\ge 2$ with exception set $X$. Let $L = L_G(X)$, and suppose $X\cup L$ is a clique. Furthermore, suppose every connected component of $G\setminus(X\cup L)$ evaporates at time exactly $t-1$ in $G$ with exception set $X$. Let $C$ be a component of $G\setminus(X\cup L)$, and let $L' = N(C)\cap L$. Let $D$ be the set of vertices of all other components of $G\setminus(X\cup L)$, and let $G_2 = G[X\cup L\cup D]$. If $L'\subsetneq L$, then when $X\cup L'$ is the exception set, we have the following: $G_2$ evaporates at time $t$, every component of $G_2\setminus(X\cup L)$ evaporates at time $t-1$ in $G_2$, and $L_{G_2}(X\cup L') = L\setminus L'$; furthermore, $D$ is nonempty.
\end{lemma}

\begin{proof}
Let $L_1,\ldots,L_t$ be the evaporation sequence of $G$ with exception set $X$. By an argument similar to the first paragraph of the proof of \cref{lemma:commute_f_tilde_2}, we know $\hat N(C)$ is a proper subset of $X\cup L$ for every connected component $C$ of $G[L_{t-1}]$. Therefore, by \cref{lemma:set_cover}, $\mathcal{S}_G\setminus\{X\cup L\}$ is a set cover for $L$, so it follows that $\mathcal{S}_{G_2}\setminus\{X\cup L\}$ is a set cover for $L\setminus L'$. Since $L'\subsetneq L$, it is easy to see that $D$ is nonempty.

By an argument similar to the second paragraph of the proof of \cref{lemma:commute_f_tilde_2}, none of the vertices of $L\setminus L'$ evaporate before time $t$ in $G_2$. Therefore, the evaporation behavior of $G_2$ (with exception set $X\cup L'$) is the same as it was for the corresponding vertices in $G$ (with exception set $X$), so $G_2$ evaporates at time $t$, every component of $G_2\setminus(X\cup L)$ evaporates at time $t-1$, and $L_{G_2}(X\cup L') = L\setminus L'$.
\end{proof}

\begin{lemma}
\label{lemma:commute_f_p_tilde_2}
Let $G_1$ and $G_2$ be connected chordal graphs such that $V(G_1)\cap V(G_2) = \widetilde X$. Suppose $G_2$ contains a clique $\hat X\supseteq\widetilde X$. Furthermore, suppose $G_1\setminus \widetilde X$ has exactly one connected component, which sees all of $\widetilde X$ and evaporates at time exactly $t-1$ in $G_1$ with exception set $\widetilde X$, where $t\ge 2$. Let $L'$ be a subset of $\widetilde X$, and let $X = \hat X\setminus L'$. Let $G$ be the graph obtained by gluing together $G_1$ and $G_2$ at $\widetilde X$. Then we have the following:
\begin{enumerate}
    \item Suppose $G_2$ evaporates at time $t$ with exception set $\hat X$, and let $\hat L = L_{G_2}(\hat X)$. Suppose $\hat X\cup\hat L$ is a clique in $G_2$, and suppose every connected component of $G_2\setminus(\hat X\cup\hat L)$ evaporates at time exactly $t-1$ in $G_2$. Then $G$ evaporates at time $t$ with exception set $X$, every connected component of $G\setminus(\hat X\cup\hat L)$ evaporates at time $t-1$ in $G$, and $L_G(X) = L'\cup\hat L$.
    \item Suppose every connected component of $G_2\setminus\hat X$ evaporates at time exactly $t-1$ in $G_2$ with exception set $\hat X$, and suppose $\widetilde X\subsetneq\hat X$. Then $G$ evaporates at time $t$ with exception set $X$, every connected component of $G\setminus\hat X$ evaporates at time $t-1$ in $G$, and $L_G(X) = L'$.
\end{enumerate}
\end{lemma}

\begin{proof}
First, we prove statement (1). When we glue $G_1$ and $G_2$ together (and change the exception set to $X$), as long as it is still the case that none of the vertices of $\hat X$ evaporate before time $t$, then we can see by an argument similar to \cref{lemma:commute_g_1} that the evaporation behavior of these two graphs does not change. Suppose for a contradiction that $u\in\hat X\setminus X$ is the first vertex in $\hat X$ to evaporate before time $t$ in $G$. Suppose this happens at time $s$, where $s<t$. Let $L_1,\ldots,L_{t-1}$ be the evaporation sequence of $G_1$ with exception set $\widetilde X$. Let $v\in L_{t-1}$, and let $C = V(G_1)\setminus\widetilde X$. By \cref{lemma:neighborhood}, we know $\widetilde X = N(C) = N(L_{t-1})\cap\widetilde X$, so $u$ and $v$ are adjacent since $u\in\widetilde X$. Since $\widetilde X\subsetneq\hat X\cup\hat L$, there is some vertex $w\in\hat X\cup\hat L$ that is adjacent to $u$ but not $v$. Therefore, $u$ is not simplicial at time $s$, which is a contradiction. Therefore, the evaporation behavior of the vertices of $G_1\setminus\widetilde X$ and $G_2\setminus\hat X$ in $G$ is the same as in $G_1$ and $G_2$. At time $t$, all of $L'\cup\hat L$ evaporates since $\hat X\cup\hat L$ is a clique.

For statement (2), the proof is exactly the same except for the following differences. The sentence that begins ``Since $\widetilde X\subsetneq\hat X\cup\hat L$…'' should instead say, ``Since $\widetilde X\subsetneq\hat X$, there is some vertex $w\in\hat X$ that is adjacent to $u$ but not $v$.'' And we change the concluding sentence to the following: ``At time $t$, all of $L'$ evaporates since $\hat X$ is a clique.''
\end{proof}

\subsection{Proofs of recurrences}
\label{sec:recurrence_proofs}

As in the counting algorithm, all graphs in this section are $\omega$-colorable. We do not explicity mention this in each lemma, since $\omega$-colorablitity does not change the recurrences (it only affects the base cases).

By \cref{lemma:simplicial}, every induced subgraph of a chordal graph contains a simplicial vertex, so a chordal graph on $n$ vertices evaporates in time at most $n$ when the exception set is $X = \emptyset$. Therefore, the number of connected chordal graphs with vertex set $[n]$ is $$c(n) = \sum_{t=1}^n\tilde g_1(t,0,n).$$ Throughout this section, suppose $t,x,l,k,z$ are nonnegative integers in the domain of the relevant function, and suppose we are not yet at a base case. Let $G(t,x,k,z)$ be the set of graphs counted by $g(t,x,k,z)$, and similarly let $\widetilde G(t,x,k,z)$, $\widetilde G_p(t,x,k,z)$, $\widetilde G_1(t,x,k)$, $\widetilde G_{\ge 2}(t,x,k)$, $F(t,x,l,k)$, $\widetilde F(t,x,l,k)$, $\widetilde F_p(t,x,l,k)$, and $\widetilde F_p(t,x,l,k,z)$, be the sets of graphs counted by $\tilde g(t,x,k,z)$, $\tilde g_p(t,x,k,z)$, $\tilde g_1(t,x,k)$, $\tilde g_{\ge 2}(t,x,k)$, $f(t,x,l,k)$, $\tilde f(t,x,l,k)$, $\tilde f_p(t,x,l,k)$, and $\tilde f_p(t,x,l,k,z)$, respectively.

We are now ready to prove correctness of the recurrences. For the proofs of the following lemmas, let $X = [x]$ according to the current value of the argument $x$. We proceed in the order in which these functions were defined in \cref{defn:counters}, so that similar functions and proofs appear consecutively. (This is somewhat different from the ordering of the lemma numbers.)

\g*

\begin{proof}
To see that $g(t,x,k,z)$ is at most this summation, suppose $G\in G(t,x,k,z)$. Let $A$ be the set of vertices in components $C$ of $G\setminus X$ such that $C$ evaporates at time exactly $t$ in $G$ with exception set $X$, and let $B$ be the set of vertices of all other components of $G\setminus X$. Let $G_1 = G[X\cup A]$, $G_2 = G[X\cup B]$, and $k' = |A|$. Now relabel $G_1$ by applying $\phi(A,[x+1,x+k'])$ to the labels in $A$, and relabel $G_2$ by applying $\phi(B,[x+1,x+k-k'])$ to the labels in $B$. Note that $G_1$ and $G_2$ are connected since $G$ is connected. Since every component of $G\setminus X$ has at least one neighbor in $X\setminus[z]$, the same is true of $G_1\setminus X$ and $G_2\setminus X$. Therefore, by \cref{lemma:commute_g}, we have $G_1\in\widetilde G(t,x,k',z)$ and $G_2\in G(t-1,x,k-k',z)$.

We claim that this process (including the relabeling step) is injective, i.e., distinct graphs $G$ give rise to distinct triples $(G_1,G_2,A)$. Indeed, suppose that starting from $G\in G(t,x,k,z)$, we obtain $(G_1,G_2,A)$ along with the set $B$ of remaining components, and starting from a distinct graph $G'\in G(t,x,k,z)$, we obtain $(G_1',G_2',A')$ and the set $B'$. If $A\ne A'$, then we are done, so suppose $A = A'$. If $G[X\cup A]\ne G'[X\cup A']$, then after relabeling, $G[X\cup A]$ still differs from $G'[X\cup A']$ since $\phi(A,[x+1,x+k'])$ is a bijection. Otherwise, $G[X\cup B]\ne G'[X\cup B']$, in which case $G[X\cup B]$ still differs from $G'[X\cup B']$ even after relabeling. Therefore, $g(t,x,k,z)$ is at most the number of possible triples $(G_1,G_2,A)$, which is given by the summation above.

We now show that $g(t,x,k,z)$ is bounded below by the same summation. Suppose we are given $0\le k'\le k$, a graph $G_1\in\widetilde G(t,x,k',z)$, a graph $G_2\in G(t-1,x,k-k',z)$, and a subset $A\subseteq[x+1,x+k]$ of size $k'$. Let $B = [x+1,x+k]\setminus A$. We construct a graph $G\in G(t,x,k,z)$ as follows: Relabel $G_1$ by applying $\phi([x+1,x+k'],A)$ to the labels in $G_1\setminus X$, and relabel $G_2$ by applying $\phi([x+1,x+k-k'],B)$ to the labels in $G_2\setminus X$. By \cref{lemma:gluing}, since $X$ is a clique in both $G_1$ and $G_2$, we can glue $G_1$ and $G_2$ together at $X$ to obtain a chordal graph $G$. The resulting label set for $G$ is $[x+k]$, and $G$ is connected since $G_1$ and $G_2$ are connected. Furthermore, every component of $G\setminus X$ has at least one neighbor in $X\setminus[z]$ since the same is true of $G_1\setminus X$ and $G_2\setminus X$. By \cref{lemma:commute_g}, $G$ evaporates in time at most $t$, so we have $G\in G(t,x,k,z)$.

Suppose $(G_1,G_2,A)$ and $(G_1',G_2',A')$ are distinct triples of the above form that give rise to $G$ and $G'$, respectively. If $A\ne A'$, then by \cref{lemma:commute_g}, $G\ne G'$. By the same lemma, if $G_1\ne G_1'$ (resp.\ $G_2\ne G_2'$), then the components of $G\setminus X$ and $G'\setminus X$ that evaporate at time exactly $t$ (resp.\ less than $t$) will differ, which implies $G\ne G'$. Therefore, $g(t,x,k,z)$ is at least the summation above.
\end{proof}



In the above lemma, we showed that the process that builds $G_1$ and $G_2$ from $G$ (and vice versa) is injective, which ensures there is no double counting. In the lemmas that follow, we will build up similar injective maps in each direction. We only prove injectivity in detail in the proof of \cref{lemma:g} since the arguments for the following lemmas are similar.

\gTilde*

\begin{proof}
To see that $\tilde g(t,x,k,z)$ is at most this summation, suppose $G\in\widetilde G(t,x,k,z)$. Let $C$ be the component of $G\setminus X$ that contains the label $x+1$, and let $D$ be the set of vertices of all other components of $G\setminus X$. Let $X' = N(C)$, $G_1 = G[X'\cup C]$, $G_2 = G[X\cup D]$, $k' = |C|$, and $x' = |N(C)|$. Note that $X'\not\subseteq[z]$ since $C$ has at least one neighbor in $X\setminus[z]$. Now relabel $G_1$ by applying $\phi(X',[x'])$ to the labels in $X'$ and applying $\phi(C,[x'+1,x'+k'])$ to the labels in $C$. Relabel $G_2$ by applying $\phi(D,[x+1,x+k-k'])$ to the labels in $D$. Since every component of $G\setminus X$ has at least one neighbor in $X\setminus[z]$, the same is true of $G_2\setminus X$. Therefore, by \cref{lemma:commute_g}, we have $G_2\in\widetilde G(t,x,k-k',z)$. We also have $G_1\in\widetilde G_1(t,x',k')$ by an argument similar to \cref{lemma:commute_g}; the only difference here is that the exception set is $X'$ rather than $X$ since $C$ does not necessarily see all of $X$. (We chose $C$ to be the component of $G\setminus X$ that contains the lowest label rather than choosing an arbitrary component of $G\setminus X$ so that this process is injective.) Hence $\tilde g(t,x,k,z)$ is at most the summation above.

We now show that $\tilde g(t,x,k,z)$ is bounded below by this summation. Suppose we are given $1\le k'\le k$, $1\le x'\le x$, a graph $G_1\in\widetilde G_1(t,x',k')$, a graph $G_2\in\widetilde G(t,x,k-k',z)$, a subset $C\subseteq[x+1,x+k]$ of size $k'$ that contains $x+1$, and a subset $X'\subseteq X$ of size $x'$ such that $X'\not\subseteq[z]$. Let $D = [x+1,x+k]\setminus C$. We construct a graph $G\in\widetilde G(t,x,k,z)$ as follows: Relabel $G_1$ by applying $\phi([x'],X')$ to the labels in $[x']$ and applying $\phi([x'+1,x'+k'],C)$ to the labels in $G_1\setminus[x']$. Relabel $G_2$ by applying $\phi([x+1,x+k-k'],D)$ to the labels in $G_2\setminus X$. By \cref{lemma:gluing}, we can glue $G_1$ and $G_2$ together at $X'$ to obtain a chordal graph $G$. Every component of $G\setminus X$ has at least one neighbor in $X\setminus[z]$ since $G_2\setminus X$ has that property, and because $X'\not\subseteq[z]$. By an argument similar to \cref{lemma:commute_g}, $G$ evaporates at time exactly $t$, so we have $G\in\widetilde G(t,x,k,z)$. Therefore, $\tilde g(t,x,k,z)$ is at least the number of possible choices for $G_1$, $G_2$, $C$, and $X'$.
\end{proof}

\gTildeP*

\begin{proof}
The argument for $\tilde g_p$ is similar to the proof for $\tilde g$ (\cref{lemma:g_tilde}), except we do not allow $x' = x$ since no component of $G\setminus X$ sees all of $X$.
\end{proof}

\gTildeOne*

\begin{proof}
To see that $\tilde g_1(t,x,k)$ is at most this summation, suppose $G\in\widetilde G_1(t,x,k)$. Let $C = V(G)\setminus X$, $L = L_G(X)$, $A = C\setminus L$, and $l = |L|$. We know $N(C) = X$ by the definition of $\widetilde G_1(t,x,k)$, and by \cref{lemma:neighborhood}, this implies $N(L)\cap X = X$. Therefore, $X\cup L$ is a clique by \cref{obs:L_union_X}. Relabel $G$ by applying $\phi(L,[x+1,x+l])$ to $L$ and applying $\phi(A,[x+l+1,x+k])$ to $A$. We now have $G\in F(t,x,l,k-l)$ since $L_G(X)$ uses the label set $[x+1,x+l]$, and relabeling vertices does not change the evaporation time of the graph. Hence $\tilde g_1(t,x,k)$ is at most the summation above.

We now show that $\tilde g_1(t,x,k)$ is bounded below by this summation. Suppose we are given $1\le l\le k$, a graph $G\in F(t,x,l,k-l)$, and a subset $L\subseteq[x+1,x+k]$ of size $l$. Let $A = [x+1,x+k]\setminus L$. Now relabel $G$ by applying $\phi([x+1,x+l],L)$ to the labels in $[x+1,x+l]$ and applying $\phi([x+l+1,x+k],A)$ to the labels in $[x+l+1,x+k]$. Observe that $L$ sees all of $X$ since $X\cup L$ is a clique, which implies $G\in\widetilde G_1(t,x,k)$. Therefore, $\tilde g_1(t,x,k)$ is at least the number of possible choices for $G$ and $L$.
\end{proof}

\gTildeTwo*

\begin{proof}
To see that $\tilde g_{\ge 2}(t,x,k)$ is at most this summation, suppose $G\in\widetilde G_{\ge 2}(t,x,k)$. Let $C$ be the component of $G\setminus X$ that contains the lowest label not in $X$, and let $D$ be the set of vertices of all other components of $G\setminus X$. Let $G_1 = G[X\cup C]$, $G_2 = G[X\cup D]$, and $k' = |C|$. Now relabel $G_1$ by applying $\phi(C,[x+1,x+k'])$ to the labels in $C$, and relabel $G_2$ by applying $\phi(D,[x+1,x+k-k'])$ to the labels in $D$. By \cref{lemma:commute_g}, if $G\setminus X$ has exactly two connected components, then $G_2\in\widetilde G_1(t,x,k-k')$, and otherwise, $G_2\in\widetilde G_{\ge 2}(t,x,k-k')$. In either case, $G_1\in\widetilde G_1(t,x,k')$. Hence $\tilde g_{\ge 2}(t,x,k)$ is at most the summation above.

We now show that $\tilde g_{\ge 2}(t,x,k)$ is bounded below by this summation. For the first case, suppose we are given $1\le k'\le k$, a graph $G_1\in\widetilde G_1(t,x,k')$, a graph $G_2\in\widetilde G_1(t,x,k-k')$, and a subset $C\subseteq [x+1,x+k]$ of size $k'$ that contains $x+1$. Let $D = [x+1,x+k]\setminus C$. We construct a graph $G\in\widetilde G(t,x,k)$ as follows: Relabel $G_1$ by applying $\phi([x+1,x+k'],C)$ to the labels in $[x+1,x+k']$, and relabel $G_2$ by applying $\phi([x+1,x+k-k'],D)$ to the labels in $[x+1,x+k-k']$. By \cref{lemma:gluing}, we can glue $G_1$ and $G_2$ together at $X$ to obtain a chordal graph $G$. By \cref{lemma:commute_g}, we have $G\in\widetilde G_{\ge 2}(t,x,k)$. Similarly, if we begin with $G_1\in\widetilde G_1(t,x,k')$ and $G_2\in\widetilde G_{\ge 2}(t,x,k-k')$, then a similar argument shows that we can glue these together to obtain a graph $G\in\widetilde G_{\ge 2}(t,x,k)$. Therefore, $\tilde g_{\ge 2}(t,x,k)$ is at least the summation above.
\end{proof}

\f*

\begin{proof}
To see that $f(t,x,l,k)$ is at most this summation, suppose $G\in F(t,x,l,k)$. Let $L = L_G(X)$. Let $A$ be the set of vertices in components $C$ of $G\setminus(X\cup L)$ such that $C$ evaporates at time exactly $t-1$ in $G$ with exception set $X$, and let $B$ be the set of vertices of all other components of $G\setminus(X\cup L)$. Let $G_1 = G[X\cup L\cup A]$, $G_2 = G[X\cup L\cup B]$, and $k' = |A|$. Now relabel $G_1$ by applying $\phi(A,[x+l+1,x+l+k'])$ to the labels in $A$, and relabel $G_2$ by applying $\phi(B,[x+l+1,x+l+k-k'])$ to the labels in $B$. Note that $G_1$, $G_1\setminus X$, and $G_2$ are connected since $G\setminus X$ is connected and $X\cup L$ is a clique. Hence we have $G_1\in\widetilde F(t,x,l,k')$ by \cref{lemma:commute_f_1}-(1). Since $G\setminus X$ is connected, every component of $G_2\setminus(X\cup L)$ has a neighbor in $L$. Therefore, by \cref{lemma:commute_f_1}-(2), we have $G_2\in G(t-2,x+l,k-k',x)$, so $f(t,x,l,k)$ is at most the summation above.

We now show that $f(t,x,l,k)$ is bounded below by this summation. Suppose we are given $1\le k'\le k$, a graph $G_1\in\widetilde F(t,x,l,k')$, a graph $G_2\in G(t-2,x+l,k-k',x)$, and a subset $A\subseteq[x+l+1,x+l+k]$ of size $k'$. Let $B = [x+l+1,x+l+k]\setminus A$. We construct a graph $G\in G(t,x,k)$ as follows: Relabel $G_1$ by applying $\phi([x+l+1,x+l+k'],A)$ to the labels in $[x+l+1,x+l+k']$, and relabel $G_2$ by applying $\phi([x+l+1,x+l+k-k'],B)$ to the labels in $[x+l+1,x+l+k-k']$. By \cref{lemma:gluing}, we can glue $G_1$ and $G_2$ together at $[x+l]$ to obtain a chordal graph $G$. We know $G\setminus X$ is connected since $G_1\setminus X$ is connected and every component of $G_2\setminus([x+l])$ has a neighbor in $[x+1,x+l]$, so by \cref{lemma:commute_f_2}, we have $G\in F(t,x,l,k)$. Therefore, $f(t,x,l,k)$ is at least the number of possible choices for $G_1$, $G_2$, and $A$.
\end{proof}

\fTilde*

\begin{proof}
To see that $\tilde f(t,x,l,k)$ is at most this summation, suppose $G\in\widetilde F(t,x,l,k)$, and let $L = L_G(X)$. The graph $G$ falls into one of three cases: either zero, one, or at least two components of $G\setminus(X\cup L)$ see all of $X\cup L$. In the first case, we have $G\in\widetilde F_p(t,x,l,k)$. If $G$ falls into the second or third case, then let $A$ be the set of vertices in the components that see all of $X\cup L$, and let $B$ be the set of vertices of all other components of $G\setminus(X\cup L)$. Let $G_1 = G[X\cup L\cup A]$, $G_2 = G[X\cup L\cup B]$, and $k' = |A|$. Now relabel $G_1$ by applying $\phi(A,[x+l+1,x+l+k'])$ to the labels in $A$, and relabel $G_2$ by applying $\phi(B,[x+l+1,x+l+k-k'])$ to the labels in $B$. If $A$ consists of exactly one component of $G\setminus(X\cup L)$, then  by \cref{lemma:commute_f_tilde_1}, we have $G_1\in\widetilde G_1(t-1,x+l,k')$, and by \cref{lemma:commute_f_tilde_2}, we have $G_2\in\widetilde F_p(t,x,l,k-k')$. Since $G\setminus X$ is connected, every component of $G_2\setminus(X\cup L)$ has a neighbor in $L$. Therefore, if $A$ consists of at least two components, then by \cref{lemma:commute_f_tilde_1}, we have $G_1\in\widetilde G_{\ge 2}(t-1,x+l,k')$ and $G_2\in\widetilde G_p(t-1,x+l,k-k',x)$. Hence $\tilde f(t,x,l,k)$ is at most the summation above.

We now show that $\tilde f(t,x,l,k)$ is bounded below by this summation. First of all, every graph in $\widetilde F_p(t,x,l,k)$ belongs to $\widetilde F(t,x,l,k)$. For the second and third case, suppose we are given $1\le k'\le k$, a subset $A\subseteq[x+l+1,x+l+k]$ of size $k'$, and either a pair of graphs $G_1\in\widetilde G_1(t-1,x+l,k')$ and $G_2\in\widetilde F_p(t,x,l,k-k')$ or a pair of graphs $G_1\in\widetilde G_{\ge 2}(t-1,x+l,k')$ and $G_2\in\widetilde G_p(t-1,x+l,k-k',x)$. Let $B = [x+l+1,x+l+k]\setminus A$. We construct a graph $G\in\widetilde F(t,x,l,k)$ as follows: Relabel $G_1$ by applying $\phi([x+l+1,x+l+k'],A)$ to the labels in $[x+l+1,x+l+k']$, and relabel $G_2$ by applying $\phi([x+l+1,x+l+k-k'],B)$ to the labels in $[x+l+1,x+l+k-k']$. By \cref{lemma:gluing}, we can glue $G_1$ and $G_2$ together at $[x+l]$ to obtain a chordal graph $G$. If $G_1\in\widetilde G_1(t-1,x+l,k')$ and $G_2\in\widetilde F_p(t,x,l,k-k')$, then we know $G\setminus X$ is connected since $G_2\setminus X$ is connected and the one component of $G_1\setminus[x+l]$ sees all of $[x+l]$. In this case, by an argument similar to \cref{lemma:commute_f_2}, gluing together $G_1$ and $G_2$ does not change the evaporation behavior of either of those graphs, so we have $G\in\widetilde F(t,x,l,k)$. For the case when $G_1\in\widetilde G_{\ge 2}(t-1,x+l,k')$ and $G_2\in\widetilde G_p(t-1,x+l,k-k',x)$, we know $G\setminus X$ is connected since every component of $G_1\setminus[x+l]$ sees all of $[x+l]$ and every component of $G_2\setminus[x+l]$ has a neighbor in $[x+1,x+l]$. Hence we have $G\in\widetilde F(t,x,l,k)$ by \cref{lemma:commute_f_tilde_3}. Therefore, $\tilde f(t,x,l,k)$ is at least the summation above.
\end{proof}

\fTildeP*

\begin{proof}
When $z = x$ in $f_p(t,x,l,k,z)$, the definitions of $f_p(t,x,l,k)$ and $f_p(t,x,l,k,z)$ both require that $G\setminus X$ is connected, and they are otherwise identical. Hence $\tilde f_p(t,x,l,k) = \tilde f_p(t,x,l,k,x)$.
\end{proof}

\fTildePWithZ*

\begin{proof}
To see that $\tilde f_p(t,x,l,k,z)$ is at most this summation, suppose $G\in\widetilde F_p(t,x,l,k,z)$, and let $L = L_G(X)$. Let $C$ be the component of $G\setminus(X\cup L)$ that contains the lowest label not in $X\cup L$, and let $D$ be the set of vertices of all other components of $G\setminus(X\cup L)$. Let $X' = N(C)\cap X$ and $L' = N(C)\cap L$. Let $G_1 = G[X'\cup L'\cup C]$, $G_2 = G[X\cup L\cup D]$, $k' = |C|$, $x' = |X'|$, and $l' = |L'|$. Note that $X'\not\subseteq[z]$ if $l' = 0$ since $G\setminus[z]$ is connected. Now relabel $G_1$ by applying $\phi(X',[x'])$ to the labels in $X'$, applying $\phi(L',[x'+1,x'+l'])$ to the labels in $L'$, and applying $\phi(C,[x'+l'+1,x'+l'+k'])$ to the labels in $C$. Relabel $G_2$ by applying $\phi(D,[x+l+1,x+l+k-k'])$ to the labels in $D$. If $l'<l$, then we also relabel $G_2$ further by applying $\phi(L',[x+1,x+l'])$ to the labels in $L'$ and applying $\phi(L\setminus L',[x+l'+1,x+l])$ to the labels in $L\setminus L'$. By an argument similar to \cref{lemma:commute_f_tilde_1}, we have $G_1\in\widetilde G_1(t-1,x'+l',k')$; the only difference here is that the exception set is $X'\cup L'$ rather than $X\cup L$ since $C$ does not see all of $X\cup L$. We know $G_2\setminus[z]$ is connected since $G\setminus[z]$ is connected. If $l'<l$, then by \cref{lemma:commute_f_p_tilde_1}, we have $G_2\in\widetilde F_p(t,x+l',l-l',k-k',z)$. If $l' = l$, then by \cref{lemma:commute_f_tilde_1}, we have $G_2\in\widetilde G_p(t-1,x+l,k-k',z)$. Hence $\tilde f_p(t,x,l,k,z)$ is at most the summation above.

We now show that $\tilde f_p(t,x,l,k)$ is bounded below by this summation. Suppose we are given $1\le k'\le k$, $0\le x'\le x$, and $0\le l'\le l$ such that $0<x'+l'<x+l$, a graph $G_1\in\widetilde G_1(t-1,x'+l',k')$, a graph $G_2$ such that $G_2\in\widetilde F_p(t,x+l',l-l',k-k',z)$ if $l'<l$ and $G_2\in\widetilde G_p(t-1,x+l,k-k',z)$ otherwise, a subset $C\subseteq[x+l+1,x+l+k]$ of size $k'$ that contains $x+l+1$, a subset $X'\subseteq X$ of size $x'$, with the added requirement $X'\not\subseteq[z]$ if $l' = 0$, and a subset $L'\subseteq[x+1,x+l]$ of size $l'$. Let $D = [x+l+1,x+l+k]\setminus C$. We construct a graph $G\in\widetilde F_p(t,x,l,k,z)$ as follows: Relabel $G_1$ by applying $\phi([x'],X')$ to the labels in $[x']$, applying $\phi([x'+1,x'+l'],L')$ to the labels in $[x'+1,x'+l']$, and applying $\phi([x'+l'+1,x'+l'+k'],C)$ to the labels in $[x'+l'+1,x'+l'+k']$. Relabel $G_2$ by applying $\phi([x+l+1,x+l+k-k'],D)$ to the labels in $[x+l+1,x+l+k-k']$. If $l'<l$, then we also relabel $G_2$ further by applying $\phi([x+1,x+l'], L')$ to the labels in $[x+1,x+l']$ and applying $\phi([x+l'+1,x+l],[x+1,\ldots,x+l]\setminus L')$ to the labels in $[x+l'+1,x+l]$. By \cref{lemma:gluing}, we can glue $G_1$ and $G_2$ together at $X'\cup L'$ to obtain a chordal graph $G$. We know $G\setminus[z]$ is connected since $G_2\setminus[z]$ is connected and $X'\cup L'\not\subseteq[z]$. By \cref{lemma:commute_f_p_tilde_2}, we have $G\in\widetilde F_p(t,x,l,k,z)$. Statement (1) of that lemma deals with the case when $l'<l$, and statement (2) deals with the case when $l' = l$. When applying that lemma, we let $\widetilde X = X'\cup L'$. For statement (1), we let $\hat X = X\cup L'$; for statement (2), we let $\hat X = [x+l]$. Therefore, $\tilde f_p(t,x,l,k,z)$ is at least the number of possible choices for $G_1$, $G_2$, $C$, $X'$, and $L'$.
\end{proof}

\subsection{Wrapping up the proof of \cref{thm:conn_counting}}
\label{sec:wrap_up_proof}

Correctness of the counting algorithm for connected $\omega$-colorable chordal graphs follows immediately from the proofs of the recurrences above. All that remains to show is that the algorithm terminates and uses at most $O(n^7)$ arithmetic operations.

To see that the algorithm terminates, we observe that every time a cycle of recursive calls returns to the same function where the cycle began, either the value of~$t$ decreases or the number of vertices in the graph decreases. We can see this from \cref{fig:control_flow}: the only cycles where the value of~$t$ does not decrease are the black self-loops. For each these self-loops, the number of vertices in the graph decreases in each call to the function. For example, when $\tilde g$ calls itself in the recurrence for~$\tilde g$, the number of vertices in the graph decreases from $x+k$ to $x+k-k'$, for some $k'>0$. Similarly, the number of vertices in the graph decreases in the self-loops corresponding to $\tilde g_{\ge 2}$, $\tilde g_p$, and $\tilde f_p$.

To speed up the computation of the recurrence for $\tilde f_p(t,x,l,k,z)$ (\cref{lemma:f_tilde_p}), we observe that when $k'$ is fixed, $\tilde g_1(t-1,x'+l',k')$ is a common factor for all choices of $x'$ and $l'$ that have the same sum $x'+l'$. Reordering the terms of the sum and letting $r = x'+l'$, we obtain
\begin{alignat*}{3}
& && \tilde f_p(t,x,l,k,z) = && \\
& && \sum_{k'=1}^k\sum_{r=1}^{x+l-1}\binom{k-1}{k'-1}\tilde g_1(t-1,r,k')\sum_{l'=\max\{0,r-x\}}^{\min\{r,l\}}\binom{l}{l'} && \cdot
\begin{cases}
\binom{x}{r-l'} & \text{ if $l'>0$}  \\
\binom{x}{r-l'}-\binom{z}{r-l'} & \text{ otherwise}
\end{cases} \\
& && && \cdot
\begin{cases}
\tilde f_p(t,x+l',l-l',k-k',z) & \text{if $l'<l$} \\
\tilde g_p(t-1,x+l,k-k',z) & \text{otherwise}.
\end{cases}
\end{alignat*}
Now the inner sum only depends on $t$, $x$, $l$, $z$, $r$, and $k-k'$. Thus we define the helper function
\begin{align*}
h(t,x,l,z,r,k)\coloneqq\sum_{l'=\max\{0,r-x\}}^{\min\{r,l\}}\binom{l}{l'}&\cdot
\begin{cases}
\binom{x}{r-l'} & \text{ if $l'>0$}  \\
\binom{x}{r-l'}-\binom{z}{r-l'} & \text{ otherwise}
\end{cases} \\
&\cdot
\begin{cases}
\tilde f_p(t,x+l',l-l',k,z) & \text{ if $l'<l$} \\
\tilde g_p(t-1,x+l,k,z) & \text{ otherwise}.
\end{cases}
\end{align*}
The recurrence for $\tilde f_p$ can now be written as
\begin{align}\label{eqn:with_helper}
\tilde f_p(t,x,l,k,z) = \sum_{k'=1}^k\sum_{r=1}^{x+l-1}\binom{k-1}{k'-1}\tilde g_1(t-1,r,k')h(t,x,l,z,r,k-k').
\end{align}

Using \cref{eqn:with_helper}, it takes $O(n^7)$ arithmetic operations to compute all values of $\tilde f_p(t,x,l,k,z)$, since there are five arguments and now only a double summation. The cost of computing $h$ is the same, i.e., $O(n^7)$ arithmetic operations, since $h$ has six arguments and a single summation. Therefore, the overall running time of the counting algorithm is $O(n^7)$ arithmetic operations.

\section{Sampling labeled chordal graphs}
\label{sec:sampling}

Using the standard sampling-to-counting reduction of~\cite{JVV}, we can use the information stored in the dynamic-programming tables from our counting algorithm to sample a uniformly random chordal graph. This sampling algorithm is all that remains to prove \cref{thm:main}. Before sampling, assume we have already run the counting algorithm from \cref{sec:counting}, and thus we have constant-time access to the counter functions. In other words, given the name of a counter function and any particular values for its arguments, we can access the value of that function in constant time.

\begin{theorem}
\label{thm:sampling}
Suppose the counter functions have been precomputed and we have constant-time access to their values. We can sample uniformly at random from the following classes of chordal graphs, using $O(n^4)$ arithmetic operations for each sample:
\begin{enumerate}
    \item $\omega$-colorable chordal graphs with vertex set $[n]$
    \item $\omega$-colorable connected chordal graphs with vertex set $[n]$,
\end{enumerate}
where $n,\omega\in\N$ with $\omega\le n$.
\end{theorem}

\begin{proof}
First, we show that we can sample uniformly at random from the following classes of chordal graphs: $G(t,x,k,z)$, $\widetilde G(t,x,k,z)$, $\widetilde G_p(t,x,k,z)$, $\widetilde G_1(t,x,k)$, $\widetilde G_{\ge 2}(t,x,k)$, $F(t,x,l,k)$, $\widetilde F(t,x,l,k)$, $\widetilde F_p(t,x,l,k)$, and $\widetilde F_p(t,x,l,k,z)$, where $t,x,l,k,z$ are any values in the domains of their respective functions. In the ``greater than or equal'' direction of the proof each recurrence in \cref{sec:recurrence_proofs}, we describe an injective map. For example, in \cref{lemma:g} we are given $0\le k'\le k$, a graph $G_1\in\widetilde G(t,x,k',z)$, a graph $G_2\in G(t-1,x,k-k',z)$, and a subset $A\subseteq[x+1,x+k]$ of size $k'$, and we describe how one can map this to a graph $G\in G(t,x,k,z)$. In each lemma, we also describe an injective map in the reverse direction (to prove the ``less than or equal'' direction). Therefore, each of these injective maps is in fact a bijection.

Consider any one of the recurrence lemmas proved in \cref{sec:recurrence_proofs}. Let $\psi$ be the injective map in the ``greater than or equal'' direction of the proof, let $\mathcal{C}_1$ be the domain of $\psi$ and let $\mathcal{C}_2$ be the target space. Since $\psi\colon\mathcal{C}_1\to\mathcal{C}_2$ is in fact a bijection, and our goal is to sample a uniformly random element of $\mathcal{C}_2$ (e.g. $G(t,x,k,z)$), it is sufficient to first sample a uniformly random element of $\mathcal{C}_1$, and then apply $\psi$.

The samplers for each of these graph classes call one another recursively/cyclically until reaching a base case. The base cases are the same as in the counting algorithm. Whenever we reach a base case, there is only one possible graph (since the counter function equals 1), so we simply output that graph. We never choose a base case where the counter function equals 0, since the probability of taking such a path in the algorithm would have been 0.

To sample a \emph{connected} chordal graph with vertex set $[n]$, we first choose a random value of the evaporation time, where each possible time $t\in[n]$ has probability proportional to the weight of the corresponding term in the summation for $c(n)$ (see section \cref{sec:recurrences}). Next, we sample a random graph with the chosen evaporation time $t$ using the sampler for $\widetilde G_1(t,0,n)$.

At the highest level, to sample a chordal graph with vertex set $[n]$, we use the standard sampling algorithm that is derived from the injective map in the ``greater than or equal'' direction of the proof \cref{lemma:disconn}.

One can also easily modify any of these samplers to only ouptut $\omega$-colorable chordal graphs. This does not explicity change the pseudocode at all (see \cref{sec:pseudocode}); we simply adjust the values of some of the base cases of the counting algorithm that are used by the sampling algorithm.

The correctness of each of these sampling procedures follows from the correctness of the counting algorithm. The sampling algorithm terminates by an argument similar to the proof of termination for the counting algorithm. For the sampling running time, the bottleneck is the time spent sampling from $\widetilde F_p(t,x,l,k,z)$. In this procedure, we compute the triple summation from the recurrence for $\widetilde F_p(t,x,l,k,z)$, and we call this procedure at most $O(n)$ times, so the overall cost is $O(n^4)$ arithmetic operations per sample.
\end{proof}


This completes the proof of \cref{thm:main}. The pseudocode for the sampling algorithm can be found in \cref{sec:pseudocode}.

\section{Approximate counting and sampling of labeled chordal graphs}
\label{sec:approx}

A \emph{split} graph is a graph whose vertex set can be partitioned into a clique and an independent set, with arbitrary edges between the two parts. It is easy to see that every split graph is chordal. On the other hand, a result by Bender, Richmond, and Wormald~\cite{bender1985almost} shows that a random $n$-vertex labeled chordal graph is a split graph with probability $1-o(1)$. There is a relatively simple algorithm that counts the number of labeled split graphs on $n$ vertices using $O(n^2)$ arithmetic operations~\cite{bina2015note}. In fact, we show that this algorithm can be modified to give an approximate answer using only $O(n\log^2(1/\eps))$ arithmetic operations. By combining this with our (exact) counting and sampling algorithms for labeled chordal graphs, we obtain efficient algorithms for approximate counting and approximate sampling of labeled chordal graphs. In particular, if $n$ is large enough (roughly $\log(1/\eps)$ or larger), then we simply run the split graph counting or sampling algorithm. Otherwise, we run the exact counting or uniform sampling algorithm for chordal graphs.

For $\eps>0$, we say a number $A$ is a $(1+\eps)$-approximation of a number $B$ if $B\le A\le B(1+\eps)$, and we say $A$ is a $(1-\eps)$-approximation of $B$ if $B(1-\eps)\le A\le B$. Similarly, we say $A$ is a $(1\pm\eps)$-approximation of $B$ if $B(1-\eps)\le A\le B(1+\eps)$.

\begin{theorem}
\label{thm:approx}
There is a deterministic algorithm that given a positive integer $n$ and $\eps>0$ as input, computes a $(1\pm\eps)$-approximation of the number of $n$-vertex labeled chordal graphs in $O(n^3\log{n}\log^7(1/\eps))$ time. Moreover, there is a randomized algorithm that given the same input, generates a random $n$-vertex labeled chordal graph according to a distribution whose total variation distance from the uniform distribution is at most $\eps$. The expected running time of this randomized algorithm is $O(n^3\log{n}\log^7(1/\eps))$.
\end{theorem}

In the $O(n^2)$ algorithm for counting split graphs that is described in~\cite{bina2015note}, there is initially a certain amount of overcounting because of the fact that some split graphs have more than one valid way of being partitioned into a clique and an independent set. To account for this, the algorithm in~\cite{bina2015note} first computes the total number of split graphs \emph{with overcounting}, and then it subtracts the value of the redundant terms to obtain the final answer. This gives an efficient counting algorithm, but this strategy does not lend itself as well to designing a sampling algorithm.

For \cref{thm:approx}, we wish to obtain not only an approximate \emph{counting} algorithm, but also an approximate \emph{sampling} algorithm. Therefore, we begin by designing a related, alternative counting algorithm for split graphs that partitions each split graph into three parts rather than two. In this three-part partition, the third set will consist of the vertices that sometimes belong to the clique and sometimes belong to the independent set (depending on the choice of the clique / independent set in the partition), which means each split graph will have a unique three-part partition.

For our purposes, it is sufficient to \emph{approximately} count split graphs, so the algorithm that we begin with is a $(1+\eps)$-approximate counting algorithm that uses $O(n^2)$ arithmetic operations when $\eps$ is fixed (see \cref{sec:n_squared}). In \cref{sec:faster}, we then speed this up by computing only the largest terms of the summation given in \cref{sec:n_squared}, to obtain an approximate counting algorithm for split graphs that uses only $O(n\log^2(1/\eps))$ arithmetic operations. In \cref{sec:split_sampling}, we design an approximate sampling algorithm for split graphs. Finally, in \cref{sec:approx_proof}, we describe the resulting approximate counting and sampling algorithms for \emph{chordal} graphs, which proves \cref{thm:approx}.

\subsection{Approximate counting of split graphs}
\label{sec:n_squared}

For convenience, we define the term \emph{partition} to allow empty parts. In other words, a partition of a set $S$ is a collection of disjoint sets whose union is $S$. We say a partition of the vertex set of a split graph $G$ is a \emph{split partition} if it partitions $G$ into a clique and an independent set; i.e., a split partition is a partition that demonstrates that $G$ is split. We denote a split partition that consists of the clique $C'$ and the independent set $I'$ by the ordered pair $(C',I')$.

\begin{definition}
Let $G$ be a split graph. The \emph{always-clique set} of $G$ is the subset of $V(G)$ consisting of vertices that belong to the clique in every split partition of $G$. The \emph{always-independent set} of $G$ is the set of vertices that belong to the independent set in every split partition of $G$. The \emph{questioning set} of $G$ is the set of vertices $v$ such that there exists a split partition that places $v$ in the clique as well as a split partition that places $v$ in the independent set.
\end{definition}

As an example, if $G$ is a complete graph, then all of $V(G)$ belongs to the questioning set. From now on, let $C$, $I$, and $Q$ denote the always-clique set, the always-independent set, and the questioning set, respectively, in a generic split graph.

\begin{observation}
\label{obs:universal_q}
For every split graph $G$, $C$ is a clique and $I$ is an independent set. Furthermore, every vertex in $Q$ is adjacent to every vertex in $C$ and is non-adjacent to every vertex in $I$.
\end{observation}

\begin{proof}
Clearly, $C$ is a clique and $I$ is an independent set, since $C\subseteq C'$ and $I\subseteq I'$ for every split partition $(C',I')$ of $G$. For the adjacencies of $Q$, let $v\in Q$. There exists some split partition $(C',I')$ of $G$ that places $v$ in the clique $C'$. Since $C\subseteq C'$, this shows that $v$ is adjacent to every vertex in $C$. Similarly, $v$ is non-adjacent to every vertex in $I$.
\end{proof}

In the proof of the next lemma, $P_3$ denotes a 3-vertex path, and $\overline{P_3}$ denotes the complement of a 3-vertex path. We say a graph is $P_3$-free (resp.\ $\overline{P_3}$-free) if it does not contain $P_3$ (resp.\ $\overline{P_3}$) as an induced subgraph.

\begin{lemma}
\label{lemma:clique_or_indep}
The questioning set $Q$ of a split graph $G$ is either a clique or an independent set.
\end{lemma}

\begin{proof}
We first observe that $G[Q]$ is $P_3$-free. To see this, suppose $G[Q]$ contains some induced path of $u,v,w$ of length 3. Since $v\in Q$, there exists some split partition $(C',I')$ of $G$ such that $v\in I'$. This means $u\in C'$ and $w\in I'$ since $(u,v)\in E(G)$ and $(u,w)\notin E(G)$. Therefore, $v$ and $w$ are adjacent and both belong to $I'$, which is a contradiction.

This shows that $G[Q]$ is $P_3$-free, so every connected component of $G[Q]$ is a clique. If $Q$ is empty or $G[Q]$ has only one component, we are done. Now we claim that if $G[Q]$ has at least two components, then $Q$ is an independent set. Suppose to the contrary that some component of $G[Q]$ contains an edge $(u,v)$, and let $w$ be a vertex in some other component. The vertices $u,w,v$ form an induced $\overline{P_3}$ in $G[Q]$. However, we can see that $G[Q]$ is $\overline{P_3}$-free by an argument similar to the previous paragraph with the roles of $C'$ and $I'$ reversed, so this is a contradiction.
\end{proof}

\begin{lemma}
\label{lemma:neighbor_exists}
In a split graph, if $Q$ is a clique, then every vertex in $C$ has a neighbor in $I$. If $Q$ is an independent set, then every vertex in $I$ has a non-neighbor in $C$.
\end{lemma}

\begin{proof}
For the case when $Q$ is a clique, suppose there is some vertex $v\in C$ that is not adjacent to any vertex in $I$. Then $(C\cup Q\setminus\{v\},I\cup\{v\})$ is a split partition that places $v$ in the independent set, contradicting the fact that $v\in C$. Therefore, every vertex in $C$ has a neighbor in $I$. The proof for the case when $Q$ is an independent set is similar.
\end{proof}

For the case when $|Q|\ge 2$, it is easy to exactly compute the number of split graphs. To do so, we will count the number of split graphs that satisfy the properties described in the following lemma.

\begin{lemma}
\label{lemma:c_i_q}
Suppose $G$ is a graph, and suppose $\{\hat C,\hat I,\hat Q\}$ is a partition of $V(G)$ with the following properties:
\begin{enumerate}
    \item $\hat C$ is a clique and $\hat I$ is an independent set
    \item $\hat Q$ is a clique or independent set, and $|\hat Q|\ge 2$
    \item Every vertex in $\hat Q$ is adjacent to every vertex in $\hat C$ and is non-adjacent to every vertex in $\hat I$
    \item If $\hat Q$ is a clique, then every vertex in $\hat C$ has a neighbor in $\hat I$
    \item If $\hat Q$ is an independent set, then every vertex in $\hat I$ has a non-neighbor in $\hat C$.
\end{enumerate}
Then $G$ is a split graph, and furthermore, $\hat C$ is the always-clique set, $\hat I$ is the always-independent set, and $\hat Q$ is the questioning set of $G$.
\end{lemma}

We have already established that the ``converse'' of this lemma is true: If $G$ is a split graph with always-clique set, always-independent set, and questioning set equal to $C$, $I$, and $Q$, with $|Q|\ge 2$, then these three sets satisfy properties (1)-(5) of \cref{lemma:c_i_q} with $\hat C = C$, $\hat I = I$, and $\hat Q = Q$. This was proved in \cref{obs:universal_q,lemma:clique_or_indep,lemma:neighbor_exists}. Combined with \cref{lemma:c_i_q}, this tells us that in order to count split graphs with $|Q|\ge 2$, we simply need to count the number of split graphs (for each given $C$, $I$, and $Q$) that satisfy properties (1)-(5).

\begin{proof}[Proof of \cref{lemma:c_i_q}]
Either $(\hat C\cup\hat Q,\hat I)$ or $(\hat C,\hat I\cup\hat Q)$ is a split partition of $G$, so $G$ is a clearly a split graph. As usual, let $C$, $I$, and $Q$ denote the always-clique set, always-independent set, and questioning set of $G$. From now on, we assume $\hat Q$ is a clique. For the case when $\hat Q$ is an independent set, the argument is similar.

We first show $\hat Q\subseteq Q$. Let $v\in\hat Q$. The split partition $(\hat C\cup\hat Q,\hat I)$ places $v$ in the clique, and the split partition $(\hat C\cup\hat Q\setminus\{v\},\hat I\cup\{v\})$ places $v$ in the independent set. Therefore, $v\in Q$.

Next, we show $\hat C\subseteq C$. Suppose there is some $v\in\hat C\setminus C$. This means there is some split partition $(C',I')$ of $G$ such that $v\in I'$. By property (4), $v$ has some neighbor $u\in\hat I$, and we must have $u\in C'$. Now choose a vertex $q\in\hat Q\subseteq Q$. By property (3), we have $(v,q)\in E(G)$ and $(u,q)\notin E(G)$. This means $q$ cannot belong to $I'$ nor $C'$, which is a contradiction.

To see that $\hat I\subseteq I$, let $(C',I')$ be a split partition of $G$. Since $\hat Q$ is a clique and $|\hat Q|\ge 2$, there is at least one vertex $\hat q\in\hat Q$ that belongs to $C'$. For every vertex $v\in\hat I$, we know $v$ is not adjacent to $\hat q$, so $v\in I'$. Therefore, $\hat I\subseteq I$. Since $\hat C$, $\hat I$, and $\hat Q$ together cover all of $V(G)$, this shows that $C = \hat C$, $I = \hat I$, and $Q = \hat Q$.
\end{proof}

Putting this together, we have a formula that (exactly) counts split graphs with $|Q|\ge 2$.

\begin{lemma}
\label{lemma:q_is_2}
Let $n\in\N$. The number of split graphs on $n$ vertices with $|Q|\ge 2$ is given by the formula\footnote{We use the convention $0^0 = 1$, so the term where $q = n$ and $c = 0$ is equal to 1.} $$2\sum_{q=2}^n\sum_{c=0}^{n-q}\binom{n}{q}\binom{n-q}{c}(2^{n-c-q}-1)^c.$$
\end{lemma}

\begin{proof}
By \cref{lemma:clique_or_indep}, $Q$ is either a clique or an independent set. First consider the case when $Q$ is a clique. Let $q$ and $c$ be integers such that $2\le q\le n$ and $0\le q\le n-q$. We wish to count the number of split graphs with $|Q| = q$, $|C| = c$, and $|I| = n-q-c$. There are $\binom{n}{q}\binom{n-q}{c}$ ways of assigning the labels $1,\ldots,n$ to the sets $C$, $I$, and $Q$, since once we have chosen the labels for $Q$ and $C$, the labels in $I$ are known. For each possible choice of label sets, there are $(2^{n-c-q}-1)^c$ split graphs of this form, since we must choose a nonempty subset of $I$ be to the neighborhood for each vertex in $C$. By \cref{obs:universal_q,lemma:clique_or_indep,lemma:neighbor_exists,lemma:c_i_q}, this counts exactly the correct set of graphs. Therefore, the number of split graphs on $n$ vertices in which $Q$ is a clique of size at least 2 is equal to $$\sum_{q=2}^n\sum_{c=0}^{n-q}\binom{n}{q}\binom{n-q}{c}(2^{n-c-q}-1)^c.$$

By symmetry, this is exactly equal to the number of split graphs on $n$ vertices in which $Q$ is an independent set of size at least 2 (the bijection is given by taking the complement of each graph). Therefore, the number of split graphs with $|Q|\ge 2$ is twice the above summation.
\end{proof}

We now move to the case when $|Q|\le 1$.

\begin{lemma}
\label{lemma:neighbor_exists_2}
Consider a split graph with $|Q|\le 1$. Every vertex in $C$ has a neighbor in $I$, and every vertex in $I$ has a non-neighbor in $C$.
\end{lemma}

\begin{proof}
This follows immediately from \cref{lemma:neighbor_exists}.
\end{proof}

We also need a lemma that is analogous to \cref{lemma:c_i_q}.

\begin{lemma}
\label{lemma:c_i_q_2}
Suppose $G$ is a graph, and suppose $\{\hat C,\hat I,\hat Q\}$ is a partition of $V(G)$ with the following properties:
\begin{enumerate}
    \item $\hat C$ is a clique and $\hat I$ is an independent set
    \item $|\hat Q|\le 1$
    \item Every vertex in $\hat Q$ is adjacent to every vertex in $\hat C$ and is non-adjacent to every vertex in $\hat I$
    \item Every vertex in $\hat C$ has a neighbor in $\hat I$, and every vertex in $\hat I$ has a non-neighbor in $\hat C$.
\end{enumerate}
Then $G$ is a split graph, and furthermore, $\hat C$ is the always-clique set, $\hat I$ is the always-independent set, and $\hat Q$ is the questioning set of $G$.
\end{lemma}

\begin{proof}
The graph $G$ is clearly a split graph since $(\hat C\cup\hat Q,\hat I)$ is a split partition of $G$. We also have the split partition $(\hat C,\hat I\cup\hat Q)$, so $\hat Q\subseteq Q$.

To see that $\hat C\subseteq C$, suppose there is some $v\in\hat C\setminus C$. This means there is some split partition $(C',I')$ of $G$ such that $v\in I'$. By property (4), $v$ has some neighbor $u\in\hat I$, and we must have $u\in C'$. And again by property (4), $u$ has some non-neighbor $w\in\hat C$, and we must have $w\in I'$. Since $v,w\in\hat C$, this means $v$ and $w$ are two adjacent vertices in $I'$, which is a contradiction. The proof that $\hat I\subseteq I$ is similar. Therefore, $C = \hat C$, $I = \hat I$, and $Q = \hat Q$.
\end{proof}

For the proof of the following lemma, a \emph{two-tone} graph is defined as a labeled graph in which, in addition to the vertex labels, each vertex is colored with one of two colors (cyan and indigo, for $C$ and $I$). A \emph{three-tone} graph is defined similarly, but with three colors (cyan, indigo, and white, for $C$, $I$, and $Q$).

Unlike the case when $|Q|\ge 2$, property (4) of \cref{lemma:c_i_q_2} now has two neighbor/non-neighbor conditions that need to be true simultaneously, so it is no longer trivial to exactly count the number of possibilities. However, we can compute a very close approximation.

\begin{lemma}
\label{lemma:q_is_empty}
If $n$ is sufficiently large, then the following formula gives a $\left(1+\frac{1}{(3/2)^{n/3}}\right)$-approximation of the number of split graphs on $n$ vertices with $Q = \emptyset$: $$\sum_{c=2}^{\lfloor\frac{n}{2}\rfloor}\binom{n}{c}(2^c-1)^{n-c}+\sum_{c=\lfloor\frac{n}{2}\rfloor+1}^{n-2}\binom{n}{c}(2^{n-c}-1)^c.$$
\end{lemma}

\begin{proof}
We begin by working with the first summation. Fix an integer $k$ such that $2\le k\le\frac{n}{2}$. Let $\mathcal T$ be the set of all two-tone graphs with $k$ cyan vertices and $n-k$ indigo vertices such that the cyan vertices form a clique, the indigo vertices form an independent set, and every indigo vertex has a non-neighbor among the cyan vertices. The term with $c = k$ in the above summation counts the number of two-tone graphs in $\mathcal T$.

By \cref{lemma:neighbor_exists_2,lemma:c_i_q_2}, the number of $n$-vertex split graphs with $|C|= k$ and $Q = \emptyset$ is equal to number of two-tone graphs in $\mathcal T$ where it happens that every cyan vertex has an indigo neighbor. Note that there are no split graphs with $|C|\in\{0,1\}$ and $Q = \emptyset$ since if all of $G$ is an independent set, then $Q = V(G)$. Now we just need to show that $|\mathcal T|$ is approximately the number of split graphs that we wish to count. Suppose we choose a graph $G'$ from $\mathcal T$ uniformly at random. If $u\in V(G')$ is a cyan vertex and $v\in V(G')$ is an indigo vertex, then the probability that $u$ and $v$ are adjacent in $G'$ is $\frac{2^{k-1}-1}{2^k-1}$. Therefore, the probability that a given cyan vertex $u$ does not have an indigo neighbor is $$\left(1-\frac{2^{k-1}-1}{2^k-1}\right)^{n-k}\le\left(\frac{2}{3}\right)^{n-k},$$ since $k\ge 2$. By a union bound, the probability that some cyan vertex $u$ does not have an indigo neighbor is at most $n(2/3)^{n-k}\le n(2/3)^{n/2}$ since $k\le\frac{n}{2}$.

Therefore, the value of the summation from $c = 2$ to $\lfloor\frac{n}{2}\rfloor$ is at most $\alpha$ times the number of split graphs with $Q = \emptyset$ and $|C|\le|I|$, where $\alpha$ is the following value:
\begin{align*}
    \alpha &= \frac{1}{1-n\left(\frac{2}{3}\right)^{n/2}} \\
    &= 1+\frac{n}{\left(\frac{3}{2}\right)^{n/2}-n} \\
    &\le 1+\frac{1}{\left(\frac{3}{2}\right)^{n/3}}.
\end{align*}
The last inequality holds when $n$ is sufficiently large.\footnote{A loose bound shows that this lemma holds for $n\ge 61$.}

Similarly, the value of the summation from $c = \lfloor\frac{n}{2}\rfloor+1$ to $n-2$ is at most $\alpha$ times the number of split graphs with $Q = \emptyset$ and $|C|>|I|$. The proof of this is similar, except $\mathcal T$ is defined as the set of all two-tone graphs with $k$ cyan vertices and $n-k$ indigo vertices such that the cyan vertices form a clique, the indigo vertices form an independent set, and \emph{every cyan vertex has a neighbor among the indigo vertices}. Therefore, the overall sum is an $\alpha$-approximation of the number of split graphs with $Q = \emptyset$.
\end{proof}

\begin{lemma}
\label{lemma:q_is_1}
If $n$ is sufficiently large, then the following formula gives a $\left(1+\frac{1}{(3/2)^{n/3}}\right)$-approximation of the number of split graphs on $n$ vertices with $|Q| = 1$: $$\sum_{c=2}^{\lfloor\frac{n-1}{2}\rfloor}n\binom{n-1}{c}(2^c-1)^{n-c-1}+\sum_{c=\lfloor\frac{n-1}{2}\rfloor+1}^{n-2}n\binom{n-1}{c}(2^{n-c-1}-1)^c.$$
\end{lemma}

\begin{proof}
For the first summation, fix an integer $k$ such that $2\le k\le\frac{n-1}{2}$. Let $\mathcal T$ be the set of all three-tone graphs with $k$ cyan vertices, $n-k-1$ indigo vertices, and one white vertex, with the following properties: the cyan vertices form a clique, the indigo vertices form an independent set, the white vertex is adjacent to all of the cyan vertices and none of the indigo vertices, and every indigo vertex has a non-neighbor among the cyan vertices. The term with $c = k$ in the above summation counts the number of three-tone graphs in $\mathcal T$. There is now an additional factor of $n$ (compared to the previous lemma) since there are $n$ possible choices for the label of the white vertex.

By an argument similar to the previous lemma, the value of the summation from $c = 2$ to $\lfloor\frac{n-1}{2}\rfloor$ is at most $\alpha$ times the number of split graphs with $|Q| = 1$ and $|C|\le|I|$, where $\alpha$ is the following value: $$\alpha = \frac{1}{1-(n-1)\left(\frac{2}{3}\right)^{(n-1)/2}}.$$ This is the same value as in the previous lemma, except with $n-1$ in place of $n$. As before, this is bounded above\footnote{A loose bound shows that this lemma holds for $n\ge 65$.} by $1+\frac{1}{(3/2)^{n/3}}$, and the case when $|C|>|I|$ is similar.
\end{proof}

Suppose $n$ is large enough that \cref{lemma:q_is_empty,lemma:q_is_1} hold. Computing and adding together the summations from \cref{lemma:q_is_2,lemma:q_is_empty,lemma:q_is_1} gives a $\left(1+\frac{1}{(3/2)^{n/3}}\right)$-approximation for the number of split graphs on $n$ vertices. If $n\ge 3\log_{3/2}\frac{1}{\eps}$, then this is a $(1+\eps)$-approximation.

The summation from \cref{lemma:q_is_2} can be computed using $O(n^2)$ arithmetic operations by precomputing all binomial coefficients and precomputing all possible values of $(2^{n-c-q}-1)^c$. Similarly, the summations from \cref{lemma:q_is_empty,lemma:q_is_1} can also be computed using $O(n^2)$ operations. Therefore, we have a $(1+\eps)$-approximation algorithm for counting split graphs on $n$ vertices that uses $O(n^2)$ arithmetic operations, for $n\ge\max\{N_1,3\log_{3/2}\frac{1}{\eps}\}$, where $N_1$ is a number large enough that \cref{lemma:q_is_empty,lemma:q_is_1} hold for $n\ge N_1$.

This counting algorithm is already well-suited for designing a corresponding sampling algorithm. However, the running time still has some room for improvement. In the next section, we describe how to speed up this algorithm by only computing the terms that dominate the value of each summation, rather than computing all $\Theta(n^2)$ terms.

\subsection{Faster approximate counting of split graphs}
\label{sec:faster}

For the case when $|Q|\ge 2$, the term $(2^{n-c-q}-1)^c$ in \cref{lemma:q_is_2} is maximized when $q = 2$ and $c = \lfloor\frac{n-2}{2}\rfloor$. As it turns out, if we wish to compute a $(1-\eps)$-approximation of the number of such split graphs, it is sufficient to compute only $O(\log^2(1/\eps))$ terms of the double summation, near the values $q = 2$ and $c = \lfloor\frac{n-2}{2}\rfloor$. In the following lemma, we show that we only need to compute $O(\log(1/\eps))$ terms of the sum over $q$. Next, in \cref{lemma:sum_over_c}, we show that we only need to compute $O(\log(1/\eps))$ terms of the sum over $c$. 

In the next few lemmas we assume $|Q|<n$, since this will be useful in the proof of \cref{lemma:sum_over_c}. For the case when $|Q| = n$, we can easily count the number of split graphs. There are just two of them: $G$ is either a complete graph or an independent set. (\cref{lemma:sum_over_q} also holds for the number of $n$-vertex split graphs with $2\le|Q|\le n$, with the same proof, if we take the sum from $q = 2$ to $\min\{s,n\}$, but requiring $|Q|<n$ is more convenient for the later lemmas.)

\begin{lemma}
\label{lemma:sum_over_q}
Let $0<\eps<1$ be given. For sufficiently large $n$, the following formula gives a $(1-\eps)$-approximation of the number of split graphs on $n$ vertices with $2\le|Q|<n$: $$2\sum_{q=2}^{\min\{s,n-1\}}\sum_{c=0}^{n-q}\binom{n}{q}\binom{n-q}{c}(2^{n-c-q}-1)^c,$$ where $s = \lceil 10\log\frac{1}{\eps}+37\rceil$.
\end{lemma}

\begin{proof}
The number of $n$-vertex split graphs with $2\le|Q|<n$ is equal to the formula from \cref{lemma:q_is_2}, except with the sum over $q$ going from 2 to $n-1$ rather than $n$. Now, in the current lemma, we have skipped the terms that go from $q = s+1$ to $n-1$. We need to show that their sum is at most $\eps$ times the total summation (from 2 to $n-1$).

For simplicity, in this proof we will not write the factor of 2 in front of the double summation, since including that factor does not change the multiplicative error. In the following inequality, we use the fact that $\lceil\frac{n-2}{2}\rceil\lfloor\frac{n-2}{2}\rfloor\ge\lceil\frac{n-2}{2}\rceil(\lfloor\frac{n-2}{2}\rfloor-1)+1$ since we can assume $\lceil\frac{n-2}{2}\rceil\ge 1$. The total value of the double summation is bounded below by the term with $q = 2$ and $c = \lfloor\frac{n-2}{2}\rfloor$:
\begin{align*}
    \binom{n}{2}\binom{n-2}{\lfloor\frac{n-2}{2}\rfloor}(2^{\lceil\frac{n-2}{2}\rceil}-1)^{\lfloor\frac{n-2}{2}\rfloor}&\ge\binom{n-2}{\lfloor\frac{n-2}{2}\rfloor}\left(2^{\lceil\frac{n-2}{2}\rceil\lfloor\frac{n-2}{2}\rfloor}-2^{\lceil\frac{n-2}{2}\rceil(\lfloor\frac{n-2}{2}\rfloor-1)}\right) \\
    &\ge\frac{1}{2}\binom{n-2}{\lfloor\frac{n-2}{2}\rfloor}2^{\lceil\frac{n-2}{2}\rceil\lfloor\frac{n-2}{2}\rfloor} \\
    &\ge\frac{1}{2}\binom{n-2}{\lfloor\frac{n-2}{2}\rfloor}2^{(\frac{n-3}{2})(\frac{n-2}{2})}.
\end{align*}

To prove a bound on the tail (the terms from $q = s+1$ to $n-1$), we first give a bound for the inner sum. Suppose $2\le q<n$ is fixed. As $c$ ranges from 0 to $\lfloor\frac{n-q}{2}\rfloor$, the product $c(n-q-c)$ increases by at least 1 whenever $c$ increases by 1. Therefore, $$\sum_{c=0}^{\lfloor\frac{n-q}{2}\rfloor}(2^{n-c-q}-1)^c\le\sum_{c=0}^{\lfloor\frac{n-q}{2}\rfloor}(2^{n-c-q})^c\le\sum_{i=0}^{\lfloor\frac{n-q}{2}\rfloor\lceil\frac{n-q}{2}\rceil} 2^i\le 2^{(n-q)^2/4+1}.$$ Similarly, the same bound holds if we take the same sum from $c = \lceil\frac{n-q}{2}\rceil$ to $n-q$, so $$\sum_{c=0}^{n-q}(2^{n-c-q}-1)^c\le 2^{(n-q)^2/4+2}.$$ Since $\binom{n-q}{c}$ is at most $\binom{n-2}{\lfloor\frac{n-2}{2}\rfloor}$ for all values of $c$ and $q$, this implies $$\sum_{c=0}^{n-q}\binom{n-q}{c}(2^{n-c-q}-1)^c\le\binom{n-2}{\lfloor\frac{n-2}{2}\rfloor}\cdot 2^{(n-q)^2/4+2}.$$ Therefore, we have the following bound on the tail:
\begin{align*}
    \sum_{q=s+1}^{n-1}\sum_{c=0}^{n-q}\binom{n}{q}\binom{n-q}{c}(2^{n-c-q}-1)^c&\le\sum_{q=s+1}^{n-1}\binom{n}{q}\binom{n-2}{\lfloor\frac{n-2}{2}\rfloor}\cdot2^{(n-q)^2/4+2} \\
    &\le\binom{n-2}{\lfloor\frac{n-2}{2}\rfloor}\cdot 2^{(\frac{n-3}{2})(\frac{n-2}{2})-1}\sum_{q=s+1}^{n-1}\binom{n}{q}2^{-qn/4+5n/4+3/2}
\end{align*}
since $q^2\le qn$. We just need to show $$\sum_{q=s+1}^{n-1}\binom{n}{q}2^{-qn/4+5n/4+3/2}\le\eps$$ to conclude that we have a $(1-\eps)$-approximation. Note that $\frac{n}{2^{n/4}}\le\frac{1}{2^{n/10}}$ for $n\ge 34$. Therefore, we can bound this sum as follows:
\begin{align*}
    \sum_{q=s+1}^{n-1}\binom{n}{q}2^{-qn/4+5n/4+3/2}&\le \sum_{q=s+1}^{n-1}n^q\cdot 2^{-qn/4+5n/4+3/2} \\
    &= 2^{5n/4+3/2}\sum_{q=s+1}^{n-1}\bigg(\frac{n}{2^{n/4}}\bigg)^q \\
    &\le 2^{5n/4+3/2}\sum_{q=s+1}^{n-1}\bigg(\frac{1}{2^{n/10}}\bigg)^q \\
    &\le 2^{5n/4+3/2}\frac{2}{2^{(s+1)n/10}} \\
    &= 2^{5n/4-(s+1)n/10+5/2} \\
    &\le\eps,
\end{align*}
where the last step follows from the definition of $s$ (see the next paragraph for more details). In the third-to-last step, we assumed $n\ge 10$ so that $2^{qn/10}$ increases by at least a factor of 2 each time $q$ increases by 1. 

The last step deserves more explanation since there are several steps contained in it. Starting from the definition of $s$, we have
\begin{align*}
    &s = \lceil 10\log\frac{1}{\eps}+37\rceil \\
    &\implies \log\frac{1}{\eps}+\frac{5}{2}\le-\frac{5}{4}+\frac{s}{10}+\frac{1}{10} \\
    &\implies \log\frac{1}{\eps}+\frac{5}{2}\le\left(-\frac{5}{4}+\frac{s}{10}+\frac{1}{10}\right)n \\
    &\implies \log\frac{1}{\eps}\le-\frac{5n}{4}+\frac{(s+1)n}{10}-\frac{5}{2}.
\end{align*}
In the second implication, where we use the fact that $1\le n$, we also need $-\frac{5}{4}+\frac{s}{10}+\frac{1}{10}\ge 0$ for this to work. This is true by the definition of $s$ (since $\eps\le 5$).

\end{proof}

\begin{lemma}
\label{lemma:sum_over_c}
Let $n,q\in\N$ with $2\le q<n$, and let $0<\eps<1$ be given. We have the following approximation: $$\sum_{c=\max\{0,t_1\}}^{\min\{t_2,n-q\}}\binom{n-q}{c}(2^{n-c-q}-1)^c\ge(1-\eps)\sum_{c=0}^{n-q}\binom{n-q}{c}(2^{n-c-q}-1)^c,$$ where $t_1 = \lfloor\frac{n-q}{2}\rfloor-\lceil\log\frac{1}{\eps}\rceil-2$ and $t_2 = \lceil\frac{n-q}{2}\rceil+\lceil\log\frac{1}{\eps}\rceil+2$.
\end{lemma}

\begin{proof}
As before, we show that the tail composed of terms we have skipped is at most $\eps$ times the total summation. In the following inequality, we use the fact that $\lceil\frac{n-q}{2}\rceil\lfloor\frac{n-q}{2}\rfloor\ge\lceil\frac{n-q}{2}\rceil(\lfloor\frac{n-q}{2}\rfloor-1)+1$ since $q<n$. The total value of the summation is bounded below by the term with $c = \lfloor\frac{n-q}{2}\rfloor$:
\begin{align*}
    \binom{n-q}{\lfloor\frac{n-q}{2}\rfloor}(2^{\lceil\frac{n-q}{2}\rceil}-1)^{\lfloor\frac{n-q}{2}\rfloor} &\ge\binom{n-q}{\lfloor\frac{n-q}{2}\rfloor}\left(2^{\lceil\frac{n-q}{2}\rceil\lfloor\frac{n-q}{2}\rfloor}-2^{\lceil\frac{n-q}{2}\rceil(\lfloor\frac{n-q}{2}\rfloor-1)}\right) \\
    &\ge\frac{1}{2}\binom{n-q}{\lfloor\frac{n-q}{2}\rfloor}2^{\lceil\frac{n-q}{2}\rceil\lfloor\frac{n-q}{2}\rfloor}.
\end{align*}

For the bound on the tail, we start with the upper end of the tail where $c>t_2$. We have
\begin{align*}
    \sum_{c=t_2+1}^{n-q}\binom{n-q}{c}(2^{n-c-q}-1)^c &= \sum_{\Delta=t_2+1-\lceil\frac{n-q}{2}\rceil}^{\lfloor\frac{n-q}{2}\rfloor}\binom{n-q}{\lceil\frac{n-q}{2}\rceil+\Delta}(2^{\lfloor\frac{n-q}{2}\rfloor-\Delta}-1)^{\lceil\frac{n-q}{2}\rceil+\Delta},
\end{align*}
where $\Delta = c-\lceil\frac{n-q}{2}\rceil$. We have $(\lfloor\frac{n-q}{2}\rfloor-\Delta)(\lceil\frac{n-q}{2}\rceil+\Delta)\le\lfloor\frac{n-q}{2}\rfloor\lceil\frac{n-q}{2}\rceil-\Delta^2$ since $\Delta\ge 0$. Therefore,
\begin{align*}
    \sum_{c=t_2+1}^{n-q}\binom{n-q}{c}(2^{n-c-q}-1)^c &\le\binom{n-q}{\lceil\frac{n-q}{2}\rceil}\sum_{\Delta=t_2+1-\lceil\frac{n-q}{2}\rceil}^{\lfloor\frac{n-q}{2}\rfloor}2^{(\lfloor\frac{n-q}{2}\rfloor-\Delta)(\lceil\frac{n-q}{2}\rceil+\Delta)} \\
    &\le\binom{n-q}{\lceil\frac{n-q}{2}\rceil}2^{\lfloor\frac{n-q}{2}\rfloor\lceil\frac{n-q}{2}\rceil}\sum_{\Delta=t_2+1-\lceil\frac{n-q}{2}\rceil}^{\lfloor\frac{n-q}{2}\rfloor}2^{-\Delta^2} \\
    &= \frac{1}{2}\binom{n-q}{\lceil\frac{n-q}{2}\rceil}2^{\lfloor\frac{n-q}{2}\rfloor\lceil\frac{n-q}{2}\rceil}\cdot 2\sum_{\Delta=t_2+1-\lceil\frac{n-q}{2}\rceil}^{\lfloor\frac{n-q}{2}\rfloor}2^{-\Delta^2}.
\end{align*}
If $t\ge\sqrt{\log\frac{8}{\eps}}$, then $$2\sum_{\Delta=t}^{\infty}2^{-\Delta^2}\le\frac{\eps}{2}.$$ When $t_2$ is defined as above, we have $t_2+1-\lceil\frac{n-q}{2}\rceil\ge\sqrt{\log\frac{8}{\eps}}$, so the value of the upper end of the tail is at most $\frac{\eps}{2}$ times the total summation (since $\frac{\eps}{8}\le 1$).

Similarly, to bound the lower end of the tail, the same argument holds if we let $\Delta = c-\lfloor\frac{n-q}{2}\rfloor$. In this case, we use the fact that $(\lceil\frac{n-q}{2}\rceil-\Delta)(\lfloor\frac{n-q}{2}\rfloor+\Delta)\le\lceil\frac{n-q}{2}\rceil\lfloor\frac{n-q}{2}\rfloor-\Delta^2$ since $\Delta\le 0$. (This is the inequality that we need since we are now using a floor rather than a ceiling in the definition of $\Delta$.) The lower end of the tail is at most $\frac{\eps}{2}$ times the total summation as long as $|t_1-1-\lfloor\frac{n-q}{2}\rfloor|\ge\sqrt{\log\frac{8}{\eps}}$, which is true by the definition of $t_1$. Thus the sum of both tails is at most $\frac{\eps}{2}+\frac{\eps}{2} = \eps$ times the total summation.
\end{proof}

Combining \cref{lemma:sum_over_q,lemma:sum_over_c} with $\frac{\eps}{2}$ in place of $\eps$ and using the fact that $(1-\frac{\eps}{2})^2\ge 1-\eps$, we immediately obtain the following:

\begin{lemma}
\label{lemma:sum_over_both}
Let $0<\eps<1$ be given. For sufficiently large $n$, the following formula gives a $(1-\eps)$-approximation of the number of split graphs on $n$ vertices with $2\le|Q|<n$: $$2\sum_{q=2}^{\min\{s,n-1\}}\sum_{c=\max\{0,t_1\}}^{\min\{t_2,n-q\}}\binom{n}{q}\binom{n-q}{c}(2^{n-c-q}-1)^c,$$ where $s = \lceil 10\log\frac{1}{\eps}+47\rceil$, $t_1 = \lfloor\frac{n-q}{2}\rfloor-\lceil\log\frac{1}{\eps}\rceil-3$, and $t_2 = \lceil\frac{n-q}{2}\rceil+\lceil\log\frac{1}{\eps}\rceil+3$.\qed
\end{lemma}

For the case when $|Q|\le 1$, we can prove similar statements. As above, let $N_1$ be large enough that \cref{lemma:q_is_empty,lemma:q_is_1} hold for $n\ge N_1$.

\begin{lemma}
\label{lemma:q_is_empty_approx}
Let $0<\eps<1$ be given. If $n\ge\max\{N_1,3\log_{3/2}\frac{1}{\eps}\}$, then the following formula gives a $(1\pm\eps)$-approximation of the number of split graphs on $n$ vertices with $Q = \emptyset$: $$\sum_{c=\max\{2,t_1\}}^{\lfloor\frac{n}{2}\rfloor}\binom{n}{c}(2^c-1)^{n-c}+\sum_{c=\lfloor\frac{n}{2}\rfloor+1}^{\min\{t_2,n-2\}}\binom{n}{c}(2^{n-c}-1)^c,$$ where $t_1 = \lfloor\frac{n}{2}\rfloor-\lceil\log\frac{1}{\eps}\rceil-2$ and $t_2 = \lceil\frac{n}{2}\rceil+\lceil\log\frac{1}{\eps}\rceil+2$.
\end{lemma}

\begin{proof}
As we saw in the previous section, when $n\ge\max\{N_1,3\log_{3/2}\frac{1}{\eps}\}$, the summation in \cref{lemma:q_is_empty} is a $(1+\eps)$-approximation of the number of $n$-vertex split graphs with $Q = \emptyset$. Therefore, we just need to show that the summation in this lemma (\cref{lemma:q_is_empty_approx}) is a $(1-\eps)$-approximation of the summation from \cref{lemma:q_is_empty}.

In the following inequality, we use the fact that $\lfloor\frac{n}{2}\rfloor\lceil\frac{n}{2}\rceil\ge\lfloor\frac{n}{2}\rfloor(\lceil\frac{n}{2}\rceil-1)+1$ since $n\ge 2$. The total value of the summation in this lemma is bounded below by the term with $c = \lfloor\frac{n}{2}\rfloor$:
\begin{align*}
    \binom{n}{\lfloor\frac{n}{2}\rfloor}(2^{\lfloor\frac{n}{2}\rfloor}-1)^{\lceil\frac{n}{2}\rceil} &\ge\binom{n}{\lfloor\frac{n}{2}\rfloor}\left(2^{\lfloor\frac{n}{2}\rfloor\lceil\frac{n}{2}\rceil}-2^{\lfloor\frac{n}{2}\rfloor(\lceil\frac{n}{2}\rceil-1)}\right) \\
    &\ge\frac{1}{2}\binom{n}{\lfloor\frac{n}{2}\rfloor}2^{\lfloor\frac{n}{2}\rfloor\lceil\frac{n}{2}\rceil}.
\end{align*}

The upper end of the tail (where $c\ge\lfloor\frac{n}{2}\rfloor+1$) and the lower end of the tail (where $c\le\lfloor\frac{n}{2}\rfloor$) are both at most $\frac{\eps}{2}$ times the total summation by an argument similar to the proof of \cref{lemma:sum_over_c} with $q = 0$. Indeed, it does not make a difference that the lower end of the tail contains the term $(2^c-1)^{n-c}$ rather than $(2^{n-c}-1)^c$ since both of these are bounded above by $2^{c(n-c)}$. Therefore, the sum of both tails is at most $\eps$ times the total summation.
\end{proof}

\begin{lemma}
\label{lemma:q_is_1_approx}
Let $0<\eps<1$ be given. If $n\ge\max\{N_1,3\log_{3/2}\frac{1}{\eps}\}$, then the following formula gives a $(1\pm\eps)$-approximation of the number of split graphs on $n$ vertices with $|Q| = 1$: $$\sum_{c=\max\{2,t_1\}}^{\lfloor\frac{n-1}{2}\rfloor}n\binom{n-1}{c}(2^c-1)^{n-c-1}+\sum_{c=\lfloor\frac{n-1}{2}\rfloor+1}^{\min\{t_2,n-2\}}n\binom{n-1}{c}(2^{n-c-1}-1)^c,$$ where $t_1 = \lfloor\frac{n}{2}\rfloor-\lceil\log\frac{1}{\eps}\rceil-2$ and $t_2 = \lceil\frac{n}{2}\rceil+\lceil\log\frac{1}{\eps}\rceil+2$.
\end{lemma}

\begin{proof}
As in the previous lemma, we just need to show that the summation in this lemma is a $(1-\eps)$-approximation of the summation from \cref{lemma:q_is_1}. For simplicity, we will not write the factor of $n$ in front of each summation, since including that factor does not change the multiplicative error. We know $\lfloor\frac{n-1}{2}\rfloor\lceil\frac{n-1}{2}\rceil\ge\lfloor\frac{n-1}{2}\rfloor(\lceil\frac{n-1}{2}\rceil-1)+1$ since $n\ge 3$. The total value of the summation in this lemma is bounded below by the term with $c = \lfloor\frac{n-1}{2}\rfloor$, which is at most $$\frac{1}{2}\binom{n-1}{\lfloor\frac{n-1}{2}\rfloor}2^{\lfloor\frac{n-1}{2}\rfloor\lceil\frac{n-1}{2}\rceil}$$ by an argument similar to the previous lemma with $n-1$ in place of $n$.

The upper end of the tail (where $c\ge\lfloor\frac{n-1}{2}\rfloor+1$) and the lower end of the tail (where $c\le\lfloor\frac{n-1}{2}\rfloor$) are both at most $\frac{\eps}{2}$ times the total summation by an argument similar to the proof of \cref{lemma:sum_over_c} with $q = 1$. Therefore, the sum of both tails is at most $\eps$ times the total summation.
\end{proof}

Let $N_2$ be large enough that \cref{lemma:sum_over_both} holds for $n\ge N_2$. We define the function $f$ as follows: $f(\eps) = \max\{N_1,N_2,3\log_{3/2}(1/\eps)\}$.



\begin{lemma}
\label{lemma:split_counting}
Given $0<\eps<1$ and $n\ge f(\eps)$, there is an algorithm that computes a $(1\pm\eps)$-approximation of the number of $n$-vertex labeled split graphs in $O(n^3\log{n}\log^2(1/\eps))$ time.
\end{lemma}

\begin{proof}
\medskip
\textbf{Algorithm.} Let $s$ be the total obtained from adding together the summations from Lemmas~\ref{lemma:sum_over_both}, \ref{lemma:q_is_empty_approx}, and \ref{lemma:q_is_1_approx}. Return $s+2$.

\medskip
\textbf{Correctness.} There are exactly two split graphs with $|Q| = n$, so by \cref{lemma:sum_over_both,lemma:q_is_empty_approx,lemma:q_is_1_approx}, $s+2$ is a $(1\pm\eps)$-approximation of the number of $n$-vertex split graphs.

\medskip
\textbf{Running time.} Let $I_q\coloneqq \{2,\ldots,\min\{s,n-1\}\}$ and $I_c\coloneqq \{\max\{0,t_1\},\ldots,\min\{t_2,n-q\}\}$ be the intervals of values of $q$ and $c$ in the summation from \cref{lemma:sum_over_both}. For each $q\in I_q$, $c\in I_c$, we can compute $(2^{n-c-q}-1)^c$ using repeated squaring, which takes $O(\log n)$ arithmetic operations. Each of the binomial coefficients can be computed in $O(n)$ operations using the formula $\binom{a}{b} = \frac{a!}{b!(a-b)!}$. Therefore, we can compute the summation from \cref{lemma:sum_over_both} using $O(n|I_q||I_c|) = O(n\log^2(1/\eps))$ operations. Similarly, the summations from \cref{lemma:q_is_empty_approx,lemma:q_is_1_approx} can be computed using $O(n\log(1/\eps))$ operations. Thus the overall running time is $O(n\log^2(1/\eps))$ arithmetic operations. Since each number in this algorithm can be stored in $O(n^2)$ bits, this amounts to a running time of $O(n^3\log{n}\log^2(1/\eps))$ in the RAM model.
\end{proof}

Interestingly, the factor of $n$ in the running time of \cref{lemma:split_counting} only comes from computing the binomial coefficients. If one were to precompute the binomial coefficients, then the rest of the computation would only take $O(\log n\log^2(1/\eps))$ arithmetic operations.

\subsection{Approximate sampling of split graphs}
\label{sec:split_sampling}

The approximate counting algorithm from the previous section (\cref{lemma:split_counting}) can naturally be extended to give an approximate sampling algorithm that works for $n\ge f(\frac{\eps}{2})$. Recall that $I_q = \{2,\ldots,\min\{s,n-1\}\}$ and $I_c = \{\max\{0,t_1\},\ldots,\min\{t_2,n-q\}\}$. For the sampling algorithm, we consider choosing a random number in $[0,1)$ to be one operation.

\begin{lemma}
\label{lemma:split_sampling}
Given $0<\eps<1$ and $n\ge f(\frac{\eps}{2})$, there is an algorithm that samples a random $n$-vertex labeled split graph according to a distribution whose total variation distance from the uniform distribution is at most $\eps$. The expected running time of this algorithm is $O(n^3\log{n}\log^2(1/\eps))$.
\end{lemma}

\begin{proof}
\medskip
\textbf{Algorithm.} First, we run the counting algorithm from \cref{lemma:split_counting} with input $(n,\eps')$, where $\eps' = \min\{\frac{\eps}{2},\frac{1}{3}\}$. Let $s+2$ be the output of this algorithm.

We randomly choose one of the following four cases --- $Q = \emptyset$, $|Q| = 1$, $2\le|Q|<n$, or $|Q| = n$ --- with probabilities given by the computed summations. For example, the probability with which we choose $Q = \emptyset$ is the value of the summation from \cref{lemma:q_is_empty_approx} divided by $s+2$, and the probability with which we choose $Q = n$ is $\frac{2}{s+2}$.

\underline{$|Q| = n$:}\hspace{0.5em} In this case, we return an $n$-vertex complete graph with probability $\frac{1}{2}$, and otherwise we return an $n$-vertex independent set.

\underline{$2\le|Q|<n$:}\hspace{0.5em} In this case, we choose a random pair $(q',c')\in I_q\times I_c$, where each pair $(q,c)$ has probability proportional to the weight of the $(q,c)$ term in the summation from \cref{lemma:sum_over_both}. We then sample a uniformly random $n$-vertex split graph with $|Q| = q'$ and $|C| = c'$ (by following the proof of \cref{lemma:q_is_2}), and we return this graph.

\underline{$Q = \emptyset$:}\hspace{0.5em} In this case, we choose a random value of $c$, say $c'$, from the values summed over in \cref{lemma:q_is_empty_approx}, where each value of $c$ has probability proportional to the weight of the corresponding term in the summation. We sample a two-tone graph $G$ uniformly at random from those with $c'$ cyan vertices, $n-c'$ indigo vertices, and the following properties: (1) the cyan vertices form a clique; (2) the indigo vertices form an independent set; and (3) if $c'\le\frac{n}{2}$, then every indigo vertex has a non-neighbor among the cyan vertices; otherwise, every cyan vertex has a neighbor among the indigo vertices. We then check whether every cyan vertex in $G$ has an indigo neighbor if $c'\le\frac{n}{2}$, and otherwise we check whether every indigo vertex has a cyan non-neighbor (call this property $P$). If $G$ satisfies property $P$, then we return $G$ (without colors). Otherwise, we go back to the second paragraph of the algorithm (to again choose between the four possible cases for $|Q|$).

\underline{$|Q| = 1$:}\hspace{0.5em} Similarly, in this case we choose a random value of $c$, say $c'$, from the values summed over in \cref{lemma:q_is_1_approx}, where each value of $c$ has probability proportional to the weight of the corresponding term in the summation. We sample a three-tone graph $G$ uniformly at random from those with $c'$ cyan vertices, $n-c'-1$ indigo vertices, one white vertex, and the following properties: (1) the cyan vertices form a clique; (2) the indigo vertices form an independent set; (3) the white vertex is adjacent to all of the cyan vertices and none of the indigo vertices; and (4) if $c'\le\frac{n-1}{2}$, then every indigo vertex has a non-neighbor among the cyan vertices; otherwise, every cyan vertex has a neighbor among the indigo vertices. We then check whether every cyan vertex in $G$ has an indigo neighbor if $c'\le\frac{n-1}{2}$, and otherwise we check whether every indigo vertex has a cyan non-neighbor (call this property $P$). If $G$ satisfies property $P$, then we return $G$ (without colors). Otherwise, we go back to the second paragraph of the algorithm.

\medskip
\textbf{Correctness.} Let $s'$ be the number of split graphs computed by the approximate counting algorithm of \cref{lemma:split_counting} (with the given error bound being $\eps'$), except we do not count the graphs that fail to satisfy property $P$ (so $s'$ is at most $s+2$). By \cref{lemma:split_counting}, $s+2$ is a $(1\pm\frac{\eps}{2})$-approximation of the number of $n$-vertex split graphs. The number of graphs that are counted by this lemma but do not satisfy property $P$ is at most $\frac{\eps}{2}$ times the number of $n$-vertex split graphs, so $s'$ is a $(1-\eps)$-approximation of the number of $n$-vertex split graphs.

The sampling algorithm never outputs a non-split graph or any graph that does not satisfy property $P$. Furthermore, the probabilities of all of the choices of $|Q|$ and $c = |C|$ are designed so that all output graphs are equally likely, since these probabilities are based on the summations from the counting algorithm. Therefore, for every split graph $G$, if the probability that we output $G$ is nonzero, then the probability that we output $G$ is in fact $1/s'$. (We know $s'\ne 0$ since $1-\eps>0$ and the number of $n$-vertex split graphs is nonzero for all $n$.) Since $s'$ is a $(1-\eps)$-approximation of the number of $n$-vertex split graphs, this implies that the total variation distance between the output distribution and the uniform distribution on split graphs is at most $\eps$.

\medskip
\textbf{Running time.} When we say one ``iteration'' of the algorithm, this refers to building a graph and then checking whether property $P$ holds for that graph. First, we show that the expected number of iterations is $O(1)$. If property $P$ does not hold for the current graph, then we call this a ``failure.'' Each time we check property $P$, the probability of failure is $\frac{s+2-s'}{s+2}$. Let $\hat s$ be the exact number of $n$-vertex split graphs. We have $s+2-s'\le\eps'\cdot\hat s\le\eps'(\frac{s+2}{1-\eps'})$ since $(1-\eps')\hat s\le s+2$, so the probability of failure is at most $\frac{\eps'}{1-\eps'}$. Therefore, the expected number of iterations is at most $$\frac{1}{1-\frac{\eps'}{1-\eps'}},$$ which is at most 2, since $\eps'\le\frac{1}{3}$.


Initially, the counting algorithm takes $O(n^3\log{n}\log^2(1/\eps))$ time. Now we just need to bound the running time of each iteration of the algorithm after that. An iteration with $|Q| = n$ takes $O(n^2)$ time. For the case when $2\le|Q|<n$, let $s_2$ be the value of the summation from \cref{lemma:sum_over_both}, and for each pair $(q,c)\in I_q\times I_c$, let $t_{(q,c)}$ denote the $(q,c)$ term in that summation. We can choose a random pair $(q',c')$ in the following way. First, we choose a random number $r\in[0,1)$. Next, we iterate over all pairs $(q,c)\in I_q\times I_c$. For each of these, we subtract $t_{(q,c)}/s_2$ from $r$ and then check whether $r<0$. When $r<0$ for the first time, we let $(q',c')$ be the current pair $(q,c)$. This process takes $O(n\log^2(1/\eps))$ arithmetic operations. At this point, we need to generate a random split graph with $|Q| = q'$ and $|C| = c'$. This can be done in $O(n^2)$ time in the following way. We assign random labels to the sets $C$, $I$, and $Q$ in $O(n)$ time. Next, for each vertex in $C$, we choose its neighborhood by repeatedly choosing a random subset of $I$ until that subset happens to be nonempty. The expected number of attempts for this is at most 2, and each attempt takes $O(n)$ time for a given vertex in $C$. Therefore, for one iteration in the case when $2\le|Q|<n$, the expected running time is $O(n\log^2(1/\eps))$ arithmetic operations plus $O(n^2)$ basic operations.
Similarly, for an iteration with $Q = \emptyset$ or $|Q| = 1$, the expected running time is $O(n\log(1/\eps))$ arithmetic operations plus $O(n^2)$ basic operations. This amounts to a running time of $O(n^3\log{n}\log^2(1/\eps)+n^2) = O(n^3\log{n}\log^2(1/\eps))$ in the RAM model.
\end{proof}



\subsection{Proof of \cref{thm:approx}: Approximate counting and sampling of chordal graphs}
\label{sec:approx_proof}

A random $n$-vertex chordal graph is a split graph with probability $1-o(1)$. More precisely, we have the following result:

\begin{proposition}[Bender et al.~\cite{bender1985almost}]
\label{prop:almost_every}
If $\beta>\sqrt 3/2$, $n$ is sufficiently large, and $G$ is a random $n$-vertex labeled chordal graph, then $$\Pr(\text{$G$ is a split graph})>1-\beta^n.$$
\end{proposition}

Let $N_3$ be large enough that \cref{prop:almost_every} holds for $n\ge N_3$. As above, let $N_1$ and $N_2$ be large enough that \cref{lemma:q_is_empty,lemma:q_is_1} hold for $n\ge N_1$ and \cref{lemma:sum_over_both} holds for $n\ge N_2$. We define the function $g$ as follows: $$g(\eps) = \max\{N_1,N_2,N_3,\log_{10/9}(2/\eps),3\log_{3/2}(2/\eps)\}.$$ We are now ready to describe the approximate counting and sampling algorithms for \cref{thm:approx}.

\medskip
\textbf{Approximate counting algorithm.} We are given $n$ and $\eps$ as input. If $n<g(\eps)$, then let $z$ be the (exact) number of $n$-vertex chordal graphs computed by the counting algorithm from \cref{thm:main}. Otherwise, we run the approximate split graph counting algorithm from \cref{lemma:split_counting} with input $(n,\frac{\eps}{2})$ to obtain a $(1\pm\frac{\eps}{2})$-approximation of the number of $n$-vertex split graphs, and let $z$ be this value. Return $z$.

\medskip
\textbf{Correctness and running time analysis (counting).} To prove correctness, we need to check that we achieve the desired approximation ratio. This is clearly true if $n<g(\eps)$ since in this case we output the exact answer. Now suppose $n\ge g(\eps)$. By \cref{prop:almost_every} with $\beta = \frac{9}{10}$, the (exact) number of $n$-vertex split graphs is a $(1-\frac{\eps}{2})$-approximation of the number $n$-vertex of chordal graphs since $n\ge\max\{N_3,\log_{10/9}(2/\eps)\}$. The approximate counting algorithm from \cref{lemma:split_counting} computes a $(1\pm\frac{\eps}{2})$-approximation of the number of $n$-vertex split graphs since $n\ge\max\{N_1,N_2,3\log_{3/2}(2/\eps)\}$. Therefore, $z$ is a $(1\pm\eps)$-approximation of the number of chordal graphs.

For the running time, if $n<g(\eps)$, then the algorithm uses $O(\log^7(1/\eps))$ arithmetic operations, so in this case the running time is $O(n^2\log n\log^7(1/\eps))$. Otherwise, the running time is $O(n^3\log{n}\log^2(1/\eps))$ by \cref{lemma:split_counting}, so the overall running time is at most $O(n^3\log{n}\log^7(1/\eps))$.

\medskip
\textbf{Approximate sampling algorithm.} We are given $n$ and $\eps$ as input. If $n<g(\frac{\eps}{2})$, then let $G$ be a uniformly random chordal graph generated by the sampling algorithm from \cref{thm:main}. Otherwise, let $G$ be a random split graph generated by the sampling algorithm from \cref{lemma:split_sampling} with input $(n,\frac{\eps}{2})$. Return $G$.

\medskip
\textbf{Correctness and running time analysis (sampling).} We need to show that the total variation distance between the output distribution and the uniform distribution is at most $\eps$. This is clearly true if $n<g(\frac{\eps}{2})$ since in this case the total variation distance is 0. For the case when $n\ge g(\frac{\eps}{2})$, let us first compare the uniform distribution over $n$-vertex split graphs to the uniform distribution over $n$-vertex chordal graphs. The total variation distance between these two distributions is equal to the number of non-split $n$-vertex chordal graphs divided by the total number of $n$-vertex chordal graphs. By \cref{prop:almost_every}, this is at most $\frac{\eps}{2}$ since $n\ge\max\{N_3,\log_{10/9}(2/\eps)\}$. Since $n\ge\max\{N_1,N_2,3\log_{3/2}(4/\eps)\}$, the approximate sampling algorithm from \cref{lemma:split_sampling} samples a random $n$-vertex labeled split graph according to a distribution whose total variation distance from uniform (over split graphs) is at most $\frac{\eps}{2}$. Therefore, by the triangle inequality, the total variation distance between the output distribution and the uniform distribution over $n$-vertex chordal graphs is at most $\eps$.

For the running time, if $n<g(\frac{\eps}{2})$, then the algorithm has a worst-cast running time bound of $O(\log^7(1/\eps))$ arithmetic operations, i.e., $O(n^2\log n\log^7(1/\eps))$ time. Otherwise, the running time is $O(n^3\log{n}\log^2(1/\eps))$ by \cref{lemma:split_sampling}, so the overall running time is at most $O(n^3\log{n}\log^7(1/\eps))$.

\medskip
This completes the proof of \cref{thm:approx}.




\section{Conclusion}

Our main result is an algorithm that given $n$, computes the number of (labeled) chordal graphs on $n$ vertices using $O(n^7)$ arithmetic operations (i.e., in $O(n^9 \log n)$ time in the RAM model). This yields a sampling algorithm that generates a chordal graph on $n$ vertices uniformly at random. For the sampling algorithm, once we have run the counting algorithm as a preprocessing step, each sample can be obtained using $O(n^4)$ arithmetic operations. In addition, we present very efficient approximate counting and approximate sampling algorithms. The former computes a $(1\pm\eps)$-approximation of the number of $n$-vertex labeled chordal graphs, and the latter generates a random labeled chordal graph according to a distribution whose total variation distance from the uniform distribution is at most $\eps$. The approximate counting algorithm runs in $O(n^3\log{n}\log^7(1/\eps))$ time, and the approximate sampling algorithm runs in $O(n^3\log{n}\log^7(1/\eps))$ expected time (both in the RAM model).

A natural open problem is to design a substantially faster algorithm for exact counting or uniform sampling of chordal graphs. Our approximate counting algorithm is quite fast. However, one could argue that our approximate sampling algorithm is unsatisfactory in the sense that it can never output a non-split chordal graph. Therefore, another interesting open problem would be to ask for an approximate sampling algorithm with a stronger guarantee: namely, a sampling algorithm for which $(1-\eps)p\le\Pr(G)\le(1+\eps)p$ for every chordal graph $G$, where $p$ is the probability given by the uniform distribution over all chordal graphs and $\Pr(G)$ is the probability that the approximate sampling algorithm outputs the graph $G$. Two interesting approaches to consider for this are Markov Chain Monte Carlo (MCMC) algorithms and the Boltzmann sampling scheme used in~\cite{Fusy} for planar graphs.

Moving beyond chordal graphs, there are many interesting graph classes for which the problem of counting/sampling $n$-vertex labeled graphs in polynomial time appears to be open, including perfect graphs, weakly chordal graphs, strongly chordal graphs, and chordal bipartite graphs, as well as many graph classes characterized by a finite set of forbidden minors, subgraphs, or induced subgraphs. It is worth noting that for some well-known graph classes of this form, such as planar graphs, polynomial-time algorithms are known \cite{BGK,Fusy}.


\bibliography{references}

\newpage
\appendix

\section{Pseudocode for the sampling algorithm}
\label{sec:pseudocode}

The procedure \Call{Sample\_Chordal}{$n$} returns a uniformly random chordal graph with vertex set $[n]$, and \Call{Sample\_Conn\_Chordal}{$n$} returns a uniformly random \emph{connected} chordal graph with vertex set $[n]$. To sample from $G(t,x,k,z)$, etc., one calls \mbox{\Call{Sample\_$G$}{$t,x,k,z$}}, etc. Recall that $\phi(A,B)$ is the bijection defined in \cref{sec:preliminaries}, which maps the smallest element of $A$ to the smallest element of $B$, and so on.

\begin{algorithm}[H]
  \caption{Chordal graph sampler}
  \begin{algorithmic}[1]
    \Procedure{Sample\_Chordal}{$n$}
    \If{$n = 0$} \Comment{Base case}
        \State Let $G$ be the empty graph with vertex set $\emptyset$ and edge set $\emptyset$
        \State \textbf{return} $G$
    \EndIf
    \State
    \State Choose $k\in[n]$ at random such that $k = k'$ with probability
    \newline\hspace*{3em} $\binom{n-1}{k'-1}c(k')a(n-k')/a(n)$ for each $k'\in[n]$
    \State $G_1\gets$ \Call{Sample\_Conn\_Chordal}{$k$}
    \State $G_2\gets$ \Call{Sample\_Chordal}{$n-k$}
    \State Choose a uniformly random subset $C\subseteq[n]$ of size $k$ containing 1
    \State $D\gets[n]\setminus C$
    \State Relabel $G_1$ by applying $\phi([k],C)$ to its label set
    \State Relabel $G_2$ by applying $\phi([n-k],D)$ to its label set
    \State Let $G$ be the union of $G_1$ and $G_2$
    \State \textbf{return} $G$
    \EndProcedure
  \end{algorithmic}
\end{algorithm}

\begin{algorithm}[H]
  \begin{algorithmic}[1]
    \Procedure{Sample\_Conn\_Chordal}{$n$}
    \State Choose $t\in[n]$ at random such that $t = t'$ with probability
    \newline\hspace*{3em} $\tilde g_1(t',0,n)/c(n)$ for each $t'\in[n]$
    \State $G\gets$ \Call{Sample\_$\widetilde G_1$}{$t,0,n$}
    \State \textbf{return} $G$
    \EndProcedure
  \end{algorithmic}
\end{algorithm}

\begin{algorithm}[H]
  \begin{algorithmic}[1]
    \Procedure{Sample\_$\widetilde G_1$}{$t,x,k$}
    \State Choose $l\in[k]$ at random such that $l = l'$ with probability
    \newline\hspace*{3em} $\binom{k}{l'}f(t,x,l',k-l')/\tilde g_1(t,x,k)$ for each $l'\in[k]$
    \State $G\gets$ \Call{Sample\_$F$}{$t,x,l,k-l$}
    \State Choose a uniformly random subset $L\subseteq[x+1,x+k]$ of size $l$
    \State $A\gets[x+1,x+k]\setminus L$
    \State Relabel $G$ by applying $\phi([x+1,x+l],L)$ to the labels in $[x+1,x+l]$ and applying
    \newline\hspace*{3em} $\phi([x+l+1,x+k],A)$ to the labels in $[x+l+1,x+k]$
    \State \textbf{return} $G$
    \EndProcedure
  \end{algorithmic}
\end{algorithm}

\begin{algorithm}[H]
  \begin{algorithmic}[1]
    \Procedure{Glue}{$G_1,G_2,Y$}
    \State Glue $G_1$ and $G_2$ together at $Y$, and let $G$ be the resulting graph
    \State \textbf{return} $G$
    \EndProcedure
  \end{algorithmic}
\end{algorithm}

\begin{algorithm}[H]
  \begin{algorithmic}[1]
    \Procedure{Sample\_$F$}{$t,x,l,k$}
    \If{$t = 1$ and $k = 0$} \Comment{Base case}
        \State Let $G$ be a complete graph with vertex set $[x+l]$
        \State \textbf{return} $G$
    \EndIf
    \State
    \State Choose $k'\in[k]$ at random such that $k' = k''$ with probability
    \newline\hspace*{3em} $\binom{k}{k''}\tilde f(t,x,l,k'')g(t-2,x+l,k-k'',x)/f(t,x,l,k)$ for each $k''\in[k]$
    \State $G_1\gets$ \Call{Sample\_$\widetilde F$}{$t,x,l,k'$}
    \State $G_2\gets$ \Call{Sample\_$G$}{$t-2,x+l,k-k',x$}
    \State Choose a uniformly random subset $A\subseteq[x+l+1,x+l+k]$ of size $k'$
    \State $B\gets[x+l+1,x+l+k]\setminus A$
    \State Relabel $G_1$ by applying $\phi([x+l+1,x+l+k'],A)$ to the labels in $[x+l+1,x+l+k']$
    \State Relabel $G_2$ by applying $\phi([x+l+1,x+l+k-k'],B)$ to the labels in $[x+l+1,x+l+k-k']$
    \State $G\gets$ \Call{Glue}{$G_1,G_2,[x+l]$}
    \State \textbf{return} $G$
    \EndProcedure
  \end{algorithmic}
\end{algorithm}

\begin{algorithm}[H]
  \begin{algorithmic}[1]
    \Procedure{Sample\_$G$}{$t,x,k,z$}
    \If{$t = 0$ and $k = 0$} \Comment{Base case}
        \State Let $G$ be a complete graph with vertex set $[x]$
        \State \textbf{return} $G$
    \EndIf
    \State
    \State Choose $k'\in\{0,1,\ldots,k\}$ at random such that $k' = k''$ with probability
    \newline\hspace*{3em} $\binom{k}{k''}\tilde g(t,x,k'',z)g(t-1,x,k-k'',z)/g(t,x,k,z)$ for each $0\le k''\le k$
    \State $G_1\gets$ \Call{Sample\_$\widetilde G$}{$t,x,k',z$}
    \State $G_2\gets$ \Call{Sample\_$G$}{$t-1,x,k-k',z$}
    \State Choose a uniformly random subset $A\subseteq[x+1,x+k]$ of size $k'$
    \State $B\gets[x+1,x+k]\setminus A$
    \State Relabel $G_1$ by applying $\phi([x+1,x+k'],A)$ to the labels in $[x+1,x+k']$
    \State Relabel $G_2$ by applying $\phi([x+1,x+k-k'],B)$ to the labels in $[x+1,x+k-k']$
    \State $G\gets$ \Call{Glue}{$G_1,G_2,[x]$}
    \State \textbf{return} $G$
    \EndProcedure
  \end{algorithmic}
\end{algorithm}

\begin{algorithm}[H]
  \begin{algorithmic}[1]
    \Procedure{Sample\_$\widetilde G$}{$t,x,k,z$}
    \If{$k = 0$} \Comment{Base case}
        \State Let $G$ be a complete graph with vertex set $[x]$
        \State \textbf{return} $G$
    \EndIf
    \State
    \State Choose $k'\in[k]$, $x'\in[x]$ at random such that $(k',x') = (k'',x'')$ with probability
    \newline\hspace*{3em} $$\left(\binom{x}{x''}-\binom{z}{x''}\right)\binom{k-1}{k''-1}\tilde g_1(t,x'',k'')\tilde g(t,x,k-k'',z)/\tilde g(t,x,k,z)$$
    \newline\hspace*{3em} for each pair $(k'',x'')\in[k]\times[x]$
    \State $G_1\gets$ \Call{Sample\_$\widetilde G_1$}{$t,x',k'$}
    \State $G_2\gets$ \Call{Sample\_$\widetilde G$}{$t,x,k-k',z$}
    \State Choose a uniformly random subset $X'\subseteq[x]$ of size $x'$ such that $X'\not\subseteq[z]$
    \State Choose a uniformly random subset $C\subseteq[x+1,x+k]$ of size $k'$ containing $x+1$
    \State $D\gets[x+1,x+k]\setminus C$
    \State Relabel $G_1$ by applying $\phi([x'],X')$ to the labels in $[x']$ and applying
    \newline\hspace*{3em} $\phi([x'+1,x'+k'],C)$ to the labels in $[x'+1,x'+k']$
    \State Relabel $G_2$ by applying $\phi([x+1,x+k-k'],D)$ to the labels in $[x+1,x+k-k']$
    \State $G\gets$ \Call{Glue}{$G_1,G_2,X'$}
    \State \textbf{return} $G$
    \EndProcedure
  \end{algorithmic}
\end{algorithm}

\begin{algorithm}[H]
  \begin{algorithmic}[1]
    \Procedure{Sample\_$\widetilde F$}{$t,x,l,k$}
    \State $S_1\gets\tilde f_p(t,x,l,k)$
    \State $S_2\gets\sum_{k'=1}^k\binom{k}{k'}\tilde g_1(t-1,x+l,k')\tilde f_p(t,x,l,k-k')$
    \State $S_3\gets\sum_{k'=1}^k\binom{k}{k'}\tilde g_{\geq 2}(t-1,x+l,k')\tilde g_p(t-1,x+l,k-k',x)$
    \State Choose $i\in[\tilde f(t,x,l,k)]$ uniformly at random
    \If{$i\le S_1$}
        \State $G\gets$ \Call{Sample\_$\widetilde F_p$}{$t,x,l,k$}
        \State \textbf{return} $G$
    \ElsIf{$i\le S_1+S_2$}
        \State Choose $k'\in[k]$ at random such that $k' = k''$ with probability
        \newline\hspace*{5em} $\binom{k}{k''}\tilde g_1(t-1,x+l,k'')\tilde f_p(t,x,l,k-k'')/S_2$ for each $k''\in[k]$
        \State $G_1\gets$ \Call{Sample\_$\widetilde G_1$}{$t-1,x+l,k'$}
        \State $G_2\gets$ \Call{Sample\_$\widetilde F_p$}{$t,x,l,k-k'$}
    \Else
        \State Choose $k'\in[k]$ at random such that $k' = k''$ with probability
        \newline\hspace*{5em} $\binom{k}{k''}\tilde g_{\geq 2}(t-1,x+l,k'')\tilde g_p(t-1,x+l,k-k'',x)/S_3$ for each $k''\in[k]$
        \State $G_1\gets$ \Call{Sample\_$\widetilde G_{\ge 2}$}{$t-1,x+l,k'$}
        \State $G_2\gets$ \Call{Sample\_$\widetilde G_p$}{$t-1,x+l,k-k',x$}
    \EndIf
    \State Choose a uniformly random subset $A\subseteq[x+l+1,x+l+k]$ of size $k'$
    \State $B\gets[x+l+1,x+l+k]\setminus A$
    \State Relabel $G_1$ by applying $\phi([x+l+1,x+l+k'],A)$ to the labels in $[x+l+1,x+l+k']$
    \State Relabel $G_2$ by applying $\phi([x+l+1,x+l+k-k'],B)$ to the labels in $[x+l+1,x+l+k-k']$
    \State $G\gets$ \Call{Glue}{$G_1,G_2,[x+l]$}
    \State \textbf{return} $G$
    \EndProcedure
  \end{algorithmic}
\end{algorithm}

\begin{algorithm}[H]
  \begin{algorithmic}[1]
    \Procedure{Sample\_$\widetilde G_{\ge 2}$}{$t,x,k$}
    \State $S_1\gets\sum_{k'=1}^{k-1}\binom{k-1}{k'-1}\tilde g_1(t,x,k')\tilde g_1(t,x,k-k')$
    \State $S_2\gets\sum_{k'=1}^{k-1}\binom{k-1}{k'-1}\tilde g_1(t,x,k')\tilde g_{\geq 2}(t,x,k-k')$
    \State Choose $i\in[\tilde g_{\ge 2}(t,x,k)]$ uniformly at random
    \If{$i\le S_1$}
        \State Choose $k'\in[k]$ at random such that $k = k''$ with probability
        \newline\hspace*{5em} $\binom{k-1}{k'-1}\tilde g_1(t,x,k'')\tilde g_1(t,x,k-k'')/S_1$ for each $k''\in[k]$
        \State $G_1\gets$ \Call{Sample\_$\widetilde G_1$}{$t,x,k'$}
        \State $G_2\gets$ \Call{Sample\_$\widetilde G_1$}{$t,x,k-k'$}
    \Else
        \State Choose $k'\in[k]$ at random such that $k = k''$ with probability
        \newline\hspace*{5em} $\binom{k-1}{k'-1}\tilde g_1(t,x,k')\tilde g_{\geq 2}(t,x,k-k')/S_2$ for each $k''\in[k]$
        \State $G_1\gets$ \Call{Sample\_$\widetilde G_1$}{$t,x,k'$}
        \State $G_2\gets$ \Call{Sample\_$\widetilde G_{\ge 2}$}{$t,x,k-k'$}
    \EndIf
    \State Choose a uniformly random subset $C\subseteq [x+1,x+k]$ of size $k'$ containing $x+1$
    \State $D\gets[x+1,x+k]\setminus C$
    \State Relabel $G_1$ by applying $\phi([x+1,x+k'],C)$ to the labels in $[x+1,x+k']$
    \State Relabel $G_2$ by applying $\phi([x+1,x+k-k'],D)$ to the labels in $[x+1,x+k-k']$
    \State $G\gets$ \Call{Glue}{$G_1,G_2,[x]$}
    \State \textbf{return} $G$
    \EndProcedure
  \end{algorithmic}
\end{algorithm}

\begin{algorithm}[H]
  \begin{algorithmic}[1]
    \Procedure{Sample\_$\widetilde G_p$}{$t,x,k,z$}
    \If{$k = 0$} \Comment{Base case}
        \State Let $G$ be a complete graph with vertex set $[x]$
        \State \textbf{return} $G$
    \EndIf
    \State
    \State Choose $k'\in[k]$, $x'\in[x-1]$ at random such that $(k',x') = (k'',x'')$ with probability
    \newline\hspace*{3em} $$\left(\binom{x}{x''}-\binom{z}{x''}\right)\binom{k-1}{k''-1}\tilde g_1(t,x'',k'')\tilde g_p(t,x,k-k'',z)/\tilde g_p(t,x,k,z)$$
    \newline\hspace*{3em} for each pair $(k'',x'')\in[k]\times[x-1]$
    \State $G_1\gets$ \Call{Sample\_$\widetilde G_1$}{$t,x',k'$}
    \State $G_2\gets$ \Call{Sample\_$\widetilde G_p$}{$t,x,k-k',z$}
    \State Choose a uniformly random subset $X'\subseteq[x]$ of size $x'$ such that $X'\not\subseteq[z]$
    \State Choose a uniformly random subset $C\subseteq[x+1,x+k]$ of size $k'$ containing $x+1$
    \State $D\gets[x+1,x+k]\setminus C$
    \State Relabel $G_1$ by applying $\phi([x'],X')$ to the labels in $[x']$ and applying $\phi([x'+1,x'+k'],C)$
    \newline\hspace*{3em} to the labels in $[x'+1,x'+k']$
    \State Relabel $G_2$ by applying $\phi([x+1,x+k-k'],D)$ to the labels in $[x+1,x+k-k']$
    \State $G\gets$ \Call{Glue}{$G_1,G_2,X'$}
    \State \textbf{return} $G$
    \EndProcedure
  \end{algorithmic}
\end{algorithm}

\begin{algorithm}[H]
  \begin{algorithmic}[1]
    \Procedure{Sample\_$\widetilde F_p$}{$t,x,l,k$} \Comment{Four arguments, without $z$}
    \State \textbf{return} \Call{Sample\_$\widetilde F_p$}{$t,x,l,k,x$}
    \EndProcedure
  \end{algorithmic}
\end{algorithm}

\begin{algorithm}[H]
  \begin{algorithmic}[1]
    \Procedure{Sample\_$\widetilde F_p\_With\_Z$}{$t,x,l,k,z$} \Comment{Five arguments, including $z$}
    \State Choose $k'\in[k]$, $x'\in\{0,1,\ldots,x\}$, $l'\in\{0,1,\ldots,l\}$ at random such that
    \newline\hspace*{3em} $0<x'+l'<x+l$ and such that $(k',x',l') = (k'',x'',l'')$ with probability
    \begin{align*}
    &\binom{k-1}{k''-1}\binom{l}{l''}\tilde g_1(t-1,x''+l'',k'') \\
    &\cdot
    \begin{cases}
    \binom{x}{x''} & \text{ if $l''>0$}  \\
    \binom{x}{x''}-\binom{z}{x''} & \text{ otherwise}
    \end{cases}
    \,\,\cdot
    \begin{cases}
    \tilde f_p(t,x+l'',l-l'',k-k'',z) & \text{ if $l''<l$}  \\
    \tilde g_p(t-1,x+l'',k-k'',z) & \text{ otherwise}
    \end{cases} \\
    &\cdot 1/\tilde f_p(t,x,l,k,z)
    \end{align*}
    \hspace*{3em} for each possible triple $(k'',x'',l'')$
    \State $G_1\gets$ \Call{Sample\_$\widetilde G_1$}{$t-1,x'+l',k'$}
    \State $G_2\gets$ $l'<l$ ? \Call{Sample\_$\widetilde F_p$}{$t,x+l',l-l',k-k',z$} : \Call{Sample\_$\widetilde G_p$}{$t-1,x+l',k-k',z$}
    \If{$l'>0$}
        \State Choose a uniformly random subset $X'\subseteq X$ of size $x'$ containing $x+l+1$
    \Else
        \State Choose a uniformly random subset $X'\subseteq X$ of size $x'$ containing $x+l+1$
        \newline\hspace*{5em} such that $X'\not\subseteq[z]$
    \EndIf
    \State Choose a uniformly random subset $L'\subseteq[x+1,x+l]$ of size $l'$
    \State Choose a uniformly random subset $C\subseteq[x+l+1,x+l+k]$ of size $k'$ containing $x+l+1$
    \State $D\gets[x+l+1,x+l+k]\setminus C$
    \State Relabel $G_1$ by applying $\phi([x'],X')$ to the labels in $[x']$, applying $\phi([x'+1,x'+l'],L')$
    \newline\hspace*{3em} to the labels in $[x'+1,x'+l']$, and applying $\phi([x'+l'+1,x'+l'+k'],C)$ to the
    \newline\hspace*{3em} labels in $[x'+l'+1,x'+l'+k']$
    \State Relabel $G_2$ by applying $\phi([x+l+1,x+l+k-k'],D)$ to the labels in $[x+l+1,x+l+k-k']$
    \If{$l'<l$}
        \State Relabel $G_2$ by applying $\phi([x+1,x+l'], L')$ to the labels in $[x+1,x+l']$ and applying
        \newline\hspace*{5em}$ \phi([x+l'+1,x+l],[x+1,\ldots,x+l]\setminus L')$ to the labels in $[x+l'+1,x+l]$
    \EndIf
    \State $G\gets$ \Call{Glue}{$G_1,G_2,X'\cup L'$}
    \State \textbf{return} $G$
    \EndProcedure
  \end{algorithmic}
\end{algorithm}

\end{document}